\documentclass[11pt,a4paper]{article}
\usepackage[a4paper,margin=1.02in,includefoot]{geometry}
\usepackage[T1]{fontenc}
\usepackage{lmodern}
\usepackage{iftex}
\ifPDFTeX
  \input{glyphtounicode}
  \IfFileExists{glyphtounicode-cmex.tex}{\input{glyphtounicode-cmex}}{}
\fi
\usepackage{amsmath,amssymb,amsfonts,amsthm,mathtools}
\usepackage[title]{appendix}
\usepackage{booktabs}
\usepackage{placeins}
\usepackage[hidelinks]{hyperref}
\AddToHook{cmd/section/before}{\ifnum\value{section}=8\relax\FloatBarrier\fi}
\AddToHook{env/thebibliography/begin}{\interlinepenalty=10000 \mathchardef\UrlBigBreakPenalty=10000\relax}

\AddToHook{env/appendices/begin}{\numberwithin{equation}{section}}
\DeclarePairedDelimiter\norm{\lVert}{\rVert}

\newcommand{\cE}{\mathcal{E}}
\newcommand{\cH}{\mathcal{H}}
\newcommand{\cP}{\mathcal{P}}
\newcommand{\cS}{\mathcal{S}}
\newcommand{\E}{\mathbb{E}}
\newcommand{\Z}{\mathbb{Z}}
\newcommand{\eqdef}{\triangleq}

\DeclareMathOperator{\supp}{supp}

\theoremstyle{plain}
\newtheorem{theorem}{Theorem}
\newtheorem{proposition}[theorem]{Proposition}
\newtheorem{corollary}[theorem]{Corollary}
\newtheorem{lemma}[theorem]{Lemma}

\theoremstyle{remark}
\newtheorem{example}{Example}
\newtheorem{remark}{Remark}

\theoremstyle{definition}
\newtheorem{definition}{Definition}

\hypersetup{
  bookmarksdepth=subsubsection,
  bookmarksnumbered=true,
  pdfdisplaydoctitle=true,
  pdftitle={Bosonic Loss Codes with Bounded Occupation},
  pdfauthor={Avraham Kreindel; Aryeh Lev Zabokritskiy (Yohananov)},
  pdfsubject={Quantum error correction under bounded bosonic occupation},
  pdflang={en-US},
  pdfkeywords={Quantum error correction, bosonic loss, Fock-state codes,
  multiset codes, discrete simplex, Tverberg theorem}
}
\allowdisplaybreaks

\title{Bosonic Loss Codes with Bounded Occupation}
\author{Avraham~Kreindel\textsuperscript{1}\quad
  Aryeh~Lev Zabokritskiy (Yohananov)\textsuperscript{2,3}\\[0.6em]
  \small\textsuperscript{1}Department of Computer Science, Reichman University,\\
  \small Herzliya, Israel\\[0.25em]
  \small\textsuperscript{2}Department of Computer Science,\\
  \small Tel-Hai University of Kiryat Shmona in the Galilee,\\
  \small Kiryat Shmona, Israel\\[0.25em]
  \small\textsuperscript{3}MIGAL -- Galilee Research Institute, Kiryat Shmona, Israel\\[0.5em]
  \small\href{mailto:avrahamkreindel@gmail.com}{\texttt{avrahamkreindel@gmail.com}}\\
  \small Corresponding author:
  \href{mailto:yuhanalev@telhai.ac.il}{\texttt{yuhanalev@telhai.ac.il}}
}
\date{}

\begin{document}
\maketitle

\begin{abstract}
We study Fock-state codes on \(q\) bosonic modes, with total excitation
\(N\) and at most \(B\) excitations per mode, that exactly correct every
pattern of at most \(\delta N\) photon losses, for \(q\) proportional to
\(N\). For every fixed cap \(B\) and every ratio \(q/N>1/B\), we prove
explicit achievable rates, positive for all sufficiently small
\(\delta>0\), in logical qubits per excitation. At \(B=2\) and \(q=N\),
the rate is positive for every \(\delta<0.1355\) and is about \(0.281\)
at \(\delta=0.1\), compared with the previously established exact-code
qutrit rate of about \(0.275\). We also give an explicit upper bound on
the dimension of every such exact code, including degenerate codes.
For the pure-loss channel truncated to at most \(\delta N\) losses and
renormalized, we obtain approximate codes with exponentially small
worst-case error. When the loss probability is positive and at most \(1/(B+1)\),
they improve our positive exact-code rate bounds; at \(B=2\), \(q=N\),
and loss probability \(0.1\), they remain positive up to
\(\delta\approx0.158\).
\end{abstract}



\section{Introduction}
\label{sec:intro}

Fock-state codes store quantum information in number states of \(q\)
bosonic modes with a fixed total excitation \(N\) and protect it against
photon loss; amplitude-damping codes with constant excitation or low
occupation were studied in
\cite{ChuangLeungYamamoto1997,WasilewskiBanaszek2007,
GrasslKongWeiYinZeng2018,OuyangChao2020}.
Aydin, Albert, and Barg (AAB) established dimension guarantees for exact
Fock-state codes supported on classical sets of occupation vectors that
remain distinguishable under the allowed losses
\cite{AydinAlbertBarg2026}; related classical-to-quantum constructions appear in
\cite{Bumgardner2011,CaoKribsLiNelsonPoonZeng2021,MovassaghOuyang2024}.
Elimelech, Aydin, and Barg then constructed asymptotically good Fock-state
codes, which correct a number of losses proportional to \(N\) at positive
rate; in their exact codes each mode holds \(O(\log N)\) excitations
\cite{ElimelechAydinBarg2026}, and they ask for a sampling
distribution of classical codewords with linear distance, positive rate, and occupation bounded
by a constant \cite[Sec.~6.2]{ElimelechAydinBarg2026}.  We study codes in
which every mode holds at most a constant number \(B\) of excitations.

The number states of this constant-excitation sector are indexed by the
discrete simplex
\[
  \cS_{q,N}
  \eqdef
  \left\{
    \underline{n}\in\Z_{\geq0}^{q}:
    \sum_{i=0}^{q-1}n_i=N
  \right\},
\]
and the cap \(B\) restricts the code to the Fock span of the capped simplex
\(\cS_{q,N}^{(B)}\) of vectors in \(\cS_{q,N}\) with every \(n_i\leq B\).
A loss pattern \(\underline{e}\leq\underline{n}\), of total order
\(|\underline{e}|=\sum_ie_i\), removes \(e_i\) excitations from mode \(i\).  Under pure loss with a common loss
probability \(\gamma\), its Kraus operator maps
\(\lvert\underline{n}\rangle\) to a multiple of
\(\lvert\underline{n}-\underline{e}\rangle\), and exact correction of all
loss patterns of total order at most \(t\) does not depend on \(\gamma\).
Two number states can reach a common state after at most \(t\) losses
precisely when their deletion distance
\(\frac12\norm{\underline{n}-\underline{m}}_1\) is at most \(t\).  A
separated support, that is, a set of occupation vectors that are pairwise
more than \(t\) apart, is therefore a multiset code correcting \(t\)
deletions \cite{AydinAlbertBarg2026,KovacevicTan2018,
KreindelEssayagZabokritskiy2026a,KreindelEssayagZabokritskiy2026b}.
We take \(q=\lfloor\alpha N\rfloor\) modes and \(t=\lfloor\delta N\rfloor\)
losses with fixed \(\alpha\) and \(\delta\), and assume \(\alpha B>1\): for
\(\alpha B<1\) the capped simplex is empty, and for \(\alpha B=1\) it has at
most one point.  Our question is how large the dimension exponent
\(N^{-1}\log_2K\) of an exact code, its number of logical qubits per
excitation, can be as a function of \(B\), \(\alpha\), and \(\delta\), and
how close the achievable exponents come to upper bounds valid for every code
in the Fock span of the capped simplex.

Positive rate and a linear number of correctable losses are already
possible under a fixed cap. The dual-rail codes of Duan, Grassl, Ji, and
Zeng \cite[Sec.~III, Lem.~1 and Thm.~1]{DuanGrasslJiZeng2010} have cap one
and \(q=2N\); see also \cite{Ouyang2021}. Their padded variants cover
every \(\alpha B>1\). Our main result gives explicit achievable dimension
exponents \(k_B(\alpha,\delta)\) for every fixed \(B\) and
\(\alpha B>1\), positive for all sufficiently small \(\delta>0\).
The contribution is quantitative, rather than the existence of positive
rate under a fixed cap. The exponent is specified by a finite-dimensional
optimization over occupation histograms, which record the number of
modes at each occupation. The classical bounds are fixed-composition
forms of the generalized Gilbert--Varshamov bound
\cite{GuFuja1993,Tolhuizen1997,SomekhBaruchScarlettGuillen2019}.

At \(B=2\) and \(q=N\), our exponent has a closed form, positive for
\(0<\delta<\delta_2^\ast\approx0.1355\) and about \(0.281\) at
\(\delta=0.1\). For qudit amplitude damping without fixed total excitation,
Elimelech, Albert, and Barg give an exact-code rate
\cite[Sec.~5.2.2 and Eq.~(82)]{ElimelechAlbertBargAQEC2026}; in logical
qubits per mode, its qutrit case is about \(0.275\) at \(\delta=0.1\)
and positive up to \(\delta\approx0.1342\). For our balanced sector
\(q=N\), rates per mode and per excitation coincide, as do the loss
fractions \(t/q=t/N\); this comparison does not impose fixed excitation
on the qutrit benchmark.
Our bound for the whole capped simplex reproduces that rate, whereas the
bound for an optimized histogram layer gives \(0.281\). The corresponding
classical exponent coincides with Tolhuizen's ternary constant-weight
exponent for Manhattan distance \cite{TolhuizenTernary2010}.
For comparison, our achievable exponents for padded dual-rail codes and
for concatenations of the qutrit inner code of
\cite[arXiv version, Sec.~VIII]{GrasslKongWeiYinZeng2018} are about
\(0.098\) and \(0.171\), respectively, at \(\delta=0.1\).
Our cap-two exponent exceeds both benchmarks for \(\delta\leq0.130\),
although they remain positive up to approximately \(0.139\) and \(0.140\),
slightly beyond \(\delta_2^\ast\). The comparison also depends on the cap:
at cap one and \(q=2N\), our closed-form exponent is positive for
\(\delta<0.151\), while dual-rail codes with quantum Gilbert--Varshamov
outer codes \cite{CalderbankRainsShorSloane1997,FengMa2004} remain positive
up to approximately \(0.189\) and give better rates for
\(\delta\geq0.1212\).

The finite-dimensional guarantee specializes that of
\cite{AydinAlbertBarg2026} to the cap. A classical support
\(C\subseteq\cS_{q,N}^{(B)}\) of minimum deletion distance at least
\(t+1\) supports an exact \(K\)-dimensional code whenever
\[
  |C|\geq(K-1)m_B(q,t)+1,
  \qquad m_B(q,t)\eqdef[z^t](1+z+\cdots+z^B)^q.
\]
Here \(m_B(q,t)\) counts order-\(t\) loss patterns with at most \(B\)
losses per mode, in place of the unrestricted count
\(\binom{q+t-1}{q-1}\) in \cite{AydinAlbertBarg2026}.
Cap restrictions and redundant error conditions were already treated in
\cite[Remark~5 and Sec.~5.2.2]{ElimelechAlbertBargAQEC2026}.
The two counts agree for \(t\leq B\); at \(B=2\), \(q=N\), and
\(\delta=0.1\), the cap-dependent count alone improves the exponent only
from about \(0.280\) to \(0.281\).

Finite-length questions also require the distance distribution and ball
volumes of the capped simplex. We give exact coefficient formulas for
these quantities and the resulting guaranteed quantum dimension, as cap
truncations of the uncapped enumerators in
\cite{KreindelEssayagZabokritskiy2026b}. In the tabulated cases, the exact
average ball size guarantees several times the dimension obtained from
the maximum ball. We also extend the extremal-ball results of that paper:
the deletion-ball volume is Schur-concave in its center, with the smallest
balls at the saturated point, whose coordinates are as concentrated at
\(B\) as possible, and the largest at the balanced point, whose coordinates
are as equal as possible. Exact generating functions describe both.
A separate loss-overlap enumerator weights pairs of number states by sums
of squared Kraus overlaps determined by their joint type. For fixed
\(B\), it admits a polynomial-size recurrence; it is not used in the
existence bounds.

These ball-volume results give sphere-packing bounds for separated
classical supports, not for arbitrary quantum codes. At \(B=2\) and
\(q=N\), they limit the exponents obtainable from a separated support
through the count \(m_2\) to about \(0.620\) at \(\delta=0.1\), and to
positive values only for \(\delta<0.2081\). For arbitrary exact codes,
including degenerate codes, we prove instead
\[
  K\,h_B(N,r)\leq[z^{N-r}](1+z+\cdots+z^B)^q
  \qquad(0\leq r\leq t),
\]
where \(h_B(N,r)\) counts order-\(r\) loss patterns acting nontrivially on
the saturated point. The underlying derivative-dimension and rank-sensitive
quantum packing bounds are due to
\cite{GarciaMarcoKoiranPecatteThomasse2017} and
\cite{ChiribellaEtAl2011}, respectively. A related Hamming-type bound for
qudit amplitude damping without fixed excitation appears in
\cite{DuttaJainMandayam2026}.
Our converse and the coherent-information bounds derived here are
incomparable: at \((B,\alpha,\delta)=(2,1,0.1)\), our upper exponent is
about \(1.188\), below the cap-free coherent-information bound \(1.413\)
but far above the achieved \(0.281\); it remains positive for
\(\delta<1\), whereas the latter bound forces zero rate for
\(\delta\geq1/2\). Channel-adapted enumerators and linear-programming
bounds \cite{OuyangLai2020,OuyangLai2021}, and intrinsic MacWilliams
transforms for symmetric-power codes
\cite{KubischtaTeixeira2026,Teixeira2026}, provide further related results;
we do not compare them numerically with our converse.

The sampling question of \cite[Sec.~6.2]{ElimelechAydinBarg2026} asks for
classical codewords with positive rate, linear distance, and uniformly
bounded occupation. Exact uniform sampling from the capped simplex and
from each histogram layer has polynomial expected bit complexity per
sample, within the general sampling framework of
\cite{FlajoletZimmermannVanCutsem1994}. At \(B=2\) and \(q=N\), the
random-support guarantees of \cite{ElimelechAlbertBargAQEC2026}
specialize to uniform sampling from optimized histogram layers: after
expurgation, the retained separated supports attain every rate below
the optimized classical exponent for \(0<\delta<1/2\).
This layer distribution has the requested bounded-occupation property.

For approximate correction, Elimelech, Albert, and Barg treat channel
families controlled by an error set and an error-intensity parameter
\cite{ElimelechAlbertBargAQEC2026}, with universal recovery established in
\cite{ElimelechAlbertBargUniversal2026}. Our result concerns one pure-loss
channel with loss probability \(\gamma\), truncated to at most \(t\)
losses and renormalized, on a fixed separated support of size \(L\).
For a uniformly random partition of this support into \(K\) equal blocks,
the uniform superpositions over the blocks define an equal-weight
partition code. With high probability its worst-case error is polynomially
small in \(L\), hence exponentially small in \(N\) for our supports, provided
\(KM\kappa\leq L^{1-\eta}\) for a fixed \(\eta>0\).
Here the error is measured by the entanglement-fidelity criterion of
\cite{ElimelechAydinBarg2026}; \(M\) counts loss patterns acting nontrivially
on the support, and the error intensity \(\kappa\) is the largest
probability of a single loss pattern at a support point
\cite{ElimelechAlbertBargAQEC2026}. The scale \(KM\kappa\ll L\) is known
for Haar-random subcodes in average-case fidelity
\cite[Eq.~(30)]{Klesse2007}, and for Pauli errors in worst-case error
\cite{MaTanWright2025}; here it holds for equal-weight partition codes
under the worst-case criterion. Adapted to the same fixed support and
evaluated with the same error-intensity bound, the bound in
\cite{ElimelechAlbertBargAQEC2026} guarantees smaller dimension exponents
than ours for every \(\gamma\), but covers the broader channel-family
setting just described.

For \(0<\gamma\leq1/(B+1)\), we give an explicit exponent for
\(\kappa\) and obtain approximate-code rates exceeding \(k_B\)
wherever \(k_B(\alpha,\delta)>0\).
At \(B=2\), \(q=N\), and \(\gamma=0.1\), they remain positive up to
\(\delta\approx0.158\), compared with \(\delta_2^\ast\approx0.1355\)
for our exact codes. For larger \(\gamma\), the estimate \(\kappa\leq1\)
used here gives no such improvement. When \(\delta>\gamma\), these codes
also have vanishing worst-case error on the untruncated channel
\cite[Thm.~5]{ElimelechAydinBarg2026}.

Finally, caps that tend to infinity provide a comparison with the fixed-cap
regime. The full-simplex average-ball exponent and its Marcus--Roth
refinement, which fixes the number of empty modes, were studied in
\cite{GoyalDaoKovacevicKiah2025}; the corresponding full-simplex quantum
rate is due to \cite{ElimelechAlbertBargAQEC2026}, which does not give
the Marcus--Roth quantum rate. We show that both exponents survive every
hard occupation cap tending to infinity with \(N\). At \(q=N\), the
resulting exact codes have positive rate up to approximately \(0.188\)
and \(0.192\), respectively, compared with \(0.177\) for the maximum-ball
codes of \cite{ElimelechAydinBarg2026}. The occupation bound
\(O(\log N)\) in \cite{ElimelechAlbertBargAQEC2026} applies to random
approximate codes after a vanishing part of their weight is discarded;
our hard caps bound every occupation of exact codes.

Two-mode loss codes and covariant bosonic encodings
\cite{ManesSubramanianJiang2026,StAmandBurelleRoyer2026} concern a different
regime from the growing number of modes and fixed cap considered here.
Our constructions are existential: efficient constructions of the
optimized classical supports, exact encoders, and recovery maps are not
supplied.

Section~\ref{sec:preliminaries} fixes the loss model and the AAB
construction, and Section~\ref{sec:restriction-principles} restricts the
average-ball bound to subsets and specializes the construction to the cap.
Section~\ref{sec:cap-two} proves the cap-two exponent, and
Section~\ref{sec:general-caps} treats general fixed caps, the cap-one boundary, and finite
enumerators.  Section~\ref{sec:converses} gives classical sphere-packing
and quantum converse bounds, Section~\ref{sec:random} treats sampling and
approximate correction for one normalized channel, and Section~\ref{sec:log-caps} treats caps
that grow with \(N\).  Section~\ref{sec:conclusion} lists open problems.

\section{Codes in the discrete simplex and quantum lifting}
\label{sec:preliminaries}

This section fixes the classical model, the bosonic loss channel, and the
Aydin--Albert--Barg construction, here called lifting, in the affine-hull
form used throughout.  As in \cite{AydinAlbertBarg2026,ElimelechAydinBarg2026}, the
number \(q\) of modes comes first and the total excitation \(N\) second in
\(\cS_{q,N}\).  Occupation and loss vectors are underlined, with mode
indices \(0,\ldots,q-1\); unqualified coordinate sums and products use
these indices.

\begin{definition}[Discrete simplex and deletion distance]
\label{def:simplex}
For integers \(N\geq0\) and \(q\geq1\), let
\[
  \cS_{q,N}
  \eqdef
  \left\{
    \underline{n}=(n_0,\ldots,n_{q-1})\in\Z_{\geq 0}^{q}:
    \sum_{i=0}^{q-1}n_i=N
  \right\}.
\]
For \(\underline{n},\underline{m}\in\cS_{q,N}\), define
\[
  d(\underline{n},\underline{m})
  \eqdef
  \frac12\norm{\underline{n}-\underline{m}}_1.
\]
\end{definition}

Write \((x)_+\eqdef\max\{x,0\}\), applying this operation coordinatewise to
vectors.  Since the coordinate sums of \(\underline{n}\) and \(\underline{m}\) agree,
\begin{equation}
\label{eq:positive-part-distance}
  d(\underline{n},\underline{m})
  =
  \sum_{i=0}^{q-1}(n_i-m_i)_+
  =
  \sum_{i=0}^{q-1}(m_i-n_i)_+.
\end{equation}
Distance controls loss confusability: for
\(\underline{n},\underline{m}\in\cS_{q,N}\),
\(d(\underline{n},\underline{m})\leq t\) if and only if there are loss
patterns \(\underline{e}\leq\underline{n}\) and
\(\underline{f}\leq\underline{m}\) in \(\Z_{\geq0}^{q}\) with
\begin{equation}
\label{rem:loss-confusability}
  |\underline{e}|=|\underline{f}|\leq t,
  \qquad
  \underline{n}-\underline{e}=\underline{m}-\underline{f},
\end{equation}
where \(|\underline{e}|\eqdef\sum_ie_i\); we call
\(\underline{n}-\underline{e}=\underline{m}-\underline{f}\) a common descendant
of \(\underline{n}\) and \(\underline{m}\).  Indeed, in one direction take
\(\underline{e}\eqdef(\underline{n}-\underline{m})_+\) and
\(\underline{f}\eqdef(\underline{m}-\underline{n})_+\); the converse follows from
\(d(\underline{n},\underline{m})\leq|\underline{e}|\).

A classical code is a subset \(C\subseteq\cS_{q,N}\), and its
minimum distance is the minimum of \(d(\underline{n},\underline{m})\) over
distinct codewords.  Identifying \(\underline{n}\) with the multiset
containing \(n_i\) copies of the symbol \(i\), the associated multisets
correct \(t\) deletions when the minimum distance is at least \(t+1\).

For any metric space \(X\), point \(\underline{n}\in X\), and radius \(t\geq0\),
write
\[
  B_X(\underline{n},t)
  \eqdef
  \{\underline{m}\in X:d(\underline{n},\underline{m})\leq t\}.
\]
For any finite subset \(T\) of a metric space, we use the ordered
close-pair count (the total ball size of \cite{GoyalDaoKovacevicKiah2025})
\begin{equation}
\label{eq:ordered-pair-count}
  P_T(t)
  \eqdef
  \left|
    \left\{
      (\underline{n},\underline{m})\in T^2:
      d(\underline{n},\underline{m})\leq t
    \right\}
  \right|.
\end{equation}
Thus \(P_T(t)/|T|\) is the average radius-\(t\) ball size in the metric
induced on \(T\).

The number-state basis is
\(\lvert\underline{n}\rangle\eqdef\lvert n_0\rangle\otimes\cdots\otimes
\lvert n_{q-1}\rangle\).  The associated constant-excitation Fock space is
\[
  \cH_{q,N}
  \eqdef
  \operatorname{span}
  \left\{
    \lvert\underline{n}\rangle:
    \underline{n}\in\cS_{q,N}
  \right\}.
\]
A multimode loss pattern is a vector
\(\underline{e}\in\Z_{\geq0}^{q}\), of total order
\(|\underline{e}|\).  For loss probability \(0<\gamma<1\), the
pure-loss Kraus operator on the \(N\)-excitation sector is
\[
  A_{\underline{e}}^{(\gamma)}\lvert\underline{n}\rangle
  \eqdef
  \mathbf 1_{\{\underline{e}\leq\underline{n}\}}
  \gamma^{|\underline{e}|/2}(1-\gamma)^{(N-|\underline{e}|)/2}
  \sqrt{\prod_i\binom{n_i}{e_i}}\,
  \lvert\underline{n}-\underline{e}\rangle.
\]
Thus it vanishes unless \(\underline{e}\leq\underline{n}\) coordinatewise.  All
modes share the loss probability \(\gamma\).  The factor
\(\gamma^{|\underline{e}|/2}(1-\gamma)^{(N-|\underline{e}|)/2}\) depends only on
\(|\underline{e}|\), so exact correction of these operators does not depend on
\(\gamma\); the binomial factor enters the diagonal Knill--Laflamme
conditions below.

For a vector \(\underline{h}\in\cS_{q,N}\), write
\[
  \binom{N}{\underline{h}}
  \eqdef
  \frac{N!}{\prod_{i=0}^{q-1}h_i!}.
\]
We set a vector multinomial coefficient to zero whenever a lower entry is
negative or the lower entries do not sum to the upper argument.
For an integer \(0\leq t\leq N\), the point \(\underline{a}_{\underline{h}}\)
of \cite[Eq.~(27)]{AydinAlbertBarg2026}, which we call the moment point of
\(\underline{h}\), has coordinates
\begin{equation}
\label{eq:aab-moment}
  a_{\underline{e},\underline{h}}
  \eqdef
  \frac{\binom{N-t}{\underline{h}-\underline{e}}}
       {\binom{N}{\underline{h}}},
  \qquad
  \underline{e}\in\cS_{q,t}.
\end{equation}
We collect these coordinates into the vector
\[
  \underline{a}_{\underline{h}}
  \eqdef
  (a_{\underline{e},\underline{h}})_{\underline{e}\in\cS_{q,t}}
  \in\mathbb R^{\cS_{q,t}}.
\]

By the Knill--Laflamme conditions \cite{KnillLaflamme1997}, a code with orthonormal basis
\(\lvert c_1\rangle,\ldots,\lvert c_K\rangle\) exactly corrects all loss
patterns of total order at most \(t\) if and only if there are numbers
\(\lambda_{\underline{e},\underline{f}}\) such that
\[
  \langle c_i\rvert(A_{\underline{e}}^{(\gamma)})^\dagger
  A_{\underline{f}}^{(\gamma)}\lvert c_j\rangle
  =
  \delta_{ij}\lambda_{\underline{e},\underline{f}}
\]
for all \(i,j\) and all such patterns \(\underline{e},\underline{f}\).
The coordinates in \eqref{eq:aab-moment} are normalized diagonal loss
moments: for \(\underline{e}\in\cS_{q,t}\),
\[
  \langle\underline{h}\rvert(A_{\underline{e}}^{(\gamma)})^\dagger
  A_{\underline{e}}^{(\gamma)}\lvert\underline{h}\rangle
  =
  \gamma^{t}(1-\gamma)^{N-t}\binom Nt\binom t{\underline{e}}\,
  a_{\underline{e},\underline{h}},
\]
and the prefactor does not depend on \(\underline{h}\).  Suppose a code
\(C\subseteq\cS_{q,N}\) has minimum distance at least \(t+1\) and the
logical states are
\(\lvert c_j\rangle=\sum_{\underline{h}\in C_j}\sqrt{p_{\underline{h}}}\,
\lvert\underline{h}\rangle\) for disjoint parts \(C_1,\ldots,C_K\) of
\(C\), with probability weights \(p_{\underline{h}}\) on each part.
This separation (minimum distance at least \(t+1\)) removes every
off-diagonal Knill--Laflamme term, and the
diagonal conditions for losses of order \(t\) say that the points
\(\sum_{\underline{h}\in C_j}p_{\underline{h}}\underline{a}_{\underline{h}}\)
coincide for all \(j\).  We call this balancing the moment points.  A
Tverberg partition of the moment points into \(K\) parts whose convex
hulls share a point supplies such parts and weights.

The exact-order coordinates also determine every lower loss order.  For
\(0\leq r\leq t\) and \(\underline{u}\in\cS_{q,r}\), the multivariate
Vandermonde identity gives
\begin{equation}
\label{eq:lower-moments-from-top}
  \frac{\binom{N-r}{\underline{h}-\underline{u}}}
       {\binom N{\underline{h}}}
  =
  \sum_{\underline{v}\in\cS_{q,t-r}}
  \binom{t-r}{\underline{v}}\,
  a_{\underline{u}+\underline{v},\underline{h}}.
\end{equation}
Thus balancing the order-\(t\) moment points balances all orders at most
\(t\); classical separation removes the off-diagonal terms.  This
is why Theorem~\ref{thm:aab} below counts only the exact-order patterns rather than
their cumulative number
\cite[Eq.~(12) and proof of Thm.~2]{AydinAlbertBarg2026}.

The following affine-hull formulation is implicit in the proof of
the Aydin--Albert--Barg theorem
\cite[Thm.~7, Eqs.~(26)--(27), and proof]{AydinAlbertBarg2026}, which
applies Tverberg's theorem~\cite{Tverberg1966} in the hyperplane given by
\eqref{eq:lower-moments-from-top} with \(r=0\); here it is applied in the
actual affine hull.
The preceding paragraphs prove its Fock part, the only part used in this
paper; its permutation-invariant and spin parts follow by the same change,
since the cited proof obtains all three families from one balancing of the
moment points; their distances are defined in \cite{AydinAlbertBarg2026}.

\begin{theorem}[Aydin--Albert--Barg lifting, affine-hull form]
\label{thm:aab}
Let \(N,K\geq1\), \(q\geq2\), and \(0\leq t\leq N\) be integers.
Let \(C\subseteq\cS_{q,N}\) have minimum distance at least \(t+1\).
Set
\[
  r_t(C)
  \eqdef
  \operatorname{affdim}
  \{\underline{a}_{\underline{h}}:\underline{h}\in C\}.
\]
If
\[
  |C|\geq (K-1)(r_t(C)+1)+1,
\]
then there exist \(K\)-dimensional Fock, permutation-invariant, and spin
codes with their respective distance parameter \(t+1\).  In the terminology
of the cited paper, this is bosonic loss distance for the Fock code, deletion
or erasure distance for the permutation-invariant code, and spin distance for
degree-\(t\) products of \(\mathfrak{su}(q)\) generators
\cite[Secs.~IV--VII and Thm.~7]{AydinAlbertBarg2026}.  In particular, the Fock
code is spanned by states
\(\sum_{\underline{h}\in C_j}\sqrt{p_{\underline{h}}}\lvert\underline{h}\rangle\)
on disjoint parts \(C_j\subseteq C\), and it exactly corrects all loss
patterns of total order at most \(t\).

The Vandermonde affine relation in the cited proof, which is
\eqref{eq:lower-moments-from-top} with \(r=0\), gives
\(r_t(C)\leq\binom{q+t-1}{q-1}-1\).  Hence the published sufficient
condition is
\begin{equation}
\label{eq:aab-threshold}
  |C|
  \geq
  (K-1)\binom{q+t-1}{q-1}+1.
\end{equation}
\end{theorem}

For a sequence of codes with \(q_N=\lfloor\alpha N\rfloor\) modes and
dimensions \(K_N\), we report the dimension exponent \(N^{-1}\log_2K_N\) per
excitation; ``dimension exponent at least \(k\)'' means
\(\liminf_{N\to\infty}N^{-1}\log_2K_N\geq k\), and the same convention
applies to the sizes of classical codes.  We use
\(k\), with subscripts when needed, for such exponents.  For a code \(Q\) of dimension
\(K\), with \(q\geq2\) and \(N\geq1\), the relative rate in
\cite{ElimelechAydinBarg2026} is
\[
  R(Q)\eqdef\frac{\log_2 K}{\log_2\dim\cH_{q,N}}.
\]
Stirling's formula gives
\begin{equation}
\label{eq:full-simplex-exponent}
  \lim_{N\to\infty}
  \frac1N\log_2|\cS_{q_N,N}|
  =
  C_\infty(\alpha)
  \eqdef
  (1+\alpha)
  h_2\left(\frac{\alpha}{1+\alpha}\right),
\end{equation}
where
\[
  h_2(p)
  \eqdef
  -p\log_2p-(1-p)\log_2(1-p).
\]
Thus an exponent \(k\) per excitation corresponds to relative rate
\(k/C_\infty(\alpha)\).  For a probability distribution \(p=(p_i)\) on a
finite set, we write
\[
  H_2(p)\eqdef-\sum_i p_i\log_2p_i
\]
for its Shannon entropy.  We use the convention \(0\log_2 0\eqdef0\).
Unless stated otherwise, all logarithms from this point onward are to base
two.

\section{Two restriction principles}
\label{sec:restriction-principles}

Our exact codes are obtained in two steps: a separated
classical support is chosen inside the capped simplex, and the
Aydin--Albert--Barg construction lifts it to a quantum code.
Subsection~\ref{sec:restricted-gv} bounds the first step when the
candidate codewords are restricted to a subset;
Subsection~\ref{sec:cap-aware} removes from the second step the moment
coordinates that vanish under the cap.

\subsection{Average-ball bounds on restricted simplices}
\label{sec:restricted-gv}

The usual greedy Gilbert argument is governed by a maximum ball.  In a
nonhomogeneous metric space, the generalized Gilbert--Varshamov (GV) bound
replaces this maximum by an average.  The following form makes clear that
the candidate codewords may be restricted before the graph argument is
applied.

\begin{proposition}[Restricted average-ball bound]
\label{prop:restricted-average}
Let \(X\) be a finite metric space, \(\varnothing\ne T\subseteq X\), and \(t\geq0\).
There is a code \(C\subseteq T\) of minimum distance greater than \(t\)
such that
\begin{align}
  |C|
  &\geq
  \left\lceil
    \sum_{\underline{n}\in T}
    \frac{1}{|B_X(\underline{n},t)\cap T|}
  \right\rceil
  \label{eq:caro-wei}\\
  &\geq
  \left\lceil\frac{|T|^2}{P_T(t)}\right\rceil
  \geq
  \left\lceil\frac{|T|^2}{P_X(t)}\right\rceil.
  \label{eq:restricted-average}
\end{align}
\end{proposition}

This is the Caro--Wei/generalized-GV bound on the graph with vertex set
\(T\), in which distinct vertices are adjacent exactly when their distance
is at most \(t\) \cite{GuFuja1993,Tolhuizen1997,AlonSpencer2016}.  To see it
directly, order the vertices uniformly at random and retain each vertex
that precedes all its neighbors.  The retained vertices form an independent
set, and a vertex \(\underline{n}\) is retained with probability
\(1/|B_X(\underline{n},t)\cap T|\).  Some ordering attains at least the
ceiling of the expected size, proving \eqref{eq:caro-wei}.
Cauchy--Schwarz and \(P_T(t)\leq P_X(t)\) give
\eqref{eq:restricted-average}.  The useful consequence is the exponent
preservation stated below.

\begin{corollary}[Subexponential restriction]
\label{cor:subexponential-restriction}
Let \(\varnothing\ne T_N\subseteq X_N\) be finite metric spaces and let \(t_N\geq0\).
If
\[
  \log_2\frac{|X_N|}{|T_N|}=o(N),
\]
then there are codes \(C_N\subseteq T_N\) of minimum distance greater than
\(t_N\) with
\[
  |C_N|\geq2^{-o(N)}\,\frac{|X_N|^2}{P_{X_N}(t_N)};
\]
thus the ambient average-ball exponent is achieved inside \(T_N\).
\end{corollary}

\begin{proof}
With \(\eta_N\eqdef|T_N|/|X_N|\), apply
Proposition~\ref{prop:restricted-average}; the factor
\(\eta_N^2=2^{-o(N)}\) does not change the exponent.
\end{proof}

Restricting the Gilbert argument to a subset of prescribed statistics is
the standard device behind Marcus--Roth-type bounds
\cite{GoyalDaoKovacevicKiah2025,MarcusRoth1992}; the corollary records the
case of subexponential density, where nothing is lost.

In particular, \(T_N\) need not have density \(1-o(1)\).  In
Section~\ref{sec:log-caps}, we apply this observation first to the simplex
under an occupation cap that tends to infinity and then inside the optimizing zero-count
slices of~\cite{GoyalDaoKovacevicKiah2025}, the subsets of the discrete
simplex with a prescribed number of zero coordinates.

\subsection{Cap-aware Tverberg lifting}
\label{sec:cap-aware}

For an integer cap \(B\geq1\), define the capped simplex
\begin{equation}
\label{eq:capped-simplex}
  \cS_{q,N}^{(B)}
  \eqdef
  \left\{
    \underline{n}\in\cS_{q,N}:
    \norm{\underline{n}}_\infty\leq B
  \right\}.
\end{equation}
Its Fock span is \(\cH_{q,N}^{(B)}\eqdef\operatorname{span}
\{\lvert\underline{n}\rangle:\underline{n}\in\cS_{q,N}^{(B)}\}\),
as in \cite[Definition~2]{ElimelechAydinBarg2026}.
The occupation constraint forces many moment coordinates to vanish.  The
capped loss-pattern set is
\begin{equation}
\label{eq:capped-errors}
  \cE_B(q,t)
  \eqdef
  \left\{
    \underline{e}\in\cS_{q,t}:
    \norm{\underline{e}}_\infty\leq B
  \right\},
\end{equation}
whose cardinality, the cap-aware denominator, is
\begin{equation}
\label{eq:moment-count}
  m_B(q,t)
  \eqdef
  |\cE_B(q,t)|
  =
  [z^t](1+z+\cdots+z^B)^q.
\end{equation}

\begin{theorem}[Cap-aware AAB lifting]
\label{thm:cap-aware-aab}
Let \(N,B,K\geq1\), \(q\geq2\), and \(0\leq t\leq N\) be integers.
Let \(C\subseteq\cS_{q,N}^{(B)}\) have minimum distance at least \(t+1\).
If
\begin{equation}
\label{eq:cap-aware-threshold}
  |C|\geq (K-1)m_B(q,t)+1,
\end{equation}
then the Aydin--Albert--Barg construction gives
\(K\)-dimensional Fock, permutation-invariant, and spin codes of distance
\(t+1\).  The Fock code is supported entirely on occupation vectors
\(\underline{h}\) satisfying \(\norm{\underline{h}}_\infty\leq B\).
\end{theorem}

\begin{proof}
Project the moment point \(\underline{a}_{\underline{h}}\) from
\eqref{eq:aab-moment} onto the capped pattern set \(\cE_B(q,t)\).  If
\(\underline{e}\notin\cE_B(q,t)\), then some coordinate satisfies
\(e_i>B\geq h_i\), so
\[
  \binom{N-t}{\underline{h}-\underline{e}}=0
  \qquad\text{and hence}\qquad
  a_{\underline{e},\underline{h}}=0
\]
for every \(\underline{h}\in C\).  Projecting the moment points onto the
coordinates in \(\cE_B(q,t)\) therefore removes only coordinates that vanish
on the whole capped simplex, which we call inactive.

The multivariate Vandermonde identity gives
\begin{align}
  \sum_{\underline{e}\in\cE_B(q,t)}
  \binom{t}{\underline{e}}\,
  a_{\underline{e},\underline{h}}
  &=
  \frac{1}{\binom{N}{\underline{h}}}
  \sum_{\underline{e}\in\cS_{q,t}}
  \binom{t}{\underline{e}}
  \binom{N-t}{\underline{h}-\underline{e}}
  \notag\\
  &=1.
  \label{eq:cap-affine-identity}
\end{align}
The first equality holds because the omitted terms vanish.  Thus all
projected points lie in an affine hyperplane of dimension
\(m_B(q,t)-1\), and the projection preserves affine dimension, so
\(r_t(C)\leq m_B(q,t)-1\).  The affine-hull form of
Theorem~\ref{thm:aab} now gives the claimed threshold and all three code
families.  Their Fock basis vectors use only the original support
\(C\subseteq\cS_{q,N}^{(B)}\).
\end{proof}

This is the specialization to the cap of the general principle that
redundant diagonal error conditions may be deleted before convex balancing;
compare \cite[Remark~5]{ElimelechAlbertBargAQEC2026}, and, for \(q\)-ary
codes whose total excitation is not fixed,
\cite[Sec.~5.2.2]{ElimelechAlbertBargAQEC2026}.  The quantitative
input used below is the explicit cap-aware denominator \(m_B(q,t)\) on a capped bosonic
support.  Every exact Fock code obtained below through
Theorem~\ref{thm:cap-aware-aab} comes with permutation-invariant and spin
codes of the same dimension and distance on the same support; we state
only the Fock codes.

\begin{remark}[Support-specific refinement]
\label{rem:support-sensitive}
For a specific support \(C\subseteq\cS_{q,N}^{(B)}\), let
\[
  \cE_t^{\mathrm{act}}(C)
  \eqdef
  \{\underline{e}\in\cS_{q,t}:
    \underline{e}\leq\underline{h}\ \text{for some }\underline{h}\in C\}.
\]
The affine-hull denominator in Theorem~\ref{thm:aab} is \(r_t(C)+1\), and
\[
  r_t(C)+1\leq|\cE_t^{\mathrm{act}}(C)|\leq m_B(q,t).
\]
Thus the active-pattern count is a support-specific bound that is easier to
evaluate, while \(r_t(C)+1\) itself can be smaller still.
The affine identity \eqref{eq:cap-affine-identity} places every column of the
matrix below in a hyperplane not containing the origin, and therefore gives the exact identity
\[
  r_t(C)+1
  =
  \operatorname{rank}_{\mathbb R}
  \bigl(a_{\underline{e},\underline{h}}\bigr)_
  {\underline{e}\in\cE_t^{\mathrm{act}}(C),\ \underline{h}\in C}.
\]
\end{remark}

\begin{remark}[Comparison with the published count]
\label{rem:moment-comparison}
Since
\[
  m_B(q,t)
  \leq
  |\cS_{q,t}|
  =
  \binom{q+t-1}{q-1},
\]
the cap-aware threshold in Theorem~\ref{thm:cap-aware-aab} is never weaker
than the published cardinality condition \eqref{eq:aab-threshold}.
The two displayed counts are equal whenever \(t\leq B\), because then
every weak composition of \(t\) already has all coordinates at most \(B\).
\end{remark}

\paragraph{Construction scheme.}
Every construction below follows the same two steps:
\[
  \text{capped simplex}
  \longrightarrow
  \text{separated classical support}
  \longrightarrow
  \text{quantum code}.
\]
For a candidate subset \(T\), Proposition~\ref{prop:restricted-average}
controls the first arrow through its ordered close-pair count \(P_T(t)\).
Theorem~\ref{thm:cap-aware-aab} controls the second: on the exponential
scale, it divides the support size by the cap-aware denominator \(m_B(q,t)\);
the support-specific counts of Remark~\ref{rem:support-sensitive} are not
used in the exponents below (see the discussion after
Theorem~\ref{thm:general-fixed-cap}).
Thus every exact-code exponent below has the form
\[
  \text{classical separated-support exponent}
  -
  \text{cap-aware moment exponent}.
\]
Sections~\ref{sec:cap-two} and~\ref{sec:general-caps} evaluate this
difference, asymptotically and at finite length.
Section~\ref{sec:converses} bounds how large separated supports and
arbitrary capped codes can be.  Section~\ref{sec:random} samples supports
exactly and, for one normalized channel, replaces the Tverberg step by a
random partition.  Section~\ref{sec:log-caps} applies
Proposition~\ref{prop:restricted-average} to caps that grow with \(N\).

\section{Exact cap-two Fock codes}
\label{sec:cap-two}

We now turn to the balanced cap-two simplex
\begin{equation}
\label{eq:cap-two-sector}
  \cS_{N,N}^{(2)}
  =
  \left\{
    \underline{n}\in\{0,1,2\}^{N}:
    \sum_{i=0}^{N-1}n_i=N
  \right\}.
\end{equation}
This is the smallest nontrivial capped simplex when the number of modes
equals the total excitation (at cap one only the all-one vector remains).  Its symmetry makes both the space exponent
and the close-pair exponent explicit.

Translation by the all-one vector maps this simplex onto the zero-sum
vectors of the ternary cube \(\{-1,0,1\}^N\) and doubles deletion distance
into Manhattan distance, so Tolhuizen's ternary pair enumerator and
constant-weight Gilbert--Varshamov bound \cite{TolhuizenTernary2010}
apply.  We first record the bound obtained by averaging over all of
\(\cS_{N,N}^{(2)}\) (Corollary~\ref{cor:cap-two-average}); it serves as a benchmark.
Theorem~\ref{thm:cap-two-cc} replaces it by the stronger exponent of
Tolhuizen's constant-weight bound, Lemma~\ref{lem:cap-two-moments} evaluates
the cap-aware denominator \(m_2(N,t)\), and Theorem~\ref{thm:exact-cap-two}
subtracts the two.

\begin{lemma}[Size of the balanced cap-two simplex]
\label{lem:central-trinomial}
The cardinality of \(\cS_{N,N}^{(2)}\) is the central trinomial coefficient, and
\begin{equation}
\label{eq:central-trinomial}
  |\cS_{N,N}^{(2)}|
  =
  [z^N](1+z+z^2)^N
  =
  \left(
    \frac{\sqrt{3}}{2\sqrt{\pi N}}
    +o(N^{-1/2})
  \right)3^N.
\end{equation}
\end{lemma}

The identity and asymptotic are the standard central-trinomial
coefficient estimate, equivalently the lattice local central limit theorem;
see, for example, \cite[Thm.~VIII.8]{FlajoletSedgewick2009}.

For \(0\leq\delta\leq4/9\), let \(u_\delta\in[0,1]\) be the unique
solution of
\begin{equation}
\label{eq:u-delta}
  \delta
  =
  \frac{2u_\delta+2u_\delta^2}
       {3+4u_\delta+2u_\delta^2}.
\end{equation}
For \(0<\delta<4/9\), this solution is
\begin{equation}
\label{eq:u-explicit}
  u_\delta
  =
  \frac{
    2\delta-1+\sqrt{1+2\delta-2\delta^2}
  }{2(1-\delta)}.
\end{equation}
Define
\begin{equation}
\label{eq:Psi-two}
  \Psi_{\mathrm{amb}}^{(2)}(\delta)
  \eqdef
  \log_2(3+4u_\delta+2u_\delta^2)
  -
  2\delta\log_2u_\delta,
\end{equation}
with the continuous value \(\Psi_{\mathrm{amb}}^{(2)}(0)=\log_2 3\).

\begin{lemma}[Tolhuizen ambient-pair bound for the balanced cap-two simplex]
\label{lem:cap-two-pairs}
Let \(t_N\eqdef\lfloor\delta N\rfloor\), where
\(0\leq\delta<4/9\).  Then
\begin{equation}
\label{eq:cap-two-pair-asymptotic}
  P_{\cS_{N,N}^{(2)}}(t_N)
  \leq
  2^{N(\Psi_{\mathrm{amb}}^{(2)}(\delta)+o(1))}.
\end{equation}
\end{lemma}

\begin{proof}
After translating \(n_i\mapsto n_i-1\), the ambient alphabet is
\(\{-1,0,1\}\), and Tolhuizen's full-cube pair enumerator is
\((3+4u+2u^2)^N\)
\cite[Prop.~3 and Eq.~(6)]{TolhuizenTernary2010}.  Since
\(\cS_{N,N}^{(2)}-\mathbf 1\subseteq\{-1,0,1\}^N\), containment gives the required upper
bound.  The ordinary one-variable saddle for total Manhattan distance at
most \(2\delta N\) satisfies \eqref{eq:u-delta}; substituting it gives
\eqref{eq:Psi-two}.
\end{proof}

Combining Lemmas~\ref{lem:central-trinomial} and~\ref{lem:cap-two-pairs} with
Proposition~\ref{prop:restricted-average} yields the classical exponent
\begin{equation}
\label{eq:cap-two-average-rate}
  R_{\mathrm{amb}}^{(2)}(\delta)
  \eqdef
  2\log_2 3-\Psi_{\mathrm{amb}}^{(2)}(\delta).
\end{equation}

\begin{corollary}[Ambient-pair cap-two bound]
\label{cor:cap-two-average}
For every \(0\leq\delta<4/9\), there are codes
\(C_N\subseteq\cS_{N,N}^{(2)}\) of minimum distance at least
\(\lfloor\delta N\rfloor+1\) such that
\[
  \liminf_{N\to\infty}
  \frac1N\log_2|C_N|
  \geq
  R_{\mathrm{amb}}^{(2)}(\delta).
\]
\end{corollary}

\begin{proof}
By Lemmas~\ref{lem:central-trinomial} and~\ref{lem:cap-two-pairs},
\[
  \frac{|\cS_{N,N}^{(2)}|^2}
       {P_{\cS_{N,N}^{(2)}}(\lfloor\delta N\rfloor)}
  \geq
  2^{N(R_{\mathrm{amb}}^{(2)}(\delta)+o(1))}.
\]
Apply Proposition~\ref{prop:restricted-average}.
\end{proof}

The general fixed-cap analogue of Corollary~\ref{cor:cap-two-average}
governs uniform sampling from the capped simplex and admits the two-saddle
evaluation in Proposition~\ref{prop:two-saddle}.

The preceding argument averages over all occupation histograms; we call it
unstratified.  A stronger bound fixes a histogram \(\underline{v}\), whose
entry \(v_a\) counts the modes with occupation \(a\) (its type is
\(\pi=\underline{v}/q\)), applies the average-ball argument inside its
layer, the type class of all vectors with that histogram, and optimizes
over \(\pi\).  By the method of types, the exponent of one layer, its
constant-composition Gilbert exponent, is \(2H_2(\pi)\) minus the largest
entropy \(H_2(P)\) of a joint distribution \(P\) on \(\{0,1,2\}^2\) with both
marginals \(\pi\) and \(\sum_{a,b}|a-b|P_{ab}\leq2\delta\).  In the balanced
cap-two simplex the optimization over \(\pi\) has a closed-form solution.

For \(0\leq\delta\leq1/2\), let the superscript \(\mathrm{cc}\) denote
constant composition, and set
\begin{equation}
\label{eq:u-cc-delta}
  u_\delta^{\mathrm{cc}}
  \eqdef
  \begin{cases}
  0,
    &\delta=0,\\[2mm]
  \displaystyle
  \frac{
    \sqrt{1-2\delta+4\delta^2}-(1-\delta)
  }{\delta},
    &0<\delta\leq1/2.
  \end{cases}
\end{equation}
Equivalently, \(u=u_\delta^{\mathrm{cc}}\) is the unique point of
\([0,1]\) satisfying
\begin{equation}
\label{eq:u-cc-implicit}
  \delta
  =
  \frac{2u}{(3-u)(1+u)}.
\end{equation}
Define
\begin{equation}
\label{eq:R-cc-two}
  R_{\mathrm{cc}}^{(2)}(\delta)
  \eqdef
  \log_2\frac{3-u_\delta^{\mathrm{cc}}}
                 {1+u_\delta^{\mathrm{cc}}}
  +
  2\delta\log_2u_\delta^{\mathrm{cc}},
\end{equation}
where the final term is interpreted by continuity at \(\delta=0\).

\begin{theorem}[Tolhuizen exponent for balanced cap two]
\label{thm:cap-two-cc}
For every \(0\leq\delta<1/2\), there are codes
\(C_N\subseteq\cS_{N,N}^{(2)}\) of minimum distance at least
\(\lfloor\delta N\rfloor+1\) such that
\begin{equation}
\label{eq:cap-two-cc-classical}
  \liminf_{N\to\infty}
  \frac1N\log_2|C_N|
  \geq
  R_{\mathrm{cc}}^{(2)}(\delta).
\end{equation}
The type that maximizes the constant-composition Gilbert exponent is,
asymptotically,
\begin{equation}
\label{eq:cap-two-cc-histogram}
  \pi_\delta
  \eqdef
  (\rho_\delta,1-2\rho_\delta,\rho_\delta),
  \qquad
  \rho_\delta
  \eqdef
  \frac1{3-u_\delta^{\mathrm{cc}}}.
\end{equation}
Moreover,
\begin{equation}
\label{eq:strict-cc-improvement}
  R_{\mathrm{cc}}^{(2)}(\delta)
  >
  R_{\mathrm{amb}}^{(2)}(\delta)
  \qquad
  (0<\delta<4/9).
\end{equation}
The exponent \(R_{\mathrm{cc}}^{(2)}(\delta)\) is positive precisely for
\(0\leq\delta<1/2\).
\end{theorem}

\begin{proof}
Write \(u\eqdef u_\delta^{\mathrm{cc}}\).  Only the optimality of
\(\pi_\delta\) among cap-two types uses
\cite[App.~B]{TolhuizenTernary2010}, and through it so does
\eqref{eq:strict-cc-improvement}.  We first prove \eqref{eq:cap-two-cc-classical} inside the layer of type
\(\pi_\delta\).  Index rows and columns by the ordered alphabet \(0,1,2\),
so that \(P_\delta(a,b)\) is the limiting fraction of coordinates with first
symbol \(a\) and second symbol \(b\).  The asymptotically optimizing
close-pair joint type of this layer is
\begin{equation}
\label{eq:cap-two-cc-joint-type}
  P_\delta
  \eqdef
  \frac1{(3-u)(1+u)}
  \begin{pmatrix}
    1 & u(1-u) & u^2\\
    u(1-u) & (1-u)^2 & u(1-u)\\
    u^2 & u(1-u) & 1
  \end{pmatrix}.
\end{equation}
Equivalently, \(P_\delta=(x_ax_bu^{|a-b|})_{a,b}\) with
\(x\propto(1,1-u,1)\) is the maximum-entropy coupling with both marginals
\(\pi_\delta\) and mean distance \(\sum_{a,b}|a-b|P_{ab}=2\delta\); since
\(\pi_a/x_a\) does not depend on \(a\), one has
\(2H_2(\pi_\delta)-H_2(P_\delta)=\log_2\frac{3-u}{1+u}+2\delta\log_2u
=R_{\mathrm{cc}}^{(2)}(\delta)\).  Hence, by the method of types,
Proposition~\ref{prop:restricted-average} applied inside the layer of
type \(\pi_\delta\) gives \eqref{eq:cap-two-cc-classical} without
\cite[App.~B]{TolhuizenTernary2010}.  At \(\delta=1/2\) one has
\(u_\delta^{\mathrm{cc}}=1\) and \(R_{\mathrm{cc}}^{(2)}(1/2)=0\); since the
envelope identity gives
\(\frac{d}{d\delta}R_{\mathrm{cc}}^{(2)}(\delta)=2\log_2u_\delta^{\mathrm{cc}}<0\)
on \((0,1/2)\), the exponent \(R_{\mathrm{cc}}^{(2)}\) is positive exactly
on \([0,1/2)\).

We next prove the optimality \eqref{eq:cap-two-cc-histogram} of
\(\pi_\delta\).  Tolhuizen's optimized ternary
constant-weight Gilbert--Varshamov bound is Eq.~(16) of
\cite{TolhuizenTernary2010}, optimized over \(\omega\) in its App.~B.  Set \(D\eqdef2\delta\),
\[
  \omega^\ast
  \eqdef
  \frac{1+D+\sqrt{D^2-D+1}}3,
  \qquad
  w_N\eqdef2\left\lfloor\frac{\omega^\ast N}{2}\right\rfloor.
\]
Translation by \(\mathbf 1\) maps the layer of the required cap-two histogram
to the
zero-sum sublayer of the ternary weight-\(w_N\) shell and doubles deletion
distance into Manhattan distance.  This sublayer has relative density
\[
  \binom{w_N}{w_N/2}/2^{w_N}=N^{-1/2+o(1)},
\]
so Corollary~\ref{cor:subexponential-restriction} transfers Tolhuizen's
exponent without exponential loss (this re-proves \eqref{eq:cap-two-cc-classical}).  Here
\cite[Eq.~(16)]{TolhuizenTernary2010} is the asymptotic form of the
constant-weight Gilbert bound on the weight-\(w_N\) shell \(X_N\), whose
balls all have the same size \cite[Prop.~4]{TolhuizenTernary2010}; hence
it is the exponent of \(|X_N|^2/P_{X_N}\) at Manhattan radius \(2t_N\).
Tolhuizen normalizes rates by \(\log_3\); the bracket in his Eq.~(16) is
the exponent in bits.  With \(\omega^\ast=2/(3-u)\), his
critical point becomes \(\delta={2u/((3-u)(1+u))}\), giving
\eqref{eq:R-cc-two} at the weight \(\omega^\ast=2\rho_\delta\) of \(\pi_\delta\).  In Tolhuizen's
parameters, with \(Z\eqdef(3-u)(1+u)\), one has \(\beta=2u/Z\) and
\(g=2u(1-u)/Z\), writing \(g\) for the variable \(\gamma\) of his
Eq.~(16); splitting the nonzero mass equally between the two signs gives
\eqref{eq:cap-two-cc-joint-type}.

It remains to show that \(\pi_\delta\) is a global optimizer, which completes
\eqref{eq:cap-two-cc-histogram}.  Every cap-two type is
\(\pi=(\rho,1-2\rho,\rho)\), and after translation its layer is the zero-sum
part of the ternary shell of weight \(\omega=2\rho\).  The reflection
\(a\mapsto2-a\) exchanges the two signs of the shell and preserves the
strictly concave maximum-entropy problem that gives the pair exponent of the
shell, so its unique maximizing close-pair coupling has both marginals
\(\pi\); hence the layer and the shell have the same word and pair exponents,
and the constant-composition Gilbert exponent of \(\pi\) is Tolhuizen's
constant-weight exponent at \(\omega=2\rho\).  For
\(\omega\leq1-\sqrt{1-2\delta}\) the product coupling of \(\pi\) has mean
distance \(\sum_{a,b}|a-b|\pi_a\pi_b=2\omega-\omega^2\leq2\delta\), so it is
admissible and this exponent is zero.  On the remaining range
\(1-\sqrt{1-2\delta}<\omega\leq1\), the right side of
\cite[Eq.~(16)]{TolhuizenTernary2010}, in bits, vanishes at the left end,
has the unique critical point \(\omega^\ast\)
\cite[App.~B]{TolhuizenTernary2010}, and tends to \(1-h_2(\delta)\) as
\(\omega\to1\).  The difference \(R_{\mathrm{cc}}^{(2)}(\delta)-(1-h_2(\delta))\)
vanishes at \(\delta=1/2\) and has derivative
\(\log_2\bigl(u(3-u^2)/2\bigr)<0\), because
\(u^3-3u+2=(u-1)^2(u+2)>0\) for \(0\leq u<1\).

We now prove \eqref{eq:strict-cc-improvement}.  For
\(0<\delta<4/9\), return to the ambient saddle \(u_\delta\) of
\eqref{eq:u-delta}, which is distinct from the optimized-type saddle
used above.  Among probability arrays \(P=(P_{ab})_{a,b=0}^2\) with
\(\sum_{a,b}|a-b|P_{ab}\leq2\delta\), the unique maximizer of
\(H_2(P)\) has entries proportional to \(u_\delta^{|a-b|}\).
Its value is \(\Psi_{\mathrm{amb}}^{(2)}(\delta)\), and its marginals are
proportional to
\(1+u_\delta+u_\delta^2,1+2u_\delta,1+u_\delta+u_\delta^2\).
They are nonuniform because \(0<u_\delta<1\).  Requiring both marginals
to equal \((1/3,1/3,1/3)\) excludes this unique maximizer.  Compactness
and strict entropy concavity therefore make the maximum pair entropy
strictly smaller than \(\Psi_{\mathrm{amb}}^{(2)}(\delta)\), while the
single-word entropy remains \(\log_2 3\).  More explicitly, types
tending to this uniform marginal have \(2^{N\log_2 3+o(N)}\) words;
their ordered close-pair count has exponent equal to the maximum joint
entropy under the two fixed marginals and the same distance constraint.
Proposition~\ref{prop:restricted-average} subtracts this pair exponent
from twice the word exponent.  The resulting Gilbert exponent of the
uniform type is therefore strictly larger than
\(2\log_2 3-\Psi_{\mathrm{amb}}^{(2)}(\delta)\); optimization over all
types can only increase it.  This proves \eqref{eq:strict-cc-improvement};
compare the numerical observation in
\cite[Sec.~7 and Fig.~2]{TolhuizenTernary2010}.
\end{proof}

It remains to evaluate the cap-aware denominator.  For
\(0\leq\delta<2\), let \(v_\delta\geq0\) solve
\begin{equation}
\label{eq:v-delta}
  \delta
  =
  \frac{v_\delta+2v_\delta^2}
       {1+v_\delta+v_\delta^2},
\end{equation}
and define
\begin{equation}
\label{eq:Lambda-two}
  \Lambda_2(\delta)
  \eqdef
  \log_2(1+v_\delta+v_\delta^2)
  -
  \delta\log_2v_\delta,
\end{equation}
with \(\Lambda_2(0)=0\).

\begin{lemma}[Cap-two moment exponent]
\label{lem:cap-two-moments}
For \(t_N=\lfloor\delta N\rfloor\) and \(0\leq\delta<2\),
\[
  m_2(N,t_N)
  =
  [z^{t_N}](1+z+z^2)^N
  =
  2^{N(\Lambda_2(\delta)+o(1))}.
\]
\end{lemma}

This is the standard large-powers coefficient estimate
\cite[Thm.~VIII.8]{FlajoletSedgewick2009}; its positive saddle is
\(v=v_\delta\) from \eqref{eq:v-delta}.

Subtracting the cap-two moment exponent from Tolhuizen's classical exponent
gives the quantum dimension exponent.

\begin{theorem}[Exact cap-two Fock codes]
\label{thm:exact-cap-two}
Define
\begin{equation}
\label{eq:kappa-two}
  k_2(\delta)
  \eqdef
  R_{\mathrm{cc}}^{(2)}(\delta)-\Lambda_2(\delta),
  \qquad
  0\leq\delta<1/2.
\end{equation}
Let \(\delta_2^\ast\) be the unique zero of
\(k_2(\delta)\) in \((0,1/2)\).  For every
\(0<\delta<\delta_2^\ast\), there are exact \(N\)-mode Fock codes of
total excitation \(N\), distance \(\lfloor\delta N\rfloor+1\), and
local occupation at most two, whose dimensions satisfy
\begin{equation}
\label{eq:cap-two-main-rate}
  \liminf_{N\to\infty}
  \frac1N\log_2K_N
  \geq
  k_2(\delta).
\end{equation}
Numerically,
\[
  \delta_2^\ast
  \approx
  0.1355.
\]
\end{theorem}

\begin{proof}
Choose the classical codes in
Theorem~\ref{thm:cap-two-cc}, and apply
Theorem~\ref{thm:cap-aware-aab} with
\(m_2(N,\lfloor\delta N\rfloor)\).  If \(L_N\eqdef|C_N|\), the guaranteed
quantum dimension is
\[
  K_N
  \geq
  1+
  \left\lfloor
    \frac{L_N-1}
         {m_2(N,\lfloor\delta N\rfloor)}
  \right\rfloor.
\]
Theorem~\ref{thm:cap-two-cc} and
Lemma~\ref{lem:cap-two-moments} give
\eqref{eq:cap-two-main-rate}.

At \(\delta=0\), one has
\(k_2(0)=\log_2 3>0\).  On \(0<\delta<1/2\), the envelope
identity for the optimized layer is
\[
  \frac{d}{d\delta}
  R_{\mathrm{cc}}^{(2)}(\delta)
  =
  2\log_2u_\delta^{\mathrm{cc}}<0,
\]
while \(\Lambda_2'(\delta)=-\log_2v_\delta>0\).  Hence
\(k_2\) is strictly decreasing.  As \(\delta\uparrow1/2\), the
classical term tends to zero while \(\Lambda_2(\delta)\) tends to the
positive value \(\Lambda_2(1/2)\).  Hence \(k_2\) becomes negative
before the endpoint and has exactly one zero in \((0,1/2)\).
The computer-assisted enclosure \(0.1355<\delta_2^\ast<0.1356\), certified
by exact rational interval arithmetic in Appendix~\ref{app:verification},
gives the stated value.
\end{proof}

\begin{remark}[Exponents at \(\delta=0.1\)]
\label{rem:cap-two-numerics}
At \(\delta=0.1\),
\[
  u_{0.1}^{\mathrm{cc}}
  \approx
  0.165,
  \qquad
  \rho_{0.1}
  \approx
  0.353,
\]
and
\[
  R_{\mathrm{cc}}^{(2)}(0.1)
  \approx
  0.763.
\]
Subtracting the moment exponent \(\Lambda_2(0.1)\) of \eqref{eq:Lambda-two}
gives
\[
  k_2(0.1)
  \approx
  0.281.
\]
Since \(C_\infty(1)=2\), the rate relative to the full
\(N\)-mode, \(N\)-excitation Hilbert space is at least
\[
  0.140.
\]
If one uses the published AAB denominator in
\eqref{eq:aab-threshold}, the achievable exponent is instead
\begin{equation}
\label{eq:cap-two-published-aab}
  k_2^{\mathrm{AAB}}(\delta)
  \eqdef
  R_{\mathrm{cc}}^{(2)}(\delta)
  -(1+\delta)h_2\left(\frac1{1+\delta}\right).
\end{equation}
Its positive-rate endpoint is approximately
\(0.1352\), against \(\delta_2^\ast\approx0.1355\), and its exponent at
\(\delta=0.1\) is approximately \(0.280\).
Table~\ref{tab:cap-two-improvement} isolates the classical input: with the
same cap-aware denominator, the optimized layer raises the exponent
at \(\delta=0.1\) from \(0.275\) to \(0.281\) and the endpoint from about
\(0.1342\) to \(0.1355\).  As a function of \(\delta\), the first row coincides with the qutrit rate of
\cite[Eq.~(82)]{ElimelechAlbertBargAQEC2026} for codes whose total
excitation is not fixed, after conversion from base-3 logarithms per qutrit
to bits per excitation (\(\delta\) is unchanged, since the \(N\) modes carry
\(N\) excitations): both equal
\(2\log_23-\Psi_{\mathrm{amb}}^{(2)}(\delta)-\Lambda_2(\delta)\), the
Gilbert exponent of the full cube \(\{0,1,2\}^N\), which the unstratified
bound uses (Lemmas~\ref{lem:central-trinomial} and~\ref{lem:cap-two-pairs}),
minus the exponent of the loss patterns with at most two losses per mode.
\end{remark}

\begin{table}[ht]
\centering
\caption{Exponents of exact balanced cap-two Fock codes;
both rows use the cap-aware denominator \(m_2(N,t)\).
All exponents and positive-rate endpoints are normalized per excitation
and rounded to three (exponents) and four (endpoints) decimal places.
The second row rests on the interval certificates of
Appendix~\ref{app:verification}; the first row is a numerical evaluation.}
\label{tab:cap-two-improvement}
\begin{tabular}{@{}lrr@{}}
\toprule
classical input&
exponent at \(\delta=0.1\)&
positive-rate endpoint\\
\midrule
unstratified average&
\(0.275\)&
\(0.1342\)\\
optimized layer&
\(0.281\)&
\(0.1355\)\\
\bottomrule
\end{tabular}
\end{table}

\section{General fixed occupation caps}
\label{sec:general-caps}

The cap-two calculation is a closed-form instance of a general
finite-alphabet method-of-types formula.  Fix \(B\geq1\) and
\(\alpha>0\) with \(\alpha B>1\), and put
\[
  q_N\eqdef\lfloor\alpha N\rfloor,
  \qquad
  t_N\eqdef\lfloor\delta N\rfloor,
  \qquad
  F_B(z)\eqdef1+z+\cdots+z^B.
\]
The condition \(\alpha B>1\) places the mean occupation \(1/\alpha\)
strictly inside the interval \([0,B]\).

\subsection{Exponents for general fixed caps}
\label{subsec:general-exponents}

As for cap two, the exponent of Theorem~\ref{thm:general-fixed-cap} is a
constant-composition Gilbert exponent \(R_B^{\mathrm{cc}}\) of the
classical support minus the moment exponent \(\Lambda_B\) of the cap-aware
denominator \(m_B(q_N,t_N)\).  We also need the space exponent \(C_B\) of
the capped simplex and, for the unstratified benchmark
(Proposition~\ref{prop:two-saddle}, Table~\ref{tab:fixed-cap-rates}) and
uniform sampling (Section~\ref{sec:random}), the pair exponent
\(\Psi_B\).

The space and moment exponents below are values of one coefficient
exponent.  The map
\[
  w\longmapsto\frac{wF'_B(w)}{F_B(w)}
\]
is strictly increasing from \(0\) to \(B\) on \((0,\infty)\).  For
\(0<\mu<B\), let \(w_\mu>0\) be the unique solution of
\[
  \frac{w_\mu F'_B(w_\mu)}{F_B(w_\mu)}=\mu,
\]
and define
\begin{equation}
\label{eq:H-cap-coefficient}
  H_B^{\mathrm{coeff}}(\mu)
  \eqdef
  \log_2F_B(w_\mu)-\mu\log_2w_\mu.
\end{equation}
Set \(H_B^{\mathrm{coeff}}(0)\eqdef0\) and \(H_B^{\mathrm{coeff}}(B)\eqdef0\) by continuity.
By the large-powers estimate \cite[Thm.~VIII.8]{FlajoletSedgewick2009},
if \(n_k\to\infty\) and \(j_k/n_k\to\mu\in(0,B)\), then
\begin{equation}
\label{eq:coefficient-exponent}
  [z^{j_k}]F_B(z)^{n_k}=2^{n_k(H_B^{\mathrm{coeff}}(\mu)+o(1))}.
\end{equation}
The envelope identity
\(\bigl(H_B^{\mathrm{coeff}}\bigr)'(\mu)=-\log_2w_\mu\) and the symmetry
\(F_B(z)=z^BF_B(1/z)\), which gives
\(H_B^{\mathrm{coeff}}(B-\mu)=H_B^{\mathrm{coeff}}(\mu)\), show that
\(H_B^{\mathrm{coeff}}\) increases on \([0,B/2]\) and decreases on
\([B/2,B]\).
Put \(\rho_B\eqdef w_{1/\alpha}\), and define the capped-space exponent
\begin{equation}
\label{eq:C-B}
  C_B(\alpha)
  \eqdef
  \alpha H_B^{\mathrm{coeff}}(1/\alpha)
  =
  \alpha\log_2F_B(\rho_B)-\log_2\rho_B.
\end{equation}

For the pair exponent, let \(\cP_B(\alpha,\delta)\) be the set of
probability distributions
\(P=(P_{ab})_{0\leq a,b\leq B}\) of a pair
\((U,V)\in\{0,\ldots,B\}^2\) satisfying
\begin{equation}
\label{eq:general-type-constraints}
  \E_P U
  =
  \E_P V
  =
  \frac1\alpha,
  \qquad
  \E_P|U-V|
  \leq
  \frac{2\delta}{\alpha}.
\end{equation}
Set
\begin{equation}
\label{eq:Psi-B}
  \Psi_B(\alpha,\delta)
  \eqdef
  \alpha
  \max_{P\in\cP_B(\alpha,\delta)}
  H_2(P).
\end{equation}

Write \(\Delta_{\{0,\ldots,B\}}\) for the probability simplex on
\(\{0,\ldots,B\}\).  For the refinement by layers, define
\begin{equation}
\label{eq:Pi-B}
  \Pi_B(\alpha)
  \eqdef
  \left\{
    \pi\in\Delta_{\{0,\ldots,B\}}:
    \sum_{a=0}^{B}a\pi_a=\frac1\alpha
  \right\}.
\end{equation}
For \(\pi\in\Pi_B(\alpha)\), let
\(\Theta_B(\pi;\alpha,\delta)\) be the largest entropy \(H_2(P)\)
among joint distributions \(P=(P_{ab})_{0\leq a,b\leq B}\) satisfying
\begin{equation}
\label{eq:Theta-B-constraints}
  \sum_b P_{ab}=\sum_b P_{ba}=\pi_a
  \quad(0\leq a\leq B),
  \qquad
  \sum_{a,b}|a-b|P_{ab}
  \leq
  \frac{2\delta}{\alpha}.
\end{equation}
Set
\begin{equation}
\label{eq:R-cc-B}
  R_B^{\mathrm{cc}}(\alpha,\delta)
  \eqdef
  \max_{\pi\in\Pi_B(\alpha)}
  \alpha
  \left(
    2H_2(\pi)
    -
    \Theta_B(\pi;\alpha,\delta)
  \right).
\end{equation}
Finally, for \(0\leq\delta<\alpha B\), define the moment exponent
\begin{equation}
\label{eq:Lambda-B}
  \Lambda_B(\alpha,\delta)
  \eqdef
  \alpha H_B^{\mathrm{coeff}}(\delta/\alpha),
  \qquad\text{so that}\qquad
  \Lambda_B(\alpha,0)=0.
\end{equation}
Since \(|\cS_{q,N}^{(B)}|=[z^N]F_B(z)^q=m_B(q,N)\), one has
\(C_B(\alpha)=\Lambda_B(\alpha,1)\).

At \((B,\alpha)=(2,1)\) and \(0\leq\delta<1/2\), the general notation
specializes as
\[
  C_2(1)=\log_2 3,
  \qquad
  R_2^{\mathrm{cc}}(1,\delta)=R_{\mathrm{cc}}^{(2)}(\delta),
  \qquad
  \Lambda_2(1,\delta)=\Lambda_2(\delta).
\]
Since \(\Pi_2(1)\) consists of the cap-two types \((\rho,1-2\rho,\rho)\), the
optimality of \(\pi_\delta\) in Theorem~\ref{thm:cap-two-cc}, which uses
\cite[App.~B]{TolhuizenTernary2010}, shows that the maximum in
\eqref{eq:R-cc-B} equals \(R_{\mathrm{cc}}^{(2)}(\delta)\); without that
appendix, evaluation at \(\pi_\delta\) gives
\(R_2^{\mathrm{cc}}(1,\delta)\geq R_{\mathrm{cc}}^{(2)}(\delta)\).  Moreover, the full-cube maximizer
\(P_{ab}\propto u_\delta^{|a-b|}\) described in the proof of
Theorem~\ref{thm:cap-two-cc} already has both marginal means equal to one,
so
\(\Psi_2(1,\delta)=\Psi_{\mathrm{amb}}^{(2)}(\delta)\) for
\(0\leq\delta<4/9\); by Proposition~\ref{prop:general-exponents}, the bound
of Lemma~\ref{lem:cap-two-pairs} is exponentially tight on the balanced
cap-two simplex.

\begin{proposition}[Capped space, pair, and moment exponents]
\label{prop:general-exponents}
For every fixed \(B\), \(\alpha\) with \(\alpha B>1\), and every fixed
\(0\leq\delta<1\),
\begin{align}
  |\cS_{q_N,N}^{(B)}|
  &=
  2^{N(C_B(\alpha)+o(1))},
  \label{eq:general-space-asymptotic}\\
  P_{\cS_{q_N,N}^{(B)}}(t_N)
  &=
  2^{N(\Psi_B(\alpha,\delta)+o(1))},
  \label{eq:general-pair-asymptotic}\\
  m_B(q_N,t_N)
  &=
  2^{N(\Lambda_B(\alpha,\delta)+o(1))}.
  \label{eq:general-moment-asymptotic}
\end{align}
\end{proposition}

The first and third identities are \eqref{eq:coefficient-exponent} with
\(n_k=q_N\) and \(j_k=N\) or \(j_k=t_N\), indexed by \(N\) (for \(\delta=0\), trivially
\(m_B(q,0)=1\)); the first is also
\cite[Sec.~IV, Thm.~13, Eqs.~(102)--(104)]{KovacevicVukobratovic2022}.  For
\eqref{eq:general-pair-asymptotic}, group ordered pairs by their joint type
\(P\) on \(\{0,\ldots,B\}^2\).  A joint \(q\)-type with both means \(N/q\)
and \(\sum_{a,b}|a-b|P_{ab}\leq2t/q\) contributes
\(q!/\prod_{a,b}(qP_{ab})!=2^{qH_2(P)+O(\log q)}\) pairs, and there are at
most \((q+1)^{(B+1)^2}\) joint types.  The continuous maximizer has full
support when \(\delta>0\) and is diagonal when \(\delta=0\), so admissible
\(q_N\)-types approximate it, and the maximum over them tends to
\(\Psi_B/\alpha\).

\begin{proposition}[Constant-composition capped bound]
\label{prop:general-cc}
For every fixed \(B\), \(\alpha\) with \(\alpha B>1\), and
\(0\leq\delta<1\), there are classical codes in
\(\cS_{\lfloor\alpha N\rfloor,N}^{(B)}\) with minimum distance at least
\(\lfloor\delta N\rfloor+1\) and exponent at least
\(R_B^{\mathrm{cc}}(\alpha,\delta)\).  Moreover,
\begin{equation}
\label{eq:cc-dominates-average}
  R_B^{\mathrm{cc}}(\alpha,\delta)
  \geq
  2C_B(\alpha)-\Psi_B(\alpha,\delta).
\end{equation}
\end{proposition}

\begin{proof}
Apply the fixed-composition generalized GV construction
\cite[Sec.~III, Eqs.~(13)--(15)]{SomekhBaruchScarlettGuillen2019} to
alphabet \(\{0,\ldots,B\}\), cost \(|a-b|/2\), and mean
\(\sum_a a\pi_a=1/\alpha\), and then maximize over \(\pi\).  Here
\(\pi_N\to\pi\) are \(q_N\)-types with \(\sum_a a\,q_N\pi_N(a)=N\); by
joint-type counting, the type class has \(2^{q_NH_2(\pi)+o(N)}\) points and
at most \(2^{q_N(\Theta_B(\pi;\alpha,\delta)+o(1))}\) ordered close pairs,
using upper semicontinuity of \(\Theta_B\) in \((\pi,\delta)\).  This gives
\(R_B^{\mathrm{cc}}\).  For \eqref{eq:cc-dominates-average}, evaluate at
\(\pi_a^\star\eqdef\rho_B^a/F_B(\rho_B)\); then
\(\alpha H_2(\pi^\star)=C_B(\alpha)\), and its admissible couplings form a
subset of those defining \(\Psi_B/\alpha\).
\end{proof}

The preceding propositions and the cap-aware lifting now give the
general fixed-cap theorem.

\begin{theorem}[Exact fixed-cap Fock codes]
\label{thm:general-fixed-cap}
Fix \(B\geq1\) and \(\alpha>0\) with \(\alpha B>1\), and define
\begin{equation}
\label{eq:kappa-B}
  k_B(\alpha,\delta)
  \eqdef
  R_B^{\mathrm{cc}}(\alpha,\delta)
  -
  \Lambda_B(\alpha,\delta).
\end{equation}
At \((B,\alpha)=(2,1)\), \(k_B(\alpha,\delta)\) specializes, by
Theorem~\ref{thm:cap-two-cc},
to the one-variable function \(k_2(\delta)\) in
\eqref{eq:kappa-two}.
For every fixed
\(0\leq\delta<1\), there are exact
\(q_N=\lfloor\alpha N\rfloor\)-mode Fock codes of total excitation \(N\),
distance \(\lfloor\delta N\rfloor+1\), and local occupation at most \(B\),
whose dimensions satisfy
\begin{equation}
\label{eq:general-kappa-bound}
  \liminf_{N\to\infty}
  \frac1N\log_2K_N
  \geq
  k_B(\alpha,\delta),
\end{equation}
whenever \(k_B(\alpha,\delta)>0\).
Moreover,
\[
  k_B(\alpha,0)=C_B(\alpha)>0.
\]
Consequently, every \((B,\alpha)\) with \(\alpha B>1\) admits a nonempty interval of
positive relative distances and positive quantum dimension exponents.
\end{theorem}

\begin{proof}
Proposition~\ref{prop:general-cc} gives classical codes in
\(\cS_{q_N,N}^{(B)}\) with exponent at least
\(R_B^{\mathrm{cc}}(\alpha,\delta)\).
Apply Theorem~\ref{thm:cap-aware-aab} and subtract the moment exponent
\eqref{eq:general-moment-asymptotic}.  This proves
\eqref{eq:general-kappa-bound}.

At \(\delta=0\), a feasible joint type must have \(U=V\).  The maximum
of \(2H_2(\pi)-H_2(P)\) over a diagonal coupling is the maximum
entropy of \(\pi\), namely \(C_B(\alpha)/\alpha\).  Thus
\(R_B^{\mathrm{cc}}(\alpha,0)=C_B(\alpha)\), while
\(\Lambda_B(\alpha,0)=0\).  The saddle expression
\(\Lambda_B(\alpha,\delta)\) is right-continuous at zero.  Also,
every feasible coupling obeys
\(\sum_{a\neq b}P_{ab}\leq2\delta/\alpha\).
Compactness of \(\Pi_B(\alpha)\), uniform continuity of entropy, and the
maximum theorem therefore show that
\(\Theta_B(\pi;\alpha,\delta)\) and
\(R_B^{\mathrm{cc}}(\alpha,\delta)\) are right-continuous at zero.
Hence
\(k_B(\alpha,\delta)>0\) on a nonempty right neighborhood of zero.
\end{proof}

The qualitative last assertion of Theorem~\ref{thm:general-fixed-cap}
already follows from the padded dual-rail codes of Remark~\ref{rem:dual-rail}; the
content of the theorem is the quantitative exponent \(k_B\).
Theorem~\ref{thm:general-fixed-cap} uses the moment exponent \(\Lambda_B\),
which is the same for every support.  Remark~\ref{rem:support-sensitive}
permits a smaller active-pattern count or the affine-hull denominator
\(r_t(C)+1\) of the chosen support, equivalently the rank of the matrix in
that remark.  Jointly optimizing the exponential order of that rank with the
classical type is a separate problem and is not used here.

\begin{proposition}[Two-saddle form of the unstratified pair exponent]
\label{prop:two-saddle}
Define
\begin{equation}
\label{eq:Z-B-two-saddle}
  Z_B(r,u)
  \eqdef
  \sum_{a=0}^{B}\sum_{b=0}^{B}
  r^{a+b}u^{|a-b|}.
\end{equation}
Let
\[
  \pi_a^\star
  \eqdef
  \frac{\rho_B^a}{F_B(\rho_B)},
  \qquad
  0\leq a\leq B,
\]
and define
\begin{equation}
\label{eq:delta-typical}
  \delta_{\mathrm{typ}}
  \eqdef
  \frac{\alpha}{2}
  \sum_{a,b=0}^{B}
  |a-b|\pi_a^\star\pi_b^\star.
\end{equation}
For \(0<\delta<\delta_{\mathrm{typ}}\), there is a unique pair
\((r,u)\in(0,\infty)\times(0,1)\) satisfying
\begin{equation}
\label{eq:two-saddle-equations}
  \frac{r\,\partial_r Z_B(r,u)}{Z_B(r,u)}
  =
  \frac2\alpha,
  \qquad
  \frac{u\,\partial_u Z_B(r,u)}{Z_B(r,u)}
  =
  \frac{2\delta}{\alpha}.
\end{equation}
At this pair,
\begin{equation}
\label{eq:Psi-two-saddle}
  \Psi_B(\alpha,\delta)
  =
  \alpha\log_2 Z_B(r,u)
  -
  2\log_2r
  -
  2\delta\log_2u.
\end{equation}
For \(\delta\geq\delta_{\mathrm{typ}}\),
\(\Psi_B(\alpha,\delta)=2C_B(\alpha)\).  The endpoint
\(\delta=0\) is obtained by \(u\downarrow0\) and
\(r\to\sqrt{\rho_B}\).  For \(0<\delta<\delta_{\mathrm{typ}}\),
\begin{equation}
\label{eq:Psi-derivative}
  \frac{\partial}{\partial\delta}
  \Psi_B(\alpha,\delta)
  =
  -2\log_2u.
\end{equation}
\end{proposition}

\begin{proof}
Let
\[
  P_\lambda
  \eqdef
  (1-\lambda)\operatorname{diag}(\pi^\star)
  +
  \lambda\,\pi^\star\otimes\pi^\star,
  \qquad
  \lambda\eqdef\frac{\delta}{\delta_{\mathrm{typ}}}.
\]
For \(0<\delta<\delta_{\mathrm{typ}}\), this is a full-support
feasible law with both means \(1/\alpha\) and distance moment exactly
\(2\delta/\alpha\).  It is therefore a relative-interior feasible
point, and the infinite inward derivative of entropy rules out a
boundary maximizer.  The independent product violates the distance
constraint in this range, so that constraint is active.  The two mean
multipliers agree by
swap symmetry, so the Lagrange equations have the exponential-family
form
\[
  P_{ab}
  =
  \frac{r^{a+b}u^{|a-b|}}{Z_B(r,u)}.
\]
The constraints give \eqref{eq:two-saddle-equations}, and substituting
this distribution into its entropy gives
\eqref{eq:Psi-two-saddle}.  Strict convexity of
\(\log Z_B(e^x,e^y)\), whose sufficient statistics \(a+b\) and \(|a-b|\) are
affinely independent, gives uniqueness.  When the distance constraint becomes inactive, the
maximum-entropy law is the independent product
\(\pi^\star\otimes\pi^\star\), giving \(2C_B(\alpha)\).
The limiting values at zero follow by restricting to the diagonal.
Finally, \eqref{eq:Psi-derivative} is the envelope identity for the
distance multiplier.
\end{proof}

For comparison with the optimized theorem, define the unstratified
cap-aware exponent
\begin{equation}
\label{eq:kappa-B-average}
  k_{B,\mathrm{av}}(\alpha,\delta)
  \eqdef
  2C_B(\alpha)-\Psi_B(\alpha,\delta)
  -
  \Lambda_B(\alpha,\delta).
\end{equation}
Proposition~\ref{prop:general-cc} implies
\(k_B(\alpha,\delta)\geq
k_{B,\mathrm{av}}(\alpha,\delta)\).
The published AAB denominator has exponent
\begin{equation}
\label{eq:Lambda-infty}
  \Lambda_\infty(\alpha,\delta)
  \eqdef
  (\alpha+\delta)
  h_2\left(\frac{\alpha}{\alpha+\delta}\right),
\end{equation}
and therefore gives
\begin{equation}
\label{eq:kappa-B-AAB}
  k_{B,\mathrm{av}}^{\mathrm{AAB}}(\alpha,\delta)
  \eqdef
  2C_B(\alpha)-\Psi_B(\alpha,\delta)
  -\Lambda_\infty(\alpha,\delta).
\end{equation}
Since \(m_B(q,t)\leq|\cS_{q,t}|\), one has
\(\Lambda_B(\alpha,\delta)\leq
\Lambda_\infty(\alpha,\delta)\).

For the last two columns below, define
\[
  \delta_{B,\mathrm{av}}^{\mathrm{AAB},\ast}
  \eqdef
  \sup\{0\leq\delta<1:
  k_{B,\mathrm{av}}^{\mathrm{AAB}}(1,\delta)>0\},
  \qquad
  \delta_{B,\mathrm{av}}^{\mathrm{cap},\ast}
  \eqdef
  \sup\{0\leq\delta<1:
  k_{B,\mathrm{av}}(1,\delta)>0\}.
\]
Table~\ref{tab:fixed-cap-rates} illustrates the convergence to the
uncapped simplex at \(\alpha=1\).  The two \(k\)-columns are dimension
exponents per excitation at \(\delta=0.1\); the last two columns are their
positive-rate endpoints.

\begin{table}[ht]
\centering
\caption{Unstratified fixed-cap benchmarks for \(\alpha=1\).  Column 2 is
the space exponent and columns 3--4 are dimension exponents, in bits per
excitation; columns 5--6 are positive-rate endpoints in relative deletion
distance.  Values are rounded to three
decimal places; equal rounded entries need not have equal exact values.}
\label{tab:fixed-cap-rates}
\begin{tabular}{@{}rrrrrr@{}}
\toprule
\(B\) &
\(C_B(1)\) &
\(k_{B,\mathrm{av}}^{\mathrm{AAB}}(1,0.1)\) &
\(k_{B,\mathrm{av}}(1,0.1)\) &
\(\delta_{B,\mathrm{av}}^{\mathrm{AAB},\ast}\) &
\(\delta_{B,\mathrm{av}}^{\mathrm{cap},\ast}\)\\
\midrule
2  & 1.585 & 0.274 & 0.275 & 0.134 & 0.134\\
3  & 1.852 & 0.517 & 0.517 & 0.166 & 0.166\\
4  & 1.939 & 0.602 & 0.602 & 0.178 & 0.178\\
5  & 1.973 & 0.636 & 0.636 & 0.183 & 0.183\\
6  & 1.987 & 0.652 & 0.652 & 0.185 & 0.185\\
8  & 1.997 & 0.663 & 0.663 & 0.187 & 0.187\\
10 & 1.999 & 0.666 & 0.666 & 0.188 & 0.188\\
\bottomrule
\end{tabular}
\end{table}

With the same average-ball input and no cap, the positive-rate endpoint at
\(\alpha=1\) is approximately \(0.188\) (full-simplex transfer,
Table~\ref{tab:log-cap-thresholds}).  Thus a constant cap of ten already
retains almost all of that positive-rate interval, while a cap of two gives
the smallest nontrivial capped simplex at \(\alpha=1\).  These rows display the
unstratified benchmark, which is available from the two scalar equations
\eqref{eq:two-saddle-equations}; by \eqref{eq:cc-dominates-average} its
exponents are lower bounds for \(k_B(1,\cdot)\) of
Theorem~\ref{thm:general-fixed-cap} and its endpoints for the endpoints of
\(k_B(1,\cdot)\), whose optimization \eqref{eq:R-cc-B}
we do not evaluate for \(B\geq3\).  The optimized cap-two values are the
strictly stronger numbers in Remark~\ref{rem:cap-two-numerics}.  In these
rows the cap-aware denominator changes the rounded values obtained with
the published AAB denominator only for \(B=2\), and then only in the third
decimal of the exponent.

\subsection{The cap-one boundary}
\label{sec:cap-one}

The case \(B=1\) is excluded at \(\alpha=1\), where
the capped simplex consists only of the all-one vector, but it becomes
nondegenerate for every \(\alpha>1\).  Classically this is the binary
constant-weight, or Johnson, space; its Gilbert bound is standard
\cite{GoyalDaoKovacevicKiah2025,GrahamSloane1980}.
Corollary~\ref{cor:cap-one} is the closed-form \(B=1\) benchmark of the
cap-aware lifting; Remark~\ref{rem:dual-rail} compares it with the
dual-rail concatenation of~\cite{DuanGrasslJiZeng2010}.

\begin{corollary}[Cap-one exact Fock codes]
\label{cor:cap-one}
Fix \(\alpha>1\) and
\(0\leq\delta<(\alpha-1)/\alpha\).  There are exact
\(\lfloor\alpha N\rfloor\)-mode Fock codes of total excitation \(N\),
local occupation at most one, and distance
\(\lfloor\delta N\rfloor+1\), with dimension exponent at least
\begin{align}
  k_1(\alpha,\delta)
  &=
  \alpha h_2\left(\frac1\alpha\right)
  -
  h_2(\delta)
  -
  (\alpha-1)
  h_2\left(\frac{\delta}{\alpha-1}\right)
  \notag\\
  &\quad
  -
  \alpha
  h_2\left(\frac{\delta}{\alpha}\right),
  \label{eq:kappa-cap-one}
\end{align}
whenever this expression is positive.  At \(\alpha=2\), its unique
positive zero is
\begin{equation}
\label{eq:cap-one-root}
  \delta_1^\ast(2)
  \approx
  0.151.
\end{equation}
More generally, \(\delta_1^\ast(\alpha)\) denotes the unique zero of
\(k_1(\alpha,\delta)\) in
\((0,(\alpha-1)/\alpha)\).
As \(\alpha\to\infty\), for every fixed \(0\leq\delta<1\)
(which lies in the displayed domain once \(\alpha>1/(1-\delta)\)),
\begin{equation}
\label{eq:cap-one-sparse-limit}
  \frac{k_1(\alpha,\delta)}
       {\alpha h_2(1/\alpha)}
  \longrightarrow
  1-2\delta,
  \qquad
  \delta_1^\ast(\alpha)\longrightarrow\frac12.
\end{equation}
\end{corollary}

\begin{proof}
For \(B=1\), the capped simplex is the binary Johnson space of length
\(q_N=\lfloor\alpha N\rfloor\) and weight \(N\).  Its standard
constant-weight Gilbert bound~\cite{GrahamSloane1980} has exponent
\[
  \alpha h_2(1/\alpha)
  -
  h_2(\delta)
  -
  (\alpha-1)h_2(\delta/(\alpha-1)).
\]
The standard binomial coefficient estimate gives moment exponent
\(\Lambda_1(\alpha,\delta)=\alpha h_2(\delta/\alpha)\).
Subtracting it via Theorem~\ref{thm:cap-aware-aab} gives
\eqref{eq:kappa-cap-one}.  Its value is positive at zero, negative at
the upper endpoint, and direct differentiation gives
\[
  \frac{\partial k_1}{\partial\delta}
  =
  -\log_2
  \frac{(1-\delta)(\alpha-1-\delta)(\alpha-\delta)}
       {\delta^3}
  <0
  \qquad
  \left(0<\delta<\frac{\alpha-1}{\alpha}\right).
\]
Thus it has a unique zero; the computer-assisted enclosure
\(0.1513<\delta_1^\ast(2)<0.1514\), certified by exact rational interval
arithmetic in Appendix~\ref{app:verification}, gives
\eqref{eq:cap-one-root}.

Finally,
\[
  \alpha h_2(1/\alpha)
  =
  \log_2\alpha+\log_2e+o(1),
\]
while the two terms with arguments of order \(1/\alpha\) contribute
\(2\delta\log_2\alpha+O(1)\).  Dividing
\eqref{eq:kappa-cap-one} proves the first limit in
\eqref{eq:cap-one-sparse-limit}.  Evaluating on either side of
\(1/2\), and using monotonicity in \(\delta\), then squeezes the unique
zero to \(1/2\).
\end{proof}

\begin{remark}[Dual-rail codes]
\label{rem:dual-rail}
The dual-rail concatenation of
\cite[Sec.~III, Lem.~1 and Thm.~1]{DuanGrasslJiZeng2010} (see also
\cite{Ouyang2021}) gives positive rate and linear loss distance at cap one
and \(\alpha=2\); the padding below extends it to every nondegenerate fixed
cap.  At \(B=1\) and \(\alpha=2\), each outer qubit is encoded in a rail
pair spanned by \(\lvert1,0\rangle\) and \(\lvert0,1\rangle\).  A loss
empties one pair, which is then identified by its photon number, so at most
\(t\) losses produce at most \(t\) erasures at known positions of the outer
code.  Outer qubit stabilizer codes on the quantum Gilbert--Varshamov bound
\cite[Thm.~2]{CalderbankRainsShorSloane1997} therefore give the dimension
exponent \(1-h_2(\delta)-\delta\log_23\).  The two exponents cross near
\(\delta=0.1211\): the dual-rail exponent exceeds \(k_1(2,\delta)\) for
\(\delta\geq0.1212\) and stays positive up to \(\delta\approx0.189\), beyond
\(\delta_1^\ast(2)\), while the cap-aware lift gives the larger exponent for
\(\delta\leq0.1211\).

For every \(\alpha B>1\), replace the rails by \(\lvert B,0\rangle\) and
\(\lvert0,B\rangle\) and pad.  Choose a constant
\(0<c<\min\{1/B,\alpha-1/B\}\) and use \(a_N\eqdef\lfloor cN\rfloor\) encoding rail
pairs.  Add \(\lfloor N/B\rfloor-a_N\) fully occupied modes, one residual
mode if \(B\nmid N\), and empty modes to reach \(q_N=\lfloor\alpha N\rfloor\).
This uses \(N\) excitations and at most \(B\) per mode.  The photon number of
an encoding pair shows how many losses it suffered; the pair is erased only
when it loses all \(B\) photons, and for \(B\geq2\) a pair that loses fewer
keeps its occupied rail and is only dephased at a known position.  Treating
every damaged pair as an erasure, at most \(t\) losses damage at most \(t\)
outer-code positions, and an asymptotically good outer erasure code on the
\(a_N=\Theta(N)\) positions retains positive rate and corrects a positive
linear number of losses.  The fixed spectator modes carry no logical
information.  With quantum Gilbert--Varshamov outer codes, this erasure
treatment gives the exponent \(c\bigl(1-h_2(\delta/c)-(\delta/c)\log_23\bigr)\),
positive for \(\delta<0.189\,c\), where \(c=a_N/N\uparrow\min\{1/B,\alpha-1/B\}\).
At \(B=1\), where this treatment is exact, the cap-aware lift has a longer
positive-rate range than these padded dual rails when
\(\alpha\leq1.607\) or \(\alpha\geq2.646\).

Single-photon qudit inner codes on more modes can have a longer
range, for example about \(0.255\) at \(\alpha=3\), against
\(\delta_1^\ast(3)\approx0.204\).  At \(B=2\) and \(\alpha=1\) the erasure
treatment gives the padded rails positive rate only for \(\delta<0.0947\), and
the qutrit inner code spanned by \(\lvert2,0\rangle\), \(\lvert1,1\rangle\),
\(\lvert0,2\rangle\)
\cite[arXiv version, Sec.~VIII]{GrasslKongWeiYinZeng2018},
with damaged pairs again treated as erasures and qutrit
Gilbert--Varshamov outer codes \cite{FengMa2004}, the exponent
\(\frac12\bigl(\log_2 3-h_2(2\delta)-6\delta\bigr)\): about \(0.132\) at
\(\delta=0.1\), positive only for \(\delta<0.1276\).  Both treatments are
conservative.  On a rail pair that lost one photon, the Knill--Laflamme
products of the loss operators are combinations of \(I\) and \(Z\), so the
outer stabilizer code need only detect the Pauli operators whose \(Z\)
factors plus twice their \(X\) and \(Y\) factors number at most \(t\).  On the
qutrit code, these products for one loss at a pair are combinations of the
identity and of matrices with zero diagonal in the basis
\((\lvert2,0\rangle\pm\lvert0,2\rangle)/\sqrt2\), \(\lvert1,1\rangle\); in the
Weyl basis built on it, the outer code need only detect the operators whose
nondiagonal factors plus twice their diagonal factors \(Z,Z^2\) number at
most \(t\).  The Gilbert--Varshamov counting for stabilizer codes
\cite[proof of Thm.~2]{CalderbankRainsShorSloane1997}, and its \(q\)-ary
form \cite{FengMa2004}, uses only the number of
operators to be detected, so counting these operators gives
the exponents
\[
  \tfrac12\Bigl(1-\max_{x+2y\leq2\delta}\bigl(H_2(x,y,1-x-y)+y\bigr)\Bigr),
  \qquad
  \tfrac12\Bigl(\log_23-\max_{x+2y\leq2\delta}\bigl(H_2(x,y,1-x-y)+x\log_26+y\bigr)\Bigr),
\]
about \(0.098\) and \(0.171\) at \(\delta=0.1\), against
\(k_2(0.1)\approx0.281\) in Theorem~\ref{thm:exact-cap-two}, and below \(k_2\)
for every \(\delta\leq0.130\) (they meet \(k_2\) near \(0.1341\) and
\(0.1303\)), but positive up
to \(\delta\approx0.139\) and \(\delta\approx0.140\), slightly beyond
\(\delta_2^\ast\approx0.1355\).
\end{remark}

\subsection{Exact finite-length bounds}
\label{sec:finite}

The asymptotic exponents do not show how ball sizes vary with the center at
finite length.  In a capped simplex, the exact average ball size is
accessible through a small generating function; with it, the guaranteed
quantum dimensions in Table~\ref{tab:finite-comparison}
(Section~\ref{sec:converses}), whose caps remove less than a fifth of the
simplex, are up to 20\% larger than with the pair count
of the uncapped simplex and several times larger than with maximum-ball
estimates.  The uncapped deletion-ball and ordered-pair
geometry was developed in
\cite[Lem.~4, Lems.~5 and 8, and Thm.~9]{KreindelEssayagZabokritskiy2026b};
the
quantity denoted there by \(A_{n,q}(m)\), here \(A_{N,q}(m)\), is the ordered-pair distance
distribution.  For a translation-invariant linear code the distance
distribution is recovered from its weight distribution
\cite{MacWilliamsSloane1977}, but the discrete simplex has no such
identification.  Dividing the partial sums of \(A_{N,q}(m)\) by
\(|\cS_{q,N}|\) gives the average ball.
The enumerators and counts below follow
\cite{KreindelEssayagZabokritskiy2026b} in listing \(N\) before \(q\), with the
cap parameter added; the statement below is the routine coordinatewise cap
truncation of the earlier enumerator.

Since distance controls loss confusability~\eqref{rem:loss-confusability},
for two independent number states chosen uniformly from a nonempty
subset \(T\subseteq\cS_{q,N}\), the probability of a common descendant
after at most \(t\) losses is \(P_T(t)/|T|^2\).  For \(T=\cS_{q,N}^{(B)}\),
this is the partial sum of the pair enumerator \(A_{N,q,B}\) of
Proposition~\ref{prop:finite-enumerator} divided by \(|T|^2\).

\begin{proposition}[Capped specialization of the pair enumerator]
\label{prop:finite-enumerator}
For integers \(q,N,B\geq1\) and \(m,t\geq0\), the capped-simplex cardinality is
\begin{equation}
\label{eq:capped-size-inclusion}
  |\cS_{q,N}^{(B)}|
  =
  \sum_{j=0}^{\lfloor N/(B+1)\rfloor}
  (-1)^j
  \binom{q}{j}
  \binom{N-j(B+1)+q-1}{q-1},
\end{equation}
where an inadmissible binomial coefficient is zero.  Define
\begin{equation}
\label{eq:Phi-B}
  \Phi_B(X,Y,Z)
  \eqdef
  \sum_{a=0}^{B}\sum_{b=0}^{B}
  X^aY^bZ^{(a-b)_+}.
\end{equation}
Then the number of ordered pairs at distance exactly \(m\) is
\begin{equation}
\label{eq:finite-shell}
  A_{N,q,B}(m)
  =
  [X^NY^NZ^m]\Phi_B(X,Y,Z)^q,
\end{equation}
and
\begin{equation}
\label{eq:finite-pairs}
  P_{\cS_{q,N}^{(B)}}(t)
  =
  \sum_{m=0}^{t}A_{N,q,B}(m).
\end{equation}
When \(N\leq Bq\), the capped simplex is nonempty and its exact average
radius-\(t\) ball size is
\begin{equation}
\label{eq:finite-average-ball}
  \overline V_{N,q,B}(t)
  =
  \frac{1}{|\cS_{q,N}^{(B)}|}
  \sum_{m=0}^{t}A_{N,q,B}(m).
\end{equation}
When \(B\geq N\), this reduces to the uncapped
\(A_{N,q}(m)\) and average-ball formula of
\cite[Thm.~9]{KreindelEssayagZabokritskiy2026b}.
\end{proposition}

Equation~\eqref{eq:capped-size-inclusion} follows by inclusion--exclusion
over the set of coordinates exceeding \(B\); this classical formula is
recorded, for \(1<B<N\), in \cite[Prop.~4.1]{SongDalfoFiolMoraZhang2026}.
For the remaining identities, mark an ordered coordinate pair
\((a,b)\in\{0,\ldots,B\}^2\) by
\(X^aY^bZ^{(a-b)_+}\), multiply over coordinates, and extract the two
fixed sums.  This is precisely the cap restriction of the generating
function in~\cite[Thm.~9]{KreindelEssayagZabokritskiy2026b}.

\begin{corollary}[Exact finite quantum lower bound]
\label{cor:finite-quantum}
For integers \(q,N,B\geq1\) and \(0\leq t\leq N\), assume \(N\leq Bq\).
Put
\begin{equation}
\label{eq:finite-M}
  L_B^{\mathrm{GV}}(N,q,t)
  \eqdef
  \left\lceil
    \frac{|\cS_{q,N}^{(B)}|^2}
         {P_{\cS_{q,N}^{(B)}}(t)}
  \right\rceil.
\end{equation}
There is a classical code in \(\cS_{q,N}^{(B)}\) of minimum distance at least
\(t+1\) and size at least \(L_B^{\mathrm{GV}}(N,q,t)\); taking a subset gives exactly
that size.  Consequently, there is an exact Fock code on \(q\) modes with
total excitation \(N\), local occupation at most \(B\), and distance
\(t+1\), that is, correcting every loss pattern of total order at most
\(t\), with dimension
\begin{equation}
\label{eq:finite-K}
  K_B(N,q,t)
  \geq
  1+
  \left\lfloor
    \frac{L_B^{\mathrm{GV}}(N,q,t)-1}
         {[z^t](1+z+\cdots+z^B)^q}
  \right\rfloor.
\end{equation}
\end{corollary}

\begin{proof}
Apply Proposition~\ref{prop:restricted-average} with
\(T\eqdef\cS_{q,N}^{(B)}\), using
Proposition~\ref{prop:finite-enumerator} for the exact ordered close-pair count.
For \(q\geq2\), apply Theorem~\ref{thm:cap-aware-aab}.  If \(q=1\),
the capped simplex is the singleton \(\{(N)\}\), both displayed
dimension bounds equal one, and its one-dimensional Fock span is a code.
\end{proof}

\begin{example}[Finite gain from inactive moments]
\label{ex:finite-moments}
For \((q,N,B,t)=(30,30,2,3)\), exact coefficient extraction gives
\[
  |\cS_{30,30}^{(2)}|
  =
  18\,252\,025\,766\,941,
  \qquad
  L_2^{\mathrm{GV}}(30,30,3)
  =
  11\,511\,777.
\]
The published AAB denominator is
\[
  \binom{30+3-1}{30-1}=4960,
\]
while the cap-aware denominator is
\[
  [z^3](1+z+z^2)^{30}=4930.
\]
Thus the guaranteed quantum dimension improves from \(2321\) to \(2336\).
\end{example}

Appendix~\ref{app:verification} records the dynamic program and independent
checks used for these numerical values.

Taking \(B\geq N\) removes the cap.  Corollary~\ref{cor:finite-quantum}
then turns the uncapped pair enumerator of
\cite{KreindelEssayagZabokritskiy2026b} into an explicit quantum
dimension; we record this immediate consequence and compare it with the
Sidon-set construction below.

\begin{corollary}[Immediate quantum consequence of the prior pair enumerator]
\label{cor:uncapped-A-quantum}
For integers \(q\geq2\), \(N\geq1\), and \(0\leq t\leq N\), let
\[
  P_{q-1}(x)
  \eqdef
  \sum_{j=0}^{q-1}\binom{q-1}{j}^{2}x^j,
  \qquad
  \sum_{m\geq0}c_{q,m}x^m
  \eqdef
  \frac{P_{q-1}(x)}{(1-x)^{q-1}},
\]
and put
\begin{equation}
  L_{N,q,t}^{\mathrm{unc}}
  \eqdef
  \left\lceil
    \frac{|\cS_{q,N}|^{2}}
         {\displaystyle
          \sum_{m=0}^{t}
          c_{q,m}\binom{N-m+q-1}{q-1}}
  \right\rceil.
  \label{eq:uncapped-A-classical-size}
\end{equation}
Then there is an exact Fock code of distance \(t+1\) on \(q\) modes and total
excitation \(N\) with dimension
\begin{equation}
\label{eq:uncapped-A-quantum-size}
  K_{N,q,t}
  \geq
  1+
  \left\lfloor
    \frac{L_{N,q,t}^{\mathrm{unc}}-1}
         {\binom{q+t-1}{q-1}}
  \right\rfloor.
\end{equation}
In particular, for fixed \(q\geq2\) and \(t\geq0\),
\begin{equation}
\label{eq:fixed-q-quantum-fraction}
  \liminf_{N\to\infty}
  \frac{K_{N,q,t}}{|\cS_{q,N}|}
  \geq
  \frac{1}{
    \displaystyle
    \binom{q+t-1}{q-1}
    \sum_{m=0}^{t}c_{q,m}},
  \qquad
  \frac{\log K_{N,q,t}}{\log|\cS_{q,N}|}\longrightarrow1.
\end{equation}
\end{corollary}

\begin{proof}
Take \(B\geq N\) in Corollary~\ref{cor:finite-quantum}; the resulting
classical bound is the generalized Gilbert--Varshamov bound
\cite[Thm.~17]{KreindelEssayagZabokritskiy2026b} with the ordered-pair
count of \cite[Thm.~9]{KreindelEssayagZabokritskiy2026b}.
For fixed \(q,t\), divide by \(|\cS_{q,N}|\) and use
\(\binom{N-m+q-1}{q-1}/|\cS_{q,N}|\to1\) for each fixed \(m\leq t\)
\cite[Lem.~12]{KreindelEssayagZabokritskiy2026b}.
\end{proof}

For fixed \(q\) and \(t\geq1\), the Sidon-set construction
\cite[Thm.~5]{AydinAlbertBarg2026} combined with
\cite[Thm.~7]{AydinAlbertBarg2026} already gives
\(\liminf_N K_{N,q,t}/|\cS_{q,N}|\geq1/(|G|\binom{q+t-1}{q-1})\) for every
abelian group \(G\) containing a \(t\)-Sidon set of size \(q\) (all sums of
\(t\) elements, taken as multisets, are distinct).  For a
suitable \(G\) this constant can be much larger than the one in
\eqref{eq:fixed-q-quantum-fraction}: for example, the \(2\)-Sidon set
\(\{0,1,3\}\subseteq\Z_7\) gives \(|G|=7\) when \((q,t)=(3,2)\), whereas
\(\sum_{m\leq2}c_{3,m}=19\).  The value of
\eqref{eq:fixed-q-quantum-fraction} is only that it follows directly from
the exact pair enumerator.

\begin{remark}[Scalar and quantum enumerators]
The scalar enumerator \(A_{N,q}(m)\) is the classical distance distribution
of the whole simplex,
not a Shor--Laflamme quantum weight enumerator
\cite{ShorLaflamme1997,Rains1998}.  By itself it neither records the
factorial loss weights, projector coefficients, or cancellations in
Knill--Laflamme overlaps, nor supplies a quantum MacWilliams transform.
Kraus weights are recorded by the amplitude-damping enumerators
introduced in \cite{OuyangLai2020}, by the general-channel auxiliary
enumerators developed in \cite[arXiv version, Secs.~III--V, esp. Eqs.~(12)--(15) and
Lem.~4]{OuyangLai2021}, and, for loss overlaps, by the loss-overlap kernel
of Subsection~\ref{subsec:loss-overlap}.
\end{remark}

\subsection{A loss-overlap refinement and its fixed-cap recurrence}
\label{subsec:loss-overlap}

The pair enumerator \(A_{N,q,B}(m)\) of
Proposition~\ref{prop:finite-enumerator} counts,
by~\eqref{rem:loss-confusability}, pairs of number states with a common
descendant after equal-order losses, but it ignores the bosonic amplitudes
of those losses.  When \(B\geq2\), pairs at the same deletion distance can
have different Kraus overlaps (Example~\ref{ex:unequal-loss-overlap}).
This subsection records these amplitudes:
Definition~\ref{def:loss-overlap} weights each common descendant by its
squared Kraus overlap, Proposition~\ref{prop:loss-overlap} shows that the
ordered pair joint type determines the weight, and
Theorem~\ref{thm:loss-overlap-dp} computes the loss-overlap enumerator
\(W_{N,q,B}(m,r)\) in polynomial time for fixed cap.  For two number states,
by \eqref{eq:loss-overlap-kraus-identity} the kernel \(\Omega_r\) of
Definition~\ref{def:loss-overlap} is, up to a factor depending only on
\(\gamma\), \(N\) and \(r\), the total of their squared Kraus overlaps, the
matrix elements that enter the Knill--Laflamme conditions, over all pairs of
loss patterns of order \(r\); by \eqref{eq:loss-overlap-support} it is
nonzero exactly when their deletion distance is at most \(r\), and
\(W_{N,q,B}(m,r)\) totals it over the ordered pairs in \(\cS_{q,N}^{(B)}\)
at deletion distance \(m\).
These results are not used in the existence bounds;
Remark~\ref{rem:loss-overlap-aqec-obstruction} explains why approximate
correction needs a different kernel.

\begin{definition}[Loss-overlap kernel]
\label{def:loss-overlap}
For \(\underline{n},\underline{m}\in\cS_{q,N}^{(B)}\) and an integer
\(0\leq r\leq N\), define
\begin{equation}
\label{eq:loss-overlap-definition}
  \Omega_r(\underline{n},\underline{m})
  \eqdef
  \sum_{\substack{
        \underline{z}\in\Z_{\geq0}^{q}\\
        \underline{z}\leq\underline{n},\underline{m}\\
        |\underline{z}|=N-r}}
  \prod_{i=0}^{q-1}
  \binom{n_i}{z_i}\binom{m_i}{z_i}.
\end{equation}
For the pure-loss Kraus operators defined in
Section~\ref{sec:preliminaries}, the operator-level meaning of
\(\Omega_r\) is

\begin{equation}
\label{eq:loss-overlap-kraus-identity}
  \sum_{\underline{e},\underline{f}\in\cS_{q,r}}
  \left|
    \langle\underline{m}\rvert
    (A_{\underline{f}}^{(\gamma)})^\dagger
    A_{\underline{e}}^{(\gamma)}
    \lvert\underline{n}\rangle
  \right|^2
  =
  \gamma^{2r}(1-\gamma)^{2(N-r)}
  \Omega_r(\underline{n},\underline{m}).
\end{equation}
Each nonzero summand corresponds to the common descendant
\(\underline{n}-\underline{e}=\underline{m}-\underline{f}\).  Thus \(\Omega_r\)
weights the common descendants of \(\underline{n}\) and \(\underline{m}\).
\end{definition}

\begin{example}[Equal distance, unequal loss overlap]
\label{ex:unequal-loss-overlap}
In \(\cS_{3,3}^{(2)}\), take
\[
  \underline{n}\eqdef(0,1,2),\qquad
  \underline{m}\eqdef(1,0,2),\qquad
  \underline{m}'\eqdef(0,2,1).
\]
Both pairs have deletion distance one.  Their unique one-loss common
descendants are \((0,0,2)\) and \((0,1,1)\), respectively, but
\[
  \Omega_1(\underline{n},\underline{m})=1,
  \qquad
  \Omega_1(\underline{n},\underline{m}')=4.
\]
Thus the scalar distance, and hence \(A_{N,q,B}(m)\) alone, is
insufficient to recover the bosonic loss weights once \(B\geq2\).
Proposition~\ref{prop:loss-overlap} below shows that the ordered pair
joint type is one sufficient refinement.
\end{example}

For a pair \((\underline{n},\underline{m})\), let
\[
  \mu_{ab}(\underline{n},\underline{m})
  \eqdef
  |\{i:(n_i,m_i)=(a,b)\}|,
  \qquad 0\leq a,b\leq B,
\]
be its ordered pair joint type, and put
\[
  g_{ab}(U)
  \eqdef
  \sum_{c=0}^{\min(a,b)}
  \binom ac\binom bc U^c.
\]

\begin{proposition}[Positivity and coefficient formulas]
\label{prop:loss-overlap}
For \(\underline{n},\underline{m}\in\cS_{q,N}^{(B)}\),
\begin{equation}
\label{eq:loss-overlap-support}
  \Omega_r(\underline{n},\underline{m})>0
  \quad\Longleftrightarrow\quad
  d(\underline{n},\underline{m})\leq r.
\end{equation}
At the first possible loss order \(r=d(\underline{n},\underline{m})\), the
common descendant is unique and
\begin{equation}
\label{eq:loss-overlap-minimum-order}
  \Omega_r(\underline{n},\underline{m})
  =
  \prod_{i=0}^{q-1}
  \binom{n_i}{\min(n_i,m_i)}
  \binom{m_i}{\min(n_i,m_i)}.
\end{equation}
For every \(r\), the pair joint type determines the overlap through
\begin{equation}
\label{eq:loss-overlap-joint-type}
  \Omega_r(\underline{n},\underline{m})
  =
  [U^{N-r}]
  \prod_{a,b=0}^{B}
  g_{ab}(U)^{\mu_{ab}(\underline{n},\underline{m})}.
\end{equation}
Finally, if
\begin{align}
  \Phi_B^{\mathrm{ov}}(X,Y,U,Z)
  &\eqdef
  \sum_{a,b=0}^{B}
  \sum_{c=0}^{\min(a,b)}
  \binom ac\binom bc
  X^aY^bU^cZ^{(a-b)_+},
  \label{eq:loss-overlap-coordinate-gf}\\
  W_{N,q,B}(m,r)
  &\eqdef
  \sum_{\substack{
        \underline{x},\underline{y}\in\cS_{q,N}^{(B)}\\
        d(\underline{x},\underline{y})=m}}
  \Omega_r(\underline{x},\underline{y}),
  \label{eq:aggregate-loss-overlap}
\end{align}
then
\nopagebreak
\begin{equation}
\label{eq:aggregate-loss-overlap-coefficient}
  W_{N,q,B}(m,r)
  =
  [X^NY^NU^{N-r}Z^m]
  \Phi_B^{\mathrm{ov}}(X,Y,U,Z)^q.
\end{equation}
\end{proposition}

Here \(X\) and \(Y\) mark the two input sums, \(U\) marks the common
descendant size, and \(Z\) marks deletion distance.

\begin{proof}
The largest common submultiset of \(\underline{n}\) and \(\underline{m}\)
is their coordinatewise minimum, whose size is
\[
  \sum_i\min(n_i,m_i)
  =
  N-d(\underline{n},\underline{m}).
\]
A common descendant of size \(N-r\) therefore exists exactly when
\(r\geq d(\underline{n},\underline{m})\).  Equality forces the unique largest
common descendant and gives
\eqref{eq:loss-overlap-minimum-order}.  For a coordinate of pair type
\((a,b)\), the choice \(z_i=c\) contributes
\(\binom ac\binom bcU^c\), proving
\eqref{eq:loss-overlap-joint-type}.  Multiplying over all coordinates proves
\eqref{eq:aggregate-loss-overlap-coefficient}.
\end{proof}

The coefficient formula already yields a polynomial algorithm for fixed
cap.  The next statement makes the recurrence and its state space
explicit.

\begin{theorem}[Polynomial fixed-cap loss-overlap recurrence]
\label{thm:loss-overlap-dp}
For integers \(B,q,N\geq1\) satisfying \(N\leq Bq\), and
\(0\leq m,r\leq N\), let
\[
  \mathfrak I_B
  \eqdef
  \{(a,b,c):0\leq a,b\leq B,\ 0\leq c\leq\min(a,b)\},
\]
and
\begin{equation}
\label{eq:loss-overlap-transition-count}
  \tau_B
  \eqdef
  |\mathfrak I_B|
  =
  (B+1)^2+\frac{B(B+1)(2B+1)}6.
\end{equation}
After \(j\) coordinates, let
\(D_j^{\mathrm{ov}}(s_x,s_y,s_c,d_{\mathrm{del}})\) be the
weighted count with the two input sums \(s_x,s_y\), common-descendant sum
\(s_c\), and accumulated positive difference
\(d_{\mathrm{del}}\eqdef\sum_{i=0}^{j-1}(x_i-y_i)_+\).  This last statistic
equals deletion distance when the two input sums agree, in particular
at the terminal state \(s_x=s_y=N\).  Initialize
\(D_0^{\mathrm{ov}}(0,0,0,0)=1\), with all other states zero, and set
\begin{equation}
\label{eq:loss-overlap-dp}
  D_{j+1}^{\mathrm{ov}}(s_x,s_y,s_c,d_{\mathrm{del}})
  =
  \sum_{(a,b,c)\in\mathfrak I_B}
  \binom ac\binom bc
  D_j^{\mathrm{ov}}\bigl(s_x-a,s_y-b,s_c-c,d_{\mathrm{del}}-(a-b)_+\bigr),
\end{equation}
where states with a negative argument are zero, and
states exceeding \((N,N,N-r,m)\) in any coordinate may be discarded.  Then
\begin{equation}
\label{eq:loss-overlap-dp-output}
  W_{N,q,B}(m,r)=D_q^{\mathrm{ov}}(N,N,N-r,m).
\end{equation}
For fixed \(B\), a specified pair \((m,r)\) is therefore computable
using at most
\[
  (N+1)^2(N-r+1)(m+1)
\]
states,
\[
  O\!\left(
    q\tau_B(N+1)^2(N-r+1)(m+1)
  \right)
\]
integer-arithmetic operations, and linear memory in the number of states.
In particular, since \(N\leq Bq\), this is polynomial in \(q\) for
every fixed cap.
\end{theorem}

\begin{proof}
A new
coordinate triple \((a,b,c)\) changes the four state variables exactly as
in \eqref{eq:loss-overlap-dp}, with the required binomial weight.
Also,
\(\tau_B=\sum_{a,b=0}^{B}(\min(a,b)+1)\), which evaluates to
\eqref{eq:loss-overlap-transition-count}.
Induction on \(j\) gives the coefficient of
\(X^{s_x}Y^{s_y}U^{s_c}Z^{d_{\mathrm{del}}}\) in
\((\Phi_B^{\mathrm{ov}})^j\), and
\eqref{eq:loss-overlap-dp-output} follows from
Proposition~\ref{prop:loss-overlap}.  The displayed bounds count the
retained states and the \(\tau_B\) possible transitions from each one.
Each transition weight is at most \(4^B\), so every counter is bounded
by \((\tau_B4^B)^q\).  It therefore has \(O_B(q)\) bits.  For fixed
\(B\), the stated arithmetic bound also gives polynomial bit complexity
in the number \(q\) of modes.
\end{proof}

At \(r=N\), only \(\underline z=\underline 0\) contributes to
\eqref{eq:loss-overlap-definition}, so \(\Omega_N\equiv1\) and
\(W_{N,q,B}(m,N)=A_{N,q,B}(m)\): the loss-overlap enumerator refines the pair
enumerator of Proposition~\ref{prop:finite-enumerator}, and the pair
recurrence used in Appendix~\ref{app:verification} is the restriction of
\eqref{eq:loss-overlap-dp} to \(c=0\).

\begin{corollary}[Cap one]
\label{cor:loss-overlap-cap-one}
If \(B=1\), \(\underline{x},\underline{y}\in\cS_{q,N}^{(1)}\),
\(m\eqdef d(\underline{x},\underline{y})\), and \(0\leq r\leq N\), then
\begin{equation}
\label{eq:loss-overlap-cap-one}
  \Omega_r(\underline{x},\underline{y})
  =
  \mathbf 1_{\{m\leq r\}}
  \binom{N-m}{r-m},
\end{equation}
and hence
\[
  W_{N,q,1}(m,r)
  =
  A_{N,q,1}(m)
  \mathbf 1_{\{m\leq r\}}
  \binom{N-m}{r-m}.
\]
\end{corollary}

\begin{proof}
The inputs are incidence vectors of two \(N\)-subsets whose intersection
has size \(N-m\).  Every binomial weight in
\eqref{eq:loss-overlap-definition} is one, and a descendant is an
\((N-r)\)-subset of that intersection.
\end{proof}

\section{Classical and quantum converse bounds}
\label{sec:converses}

We first bound separated classical supports using capped-ball geometry,
then derive a rank bound for arbitrary capped quantum subspaces and
compare it with coherent-information bounds.

\subsection{Classical Plotkin and sphere-packing bounds}

The next statements concern classical separated supports inside the
capped simplex.  They therefore limit what the lifting of a separated
support can give, but they are not converses for arbitrary quantum codes.

Let \(L_B^{\max}(N,q,D)\) denote the largest size of a code in
\(\cS_{q,N}^{(B)}\) with minimum deletion distance at least \(D\).
The map
\[
  a\longmapsto 1^a0^{B-a},
  \qquad
  0\leq a\leq B,
\]
applied coordinatewise embeds \(\cS_{q,N}^{(B)}\) into the binary
constant-weight space of length \(Bq\) and weight \(N\), and doubles
deletion distance into Hamming distance.  For a capped code of cardinality
\(L\geq2\) and minimum deletion distance \(D\), sum the Hamming distances
of all ordered pairs of distinct embedded codewords column by column: a
column of weight \(w\) contributes \(2w(L-w)\), and the \(Bq\) column
weights add up to \(LN\).  Concavity then gives the Plotkin--Johnson
inequality (compare the nonrecursive bound of \cite{LuoFuVinckChen2003})
\begin{equation}
\label{eq:capped-plotkin}
  D\left(1-\frac1L\right)
  \leq
  N\left(1-\frac{N}{Bq}\right).
\end{equation}
In particular, when \(q\sim\alpha N\), positive rate requires
\[
  \delta\leq1-\frac1{\alpha B}.
\]
For \(B=2,\alpha=1\), this bound and
Theorem~\ref{thm:cap-two-cc} show that the supremum of the relative
distances with positive classical rate is exactly \(1/2\).

Identifying the minimum capped-ball volume gives a sphere-packing bound.
Our earlier uncapped result gives the generating function for an arbitrary
center, both extremal centers, and the exact average ball
\cite[Lem.~4, Lems.~5 and 8, and Thm.~9]{KreindelEssayagZabokritskiy2026b}.
We extend these formulas by truncating the local loss and gain factors;
for the extremal centers we replace the completion counting of
\cite{KreindelEssayagZabokritskiy2026b} by a cap-preserving Robin--Hood injection.  While corner minima
are known for product grids \cite{GarciaClaroGutierrez2024}, the following
majorization theorem applies to the capped simplex, which is not a product
set, and identifies both extremal centers and the generating functions of
their balls.

\begin{theorem}[Majorization, extremal capped balls, and sphere packing]
\label{thm:minimum-capped-ball}
Let \(B,q,N\geq1\) be integers with \(N\leq Bq\).
For \(0\leq j\leq B\), define
\[
  b_j(U,V)
  \eqdef
  1+\sum_{i=1}^{j}U^i+\sum_{i=1}^{B-j}V^i.
\]
Here \(U\) marks units removed from a coordinate and \(V\) marks units
added.
For every integer \(r\geq0\) and every center
\(\underline{n}\in\cS_{q,N}^{(B)}\),
\begin{equation}
\label{eq:center-ball-volume}
  |B_{\cS_{q,N}^{(B)}}(\underline{n},r)|
  =
  \sum_{m=0}^{r}
  [U^mV^m]\prod_{i=0}^{q-1}b_{n_i}(U,V).
\end{equation}
If \(\underline{n}\) majorizes \(\underline{m}\), then for every \(r\geq0\),
\begin{equation}
\label{eq:ball-majorization}
  |B_{\cS_{q,N}^{(B)}}(\underline{n},r)|
  \leq
  |B_{\cS_{q,N}^{(B)}}(\underline{m},r)|.
\end{equation}
Thus ball volume is Schur-concave in the center.

Write
\[
  N=aB+c,
  \qquad
  0\leq c<B,
  \qquad
  \varepsilon\eqdef\mathbf{1}_{\{c>0\}},
  \qquad
  z_0\eqdef q-a-\varepsilon;
\]
inside a displayed center vector, a superscript counts repeated
coordinates; it is not an exponentiation.
Up to coordinate permutation, the saturated point
\begin{equation}
\label{eq:extreme-center}
  \underline{n}^{\mathrm{sat}}\eqdef(B^a,c^\varepsilon,0^{z_0})
\end{equation}
minimizes the ball volume at every radius.  The minimum radius-\(r\) volume is
\begin{equation}
\label{eq:minimum-ball-volume}
  V_{\min}(r)
  =
  \sum_{m=0}^{r}
  [U^mV^m]\,
  F_B(U)^a
  F_B(V)^{z_0}
  b_c(U,V)^\varepsilon.
\end{equation}
Writing \(N=\ell q+s\), \(0\leq s<q\), the balanced point
\[
  \underline{n}^{\mathrm{bal}}
  \eqdef
  ((\ell+1)^s,\ell^{q-s})
\]
maximizes the ball volume at every radius, and
\begin{equation}
\label{eq:maximum-ball-volume}
  V_{\max}(r)
  =
  \sum_{m=0}^{r}
  [U^mV^m]\,
  b_\ell(U,V)^{q-s}b_{\ell+1}(U,V)^s,
\end{equation}
where the second factor is absent when \(s=0\).
Consequently, for every integer \(D\geq1\),
\begin{equation}
\label{eq:capped-sphere-packing}
  L_B^{\max}(N,q,D)
  \leq
  \left\lfloor
    \frac{|\cS_{q,N}^{(B)}|}
         {V_{\min}(\lfloor(D-1)/2\rfloor)}
  \right\rfloor.
\end{equation}
\end{theorem}

When \(B\geq N\), the cap is inactive and this theorem reduces to the
uncapped individual- and extremal-ball results of
\cite[Lem.~4, Lems.~5 and 8]{KreindelEssayagZabokritskiy2026b}.
The proof is given in Appendix~\ref{app:minimum-ball}.  Its key step is a
fiberwise injection showing that every integer Robin--Hood smoothing of
the center can only enlarge its ball.

Theorem~\ref{thm:minimum-capped-ball} also gives the exact maximum ball,
which supplies one more ball-size estimate for achievability.
Table~\ref{tab:finite-comparison} compares four finite ball-size estimates.  The
ambient column uses \(\lceil|\cS_{q,N}^{(B)}|^2/P_{\cS_{q,N}}(t)\rceil\), the
last bound in \eqref{eq:restricted-average}; the analytic-maximum
column uses the ball bound of \cite[Lem.~25]{ElimelechAydinBarg2026} for
centers with a prescribed number of nonzero coordinates, at the largest
possible number \(\min\{q,N\}\); since that bound does not decrease with the
number and each capped ball lies in the uncapped ball with the same center,
it bounds every ball of the capped simplex, and for \(q=N\) it is
\(\sum_{j=0}^{t}\binom{j+q-1}{q-1}^2\); the exact-maximum
column uses \eqref{eq:maximum-ball-volume}; and the exact-average column
uses \eqref{eq:finite-average-ball}.  Each entry is the guaranteed quantum
dimension \(1+\lfloor(L-1)/m_B(q,t)\rfloor\) after the AAB lift, where the
classical size \(L\) is \(\lceil|\cS_{q,N}^{(B)}|/V_{\mathrm{b}}\rceil\) for a
maximum-ball bound \(V_{\mathrm{b}}\), the ratio above for the ambient column, and
\(L_B^{\mathrm{GV}}\) of \eqref{eq:finite-M} for the exact-average column.
In these rows \(t\leq B\), so the published
AAB and cap-aware denominators coincide.

\begin{table}[ht]
\centering
\caption{Guaranteed finite quantum dimensions obtained from four
classical ball-size estimates in the capped simplex.}
\label{tab:finite-comparison}
\begin{tabular}{@{}rrrrrrr@{}}
\toprule
\(q=N\)&\(B\)&\(t\)&
ambient pair&
analytic maximum&
exact maximum&
exact average\\
\midrule
12&5&2&12&3&5&13\\
16&5&2&391&98&143&458\\
16&5&3&5&1&2&6\\
18&6&3&31&3&7&33\\
20&6&3&178&17&36&192\\
20&6&4&3&1&1&3\\
\bottomrule
\end{tabular}
\end{table}

At \((q,N,B,t)=(16,16,5,2)\), for example, the exact maximum-ball
argument guarantees dimension \(143\), whereas the exact average gives
\(458\).

The asymptotic form uses the coefficient exponent
\(H_B^{\mathrm{coeff}}\) of \eqref{eq:H-cap-coefficient}, for which
\(C_B(\alpha)=\alpha H_B^{\mathrm{coeff}}(1/\alpha)\).

\begin{corollary}[Asymptotic capped sphere packing]
\label{cor:asymptotic-sphere-packing}
Fix \(B\geq1\) and \(\alpha>0\) with \(\alpha B>1\).  For
\(0\leq\delta<\min\{1,\alpha B-1\}\), define
\begin{equation}
\label{eq:S-B-sphere}
  S_B(\alpha,\delta)
  \eqdef
  \frac1B
  H_B^{\mathrm{coeff}}\left(\frac{B\delta}{2}\right)
  +
  \left(\alpha-\frac1B\right)
  H_B^{\mathrm{coeff}}\left(
    \frac{\delta}{2(\alpha-1/B)}
  \right).
\end{equation}
Let \(q_N\eqdef\lfloor\alpha N\rfloor\) and
\(D_N\eqdef\lfloor\delta N\rfloor+1\).  Then
\begin{equation}
\label{eq:asymptotic-sphere-packing}
  \limsup_{N\to\infty}
  \frac1N\log_2 L_B^{\max}(N,q_N,D_N)
  \leq
  C_B(\alpha)-S_B(\alpha,\delta).
\end{equation}
In particular, for \(B=2,\alpha=1\),
\begin{equation}
\label{eq:cap-two-sphere-packing}
  \limsup_{N\to\infty}
  \frac1N\log_2 L_2^{\max}(N,N,\lfloor\delta N\rfloor+1)
  \leq
  \log_2 3-\Lambda_2(\delta).
\end{equation}
\end{corollary}

\begin{proof}
In \eqref{eq:minimum-ball-volume}, the saturated point has
\(N/B+O(1)\) saturated coordinates and
\((\alpha-1/B)N+O(1)\) zero coordinates.  A point at deletion distance
\(m\) loses and gains exactly \(m\) units.  By
\eqref{eq:coefficient-exponent}, the two coefficient exponents are
\[
  \frac1B H_B^{\mathrm{coeff}}\left(\frac{Bm}{N}\right)
  \quad\text{and}\quad
  \left(\alpha-\frac1B\right)
  H_B^{\mathrm{coeff}}\left(
    \frac{m/N}{\alpha-1/B}
  \right).
\]
The optional factor \(b_c(U,V)^\varepsilon\) has bounded degree and only
shifts these coefficient indices by \(O(1)\), so it contributes
\(2^{o(N)}\).
Since \(H_B^{\mathrm{coeff}}\) increases on \([0,B/2]\) and the stated
range keeps both arguments in that interval through
\(m=\delta N/2+o(N)\), the outer shell determines the exponential
order of the sum.  Subtracting this minimum-ball exponent from the space
exponent \(C_B(\alpha)\) proves
\eqref{eq:asymptotic-sphere-packing}.  For \(B=2,\alpha=1\), the two
terms in \eqref{eq:S-B-sphere} combine to
\(\Lambda_2(\delta)\), since \(\Lambda_B(\alpha,\delta)=\alpha H_B^{\mathrm{coeff}}(\delta/\alpha)\).
\end{proof}

At \((B,\alpha)=(2,1)\), the cap-aware denominator of
Theorem~\ref{thm:cap-aware-aab} and \eqref{eq:cap-two-sphere-packing} bound the
exponent that this lift can guarantee from any separated support by
\(\log_2 3-2\Lambda_2(\delta)\), about \(0.620\) at \(\delta=0.1\) and
positive only for \(\delta<0.2081\).

\subsection{A converse for arbitrary capped quantum codes}
\label{sec:quantum-converse}

The preceding packing bounds concern separated classical supports.
We now bound an arbitrary quantum subspace of \(\cH_{q,N}^{(B)}\),
including codes with linearly dependent error images.
The argument combines the extremal-monomial derivative bound of
\cite[Thm.~1]{GarciaMarcoKoiranPecatteThomasse2017}
with rank-sensitive quantum packing
\cite[Eq.~(11)]{ChiribellaEtAl2011}.
Annihilation becomes differentiation in the polynomial representation.
We bound the dimension spanned by a codeword's derivatives;
Knill--Laflamme orthogonality of the error-image spans of distinct
logical basis states then gives the quantum dimension bound.
Concentrating the occupation in as few modes as the cap allows gives an
explicit cap-dependent lower bound on this derivative rank.  For qudit
amplitude damping without a fixed excitation number, a noise-adapted
Hamming-type count of the capped loss patterns appears in
\cite[Lem.~2 and Suppl. Mat., Sec.~D.2]{DuttaJainMandayam2026}; the present bound
fixes the excitation number, replaces that count by a derivative rank, and
applies to degenerate codes.

\begin{theorem}[Occupation-cap rank bound]
\label{thm:quantum-cap-converse}
Let \(q,N,B\geq1\) and \(0\leq t\leq N\) be integers with
\(N\leq Bq\).  Suppose a \(K\)-dimensional quantum code
\(Q\subseteq\cH_{q,N}^{(B)}\), \(K\geq1\), exactly corrects
\(\{A_{\underline e}^{(\gamma)}:|\underline e|\leq t\}\)
for some \(0<\gamma<1\).  Write \(N=aB+c\), where
\(a\eqdef\lfloor N/B\rfloor\) and \(0\leq c<B\), and set
\[
  h_B(N,r)\eqdef[z^r]F_B(z)^aF_c(z),\qquad F_0(z)\eqdef1.
\]
Then, for every \(0\leq r\leq t\),
\begin{equation}
\label{eq:quantum-cap-converse}
  K h_B(N,r)\leq
  \dim\cH_{q,N-r}^{(B)}=[z^{N-r}]F_B(z)^q.
\end{equation}
\end{theorem}

\begin{proof}
Remove the nonzero scalar
\(\gamma^{r/2}(1-\gamma)^{(N-r)/2}\) from each order-\(r\)
Kraus operator and denote the resulting operator by \(L_{\underline e}\).
Under the Bargmann identification
\(\lvert\underline n\rangle\mapsto
\prod_i x_i^{n_i}/\sqrt{\prod_i n_i!}\),
it acts as
\(\partial^{\underline e}/\sqrt{\prod_i e_i!}\).
Fix a nonzero \(\psi\in Q\).  It is represented by a homogeneous polynomial \(f\)
of degree \(N\), of degree at most \(B\) in each variable.

Let \(\prod_i x_i^{n_i^\ast}\) be the lexicographically largest
monomial of \(f\).  For every \(\underline e\leq\underline n^\ast\)
of total order \(r\), the leading monomial of
\(\partial^{\underline e}f\) has exponent
\(\underline n^\ast-\underline e\): differentiation preserves the
lexicographic order among the surviving terms.  These leading monomials
are distinct, so the selected derivatives are linearly independent.
This is the extremal-monomial argument of
\cite[Thm.~1]{GarciaMarcoKoiranPecatteThomasse2017}, which applies over
\(\mathbb C\) as well.  It gives
\[
 \dim\operatorname{span}\{L_{\underline e}\psi:|\underline e|=r\}
 \geq [z^r]\prod_i F_{n_i^\ast}(z).
\]
It remains to minimize this coefficient over capped occupation profiles.
For \(1\leq u\leq v<B\),
\[
 F_u(z)F_v(z)-F_{u-1}(z)F_{v+1}(z)=\sum_{j=u}^{v}z^j.
\]
Hence transferring one unit from an occupation \(u\geq1\) to another
occupation \(v\) with \(u\leq v<B\) cannot increase any coefficient of the
product.
Repeated transfers terminate at \(a\) entries equal to \(B\), one
entry \(c\) if \(c>0\), and zeros, that is, at the saturated point
\(\underline n^{\mathrm{sat}}\) of \eqref{eq:extreme-center}.  Termination follows since each
transfer strictly increases the bounded integer sum of squares.
The derivative rank is therefore at least \(h_B(N,r)\).

For an orthonormal basis \(\psi_1,\ldots,\psi_K\) of \(Q\),
the exact Knill--Laflamme conditions give a positive semidefinite matrix
\(\lambda^{(r)}\), indexed by order-\(r\) loss patterns, such that
\[
 \langle L_{\underline e}\psi_i,L_{\underline f}\psi_j\rangle
 =\delta_{ij}\lambda^{(r)}_{\underline e,\underline f}.
\]
The joint error-image Gram matrix is \(I_K\otimes\lambda^{(r)}\);
its rank is \(K\operatorname{rank}\lambda^{(r)}\).
The one-state bound just proved gives
\(\operatorname{rank}\lambda^{(r)}\geq h_B(N,r)\).
Every image lies in \(\cH_{q,N-r}^{(B)}\), proving
\eqref{eq:quantum-cap-converse}.  This uses the actual Gram rank,
as in \cite[Eq.~(11)]{ChiribellaEtAl2011}, without assuming
that the error matrix is invertible.
\end{proof}

The formula includes \(r=0\), which gives the input dimension bound,
and \(r=N\), which forces \(K=1\).  For \(B=1\),
\(h_1(N,r)=\binom Nr\).  The lower bound on the one-state derivative
rank is attained by a saturated occupation monomial; this does not
assert attainability of the quantum dimension bound for \(K>1\).
The hypothesis concerns correction through order \(t\), rather than
exact correction of every event of the full physical loss channel.

\begin{corollary}[Fixed-cap converse exponent]
\label{cor:quantum-cap-exponent}
Fix \(B\geq1\).  Let \(Q_N\) be nonzero exact codes as in
Theorem~\ref{thm:quantum-cap-converse}, with dimensions \(K_N\),
\(q_N/N\to\alpha\geq1/B\) finite, and \(t_N/N\to\delta\in[0,1]\).
Then
\begin{equation}
\label{eq:quantum-cap-exponent}
 \limsup_{N\to\infty}\frac{\log_2K_N}{N}
 \leq
 \min_{0\leq s\leq\delta}
 \left\{
 \alpha H_B^{\mathrm{coeff}}\!\left(\frac{1-s}{\alpha}\right)
 -\frac1B H_B^{\mathrm{coeff}}(Bs)
 \right\}.
\end{equation}
The endpoint conventions are those of
\eqref{eq:H-cap-coefficient}.
\end{corollary}

\begin{proof}
For each \(s\in[0,\delta]\), choose
\(r_N\eqdef\min\{t_N,\lfloor sN\rfloor\}\), so \(r_N/N\to s\).
The coefficient estimate \eqref{eq:coefficient-exponent} gives exponents
\(\alpha H_B^{\mathrm{coeff}}((1-s)/\alpha)\) for the right side of
\eqref{eq:quantum-cap-converse} and
\(B^{-1}H_B^{\mathrm{coeff}}(Bs)\) for \(h_B(N,r_N)\).
Indeed, \(\lfloor N/B\rfloor=N/B+O(1)\), and the factor
\(F_{N\bmod B}\) only contributes bounded coefficient shifts;
at an endpoint take a feasible shift and use continuity.
Taking normalized logarithms proves the bound for every \(s\).
Continuity on the compact interval gives the stated minimum.
\end{proof}

Coherent information gives two further upper bounds, which we derive by a
standard argument.  Identifying the
\(N\)-excitation sector with the symmetric subspace of \(N\) particles
with \(q\) levels turns the loss of \(r\) excitations into a partial trace
over \(r\) particles; for \(2r\leq N\) this channel is degradable, so its
coherent information is concave in the input
\cite[arXiv version, Sec.~7.3]{YardHaydenDevetak2008}.

\begin{proposition}[Coherent-information comparison bounds]
\label{prop:quantum-converse-comparison}
Under the assumptions of Theorem~\ref{thm:quantum-cap-converse},
for \(0\leq r\leq t\) with \(2r\leq N\),
\begin{equation}
\label{eq:symmetric-converse-comparison}
 K\leq
 \frac{\binom{q+N-r-1}{N-r}}{\binom{q+r-1}{r}}.
\end{equation}
If \(N\geq2\) and \(t\geq1\), one also has
\begin{equation}
\label{eq:single-loss-cap-comparison}
 K\leq \frac{[z^{N-1}]F_B(z)^q}{q}.
\end{equation}
\end{proposition}

\begin{proof}
Identify the \(N\)-excitation sector with
\(\operatorname{Sym}^N(\mathbb C^q)\), the symmetric subspace of
\(N\) particles with \(q\) levels.  Splitting off \(r\) particles gives
the isometry
\[
 |\underline n\rangle\longmapsto
 \sum_{\substack{\underline e\leq\underline n\\|\underline e|=r}}
 \sqrt{\frac{\prod_i\binom{n_i}{e_i}}{\binom Nr}}\,
 |\underline e\rangle\otimes|\underline n-\underline e\rangle.
\]
Consequently the particle partial trace \(\Phi_r\), retaining \(N-r\)
particles, has Kraus operators \(L_{\underline e}/\sqrt{\binom Nr}\); it
is pure loss conditioned on exactly \(r\) losses, not the truncated channel
of Section~\ref{sec:random}.
It is exactly correctable on \(Q\).  Its complementary channel
\(\Phi_r^c\) retains \(r\) particles, and is obtained from its output
by tracing a further \(N-2r\) particles.  Thus \(\Phi_r\) is degradable,
meaning that its output can simulate its environment.

Write \(S\) for von Neumann entropy in bits.  Coherent information is
\(I_c(\rho,\Phi_r)\eqdef S(\Phi_r(\rho))-S(\Phi_r^c(\rho))\).
It is concave in \(\rho\) for a degradable channel
\cite[arXiv version, Sec.~7.3]{YardHaydenDevetak2008}.
For the maximally mixed code state \(\pi_Q\), exact correction makes
the environment independent of a purifying reference \(R\).
Purity then gives \(I_c(\pi_Q,\Phi_r)=S(R)=\log_2K\).

On the unrestricted symmetric sector, covariance under mode unitaries
\(U(q)\) and concavity allow Haar averaging of \(\pi_Q\).
Each symmetric power is irreducible, so the averaged input and both
marginals are maximally mixed on their symmetric sectors.
Their dimensions are
\(d_m\eqdef\binom{q+m-1}{m}\), giving
\(\log_2K\leq\log_2d_{N-r}-\log_2d_r\).
This proves \eqref{eq:symmetric-converse-comparison}; the unrestricted
averaged input need not retain the cap.

For \(r=1\), instead average only over independent mode phases and
mode permutations, which preserve the cap.  The one-particle
environment is then \(I_q/q\), whereas the output stays in
\(\cH_{q,N-1}^{(B)}\).  Concavity and
\(S(\Phi_1(\rho))\leq\log_2\dim\cH_{q,N-1}^{(B)}\) prove
\eqref{eq:single-loss-cap-comparison}.
\end{proof}

Since \(h_B(N,1)=\lceil N/B\rceil\leq q\),
\eqref{eq:single-loss-cap-comparison} is at least as strong as the
rank bound for every single-loss case with \(N\geq2\).
For \(q=N=8\), \(B=2\), and \(t=2\), the coefficients in
\eqref{eq:quantum-cap-converse} are \(784\) and \(10\), giving
\(K\leq78\), whereas \eqref{eq:symmetric-converse-comparison} gives
\(K\leq\lfloor1716/36\rfloor=47\).
For \(q=N=20\), \(B=2\), and \(t=2\),
\eqref{eq:quantum-cap-converse} with \(r=2\) gives \(K\leq5\,936\,860\),
whereas \eqref{eq:symmetric-converse-comparison} and
\eqref{eq:single-loss-cap-comparison} give \(K\leq84\,155\,390\) and
\(K\leq18\,199\,284\).  Neither the rank bound nor
\eqref{eq:symmetric-converse-comparison} dominates the other.

Asymptotically, at \((B,\alpha,\delta)=(2,1,0.1)\)
the exponent in \eqref{eq:quantum-cap-exponent} is approximately
\(1.188\), while the cap-free exponent from
\eqref{eq:symmetric-converse-comparison} is approximately \(1.413\); at
\((B,\alpha,\delta)=(3,1,0.1)\) the order is reversed (\(1.449\) versus
\(1.413\)).  At \((B,\alpha,\delta)=(2,1,0.1)\),
Theorem~\ref{thm:exact-cap-two} achieves the exponent \(0.281\), so a
substantial gap remains between the constructions and the converses.  In
terms of positive-rate endpoints, the bracket in \eqref{eq:quantum-cap-exponent} decreases in
\(s\), since by the envelope identity below \eqref{eq:coefficient-exponent}
its derivative is \(\log_2(w_{(1-s)/\alpha}w_{Bs})<0\) when
\(\alpha B>1\), and it vanishes at \(s=1\); hence the minimum is attained at
\(s=\delta\) and stays positive for every \(\delta<1\).  By contrast,
\eqref{eq:symmetric-converse-comparison} with \(r=\lfloor N/2\rfloor\) forces
zero rate for \(\delta\geq1/2\).  At \((B,\alpha)=(2,1)\) none of these bounds
therefore excludes positive rate for \(\delta<1/2\).

Wilde proves a strong converse for the energy-unconstrained quantum
capacity of the pure-loss bosonic channel \cite{Wilde2026}. This concerns
asymptotic transmission over the full channel, whereas the rank bound of
Theorem~\ref{thm:quantum-cap-converse} gives a finite-length dimension bound under a local
occupation cap.

\section{Random capped codes and approximate correction}
\label{sec:random}

This section treats two independent questions: exact sampling of capped
supports (Subsection~\ref{subsec:samplers}), and approximate correction of
one normalized loss channel by a random partition of a separated support
(Subsection~\ref{subsec:random-partition}).

\subsection{Exact samplers for capped supports}
\label{subsec:samplers}

Proposition~\ref{prop:general-cc} obtains classical supports by a counting
argument.  The random-code discussion of
\cite[Sec.~6.2]{ElimelechAydinBarg2026} asks for a sampling distribution with
\(O(1)\) occupation that still gives linear distance and positive rate.  The
capped simplex and each histogram layer admit exact samplers: the sampler
for the capped simplex below is the standard recursive (counting) method of
random generation \cite{FlajoletZimmermannVanCutsem1994}, specialized to
bounded compositions, and the layer sampler is a uniform permutation of a
multiset.  The expurgation estimate that recovers the full average-ball
exponent is the prior argument cited in the proof.  Without expurgation, a
union bound gives only half of that exponent.

\begin{theorem}[Sampling and expurgation for fixed caps]
\label{thm:random-fixed-cap}
Fix \(B\geq1\), \(\alpha B>1\), and
\(0\leq\delta<1\).  Put \(q_N\eqdef\lfloor\alpha N\rfloor\), and
consider all sufficiently large \(N\), so that the capped simplex is nonempty.

\begin{enumerate}
\item A uniform point of \(\cS_{q_N,N}^{(B)}\) can be sampled exactly in
polynomial expected bit complexity.  For every
\[
  0<\ell<
  2C_B(\alpha)-\Psi_B(\alpha,\delta),
\]
sampling \(L_N\eqdef\lfloor2^{\ell N}\rfloor\) independent points and deleting
one endpoint from every remaining close pair leaves, with probability
\(1-o(1)\), a code of size \(L_N(1-o(1))\) and minimum distance at least
\(\lfloor\delta N\rfloor+1\).

\item For each such \(N\), let
\(\underline{v}^{(N)}=(v_0^{(N)},\ldots,v_B^{(N)})\) be an integer histogram
satisfying
\[
  \sum_{a=0}^{B}v_a^{(N)}=q_N,
  \qquad
  \sum_{a=0}^{B}a\,v_a^{(N)}=N,
\]
and put \(\pi_N(a)\eqdef v_a^{(N)}/q_N\).  Choose these histograms so that
\(\pi_N\) approaches an optimizer in \eqref{eq:R-cc-B}.  The corresponding
layer is
\[
  \{\underline{n}\in\cS_{q_N,N}^{(B)}:
  |\{i:n_i=a\}|=v_a^{(N)}\text{ for every }0\leq a\leq B\}.
\]
A uniform point of this layer can be sampled by a random permutation of
its multiset of symbols.  The same expurgation procedure yields every exponent
\[
  0<\ell<R_B^{\mathrm{cc}}(\alpha,\delta).
\]
\end{enumerate}
All sampled and retained points have local occupation at most \(B\).
\end{theorem}

\begin{proof}
Precompute the integers
\[
  D_B^{\mathrm{samp}}(i,s)\eqdef[z^s]F_B(z)^i
\]
from \(D_B^{\mathrm{samp}}(0,0)=1\) and
\(D_B^{\mathrm{samp}}(i,s)=
\sum_{j=0}^{B}D_B^{\mathrm{samp}}(i-1,s-j)\), with
\(D_B^{\mathrm{samp}}(i,s)=0\) outside \(0\leq s\leq Bi\).
After \(i\) coordinates with sum
\(S_i\) have been chosen, select a feasible next symbol
\(j\in\{0,\ldots,B\}\) with exact
conditional probability
\[
  \frac{D_B^{\mathrm{samp}}(q_N-i-1,N-S_i-j)}
       {D_B^{\mathrm{samp}}(q_N-i,N-S_i)}.
\]
This sequentially produces the uniform distribution on
\(\cS_{q_N,N}^{(B)}\).  The table has polynomially many entries, each of
polynomial bit length for fixed \(B\), so exact rational sampling has
polynomial expected bit complexity.  An exact uniform integer in each
required finite range can be obtained from fair random bits by rejection
from the next power-of-two range, with acceptance probability at least
\(1/2\).  Comparing it with cumulative integer counts selects the next
symbol with the stated probability.  The same uniform-integer procedure
implements the random permutation for a fixed histogram.

For either sampler, the expurgation step is that of
the second statement of \cite[Prop.~23]{ElimelechAlbertBargAQEC2026}
together with \cite[Lem.~24]{ElimelechAlbertBargAQEC2026}: its proof
deletes one endpoint of every pair at distance at most \(t_N\), so the
retained code has minimum distance at least \(t_N+1\) (for \(\delta=0\) it
only removes repeated samples).  For uniform sampling from a set \(T\),
\cite[Lem.~24]{ElimelechAlbertBargAQEC2026} identifies the close-pair probability with \(P_T(t_N)/|T|^2\).  On
the whole capped simplex, \eqref{eq:general-space-asymptotic} and
\eqref{eq:general-pair-asymptotic} give
\(P_T(t_N)/|T|^2=2^{-N(2C_B-\Psi_B+o(1))}\).  On the layer with histogram
\(\underline v^{(N)}\), the method of types gives
\(|T|=2^{q_NH_2(\pi_N)+O(\log N)}\) and
\(P_T(t_N)\leq(q_N+1)^{(B+1)^2}2^{q_N\Theta_N}\), where \(\Theta_N\) is the
largest entropy of a coupling with both marginals \(\pi_N\) and
\(\sum_{a,b}|a-b|P_{ab}\leq2t_N/q_N\); upper semicontinuity gives
\(\limsup_N\Theta_N\leq\Theta_B(\bar\pi;\alpha,\delta)\) when
\(\pi_N\to\bar\pi\), the chosen optimizer of \eqref{eq:R-cc-B}.  Substitution in that statement of \cite[Prop.~23]{ElimelechAlbertBargAQEC2026} gives
the two claimed exponents.
\end{proof}

At \(\alpha=1,B=2\), the layer sampler supplies, for every
\(0<\delta<1/2\), a bounded-occupation distribution of classical codewords of
the kind requested in \cite[Sec.~6.2]{ElimelechAydinBarg2026}, with classical
exponent \(R_{\mathrm{cc}}^{(2)}(\delta)\) after expurgation; the sampler for the whole capped simplex does so only
for \(0<\delta<\delta_{\mathrm{typ}}=4/9\), since
\(\Psi_2(1,\delta)=2C_2(1)\) for \(\delta\geq\delta_{\mathrm{typ}}\) by
Proposition~\ref{prop:two-saddle}.
If an independent and identically distributed (i.i.d.) codebook without
expurgation is required, the direct union bound
still gives every exponent below
\(R_{\mathrm{cc}}^{(2)}(\delta)/2\).
The polynomial-time statement is per sample; it is not a
polynomial-time procedure for listing an exponentially large codebook.

\subsection{Approximate correction by random equipartition}
\label{subsec:random-partition}

The exact codes above need a Tverberg partition of the moment points, and no
efficient procedure for that step is given here.  For one normalized loss
channel we instead split a deterministic separated support from
Proposition~\ref{prop:general-cc} into equal blocks uniformly at random and
take the uniform positive superposition on each block.  Separation makes the
complementary channel classical on the logical basis, so approximate
recovery reduces to concentration of probability vectors
(Lemma~\ref{lem:random-equipartition}).  The resulting exponent
(Corollary~\ref{cor:fixed-cap-aqec}) uses the error intensity of the
normalized channel; for loss probability \(\gamma\leq1/(B+1)\), where we
evaluate that intensity, the exponent exceeds
that of Theorem~\ref{thm:general-fixed-cap}, and in every case a uniformly
random partition succeeds with high probability, with an explicit error
exponent.
We first remove Kraus operators that vanish on the code and fix the channel
normalization.

\begin{lemma}[Deleting Kraus operators that vanish on the code]
\label{lem:zero-kraus}
Let \(\Pi_{Q}\) be the projector onto a code space \(Q\), and let
\(\{E_a\}_{a\in\mathcal I}\) be a Kraus representation of a channel from
\(\cH_{\mathrm{in}}\) to \(\cH_{\mathrm{out}}\).  If
\(E_a\Pi_{Q}=0\) for \(a\in\mathcal I_0\subseteq\mathcal I\), then the operators in
\(\mathcal I_0\) may be removed from every exact or approximate
Knill--Laflamme expression restricted to inputs in \(Q\).  The remaining
family is trace preserving on \(Q\).  After adjoining an orthogonal
output summand, it can be completed to a globally trace-preserving Kraus
family by one further operator that also vanishes on \(Q\).
\end{lemma}

\begin{proof}
Put \(S\eqdef\sum_{a\in\mathcal I\setminus\mathcal I_0}E_a^\dagger E_a\).
For \(a\in\mathcal I_0\) and every \(b\in\mathcal I\), one has
\(\Pi_{Q}E_a^\dagger E_b\Pi_{Q}=\Pi_{Q}E_b^\dagger E_a\Pi_{Q}=0\), so all
rows and columns indexed by \(\mathcal I_0\) vanish in the Knill--Laflamme
error matrix and in its approximate variants, the restricted action on the
code is unchanged, and \(\Pi_{Q}S\Pi_{Q}=\Pi_{Q}\), since
\(\sum_{a\in\mathcal I}E_a^\dagger E_a=I\).  Moreover
\(0\leq S\leq I\).  If \(J\) is an isometry from \(\cH_{\mathrm{in}}\) into
an output summand orthogonal to \(\cH_{\mathrm{out}}\), then
\(E_\perp\eqdef J(I-S)^{1/2}\) completes
\(\{E_a:a\notin\mathcal I_0\}\), with the original outputs embedded in the
first summand, to a Kraus family of a channel, and positivity together
with \(\Pi_{Q}(I-S)\Pi_{Q}=0\) gives \(E_\perp\Pi_{Q}=0\).  Thus the
completion introduces no new Knill--Laflamme term on the code.
\end{proof}

Fix \(0<\gamma<1\) and an integer \(0\leq t\leq N\), and make the channel
normalization explicit; the channel below is the truncated channel of
\cite[Eq.~(13) and Props.~3--4]{ElimelechAydinBarg2026}.  Put
\[
  p_{N,t,\gamma}
  \eqdef
  \sum_{r=0}^{t}
  \binom Nr\gamma^r(1-\gamma)^{N-r}.
\]
This is the binomial probability of at most \(t\) losses from an
\(N\)-excitation input.
On the \(N\)-excitation sector of \(q_N\) modes, the normalized channel
truncated to at most \(t\) losses is
\begin{equation}
\label{eq:normalized-truncated-channel}
  \mathcal{N}_{\gamma,\leq t}(\rho)
  \eqdef
  \frac{1}{p_{N,t,\gamma}}
  \sum_{r=0}^{t}
  \sum_{\underline{e}\in\cS_{q_N,r}}
  A_{\underline{e}}^{(\gamma)}\rho
  (A_{\underline{e}}^{(\gamma)})^\dagger.
\end{equation}
The normalization \(p_{N,t,\gamma}\) is independent of the state on the fixed-excitation
sector, so \eqref{eq:normalized-truncated-channel} is trace preserving
there.  Its normalized Kraus operators are
\begin{equation}
\label{eq:normalized-truncated-kraus}
  E_{\underline{e}}
  \eqdef
  \frac{A_{\underline{e}}^{(\gamma)}}
       {\sqrt{p_{N,t,\gamma}}},
  \qquad
  |\underline{e}|\leq t.
\end{equation}

We use the approximate quantum error-correction (AQEC) convention of
\cite{ElimelechAydinBarg2026}.  Let \(L(Q)\) denote the operator
space supported on the code \(Q\), and let
\(\operatorname{id}_{L(Q)}\) be its identity channel.  For states, the
Uhlmann root fidelity is
\(\mathcal F(\rho,\sigma)\eqdef\norm{\sqrt\rho\sqrt\sigma}_1\).  For two
channels with this input space, write
\[
  d_{\mathrm{ch}}(\mathcal A,\mathcal B)
  \eqdef
  \sqrt{1-\mathcal F_e(\mathcal A,\mathcal B)},
\]
where \(\mathcal F_e\) is the worst-case entanglement root fidelity, the
infimum over purified inputs of \(\mathcal F\) between the two outputs.  A
code \(Q\) is \(\varepsilon\)-correcting for \(\mathcal N\) if there is a
recovery channel \(\mathcal D\) such that
\[
  d_{\mathrm{ch}}\!\left(
    \left.\mathcal D\circ\mathcal N\right|_{L(Q)},
    \operatorname{id}_{L(Q)}
  \right)
  \leq\varepsilon.
\]
For a support \(C\) and a trace-preserving Kraus family
\(\{E_e\}_{e=1}^M\) of loss orders at most \(t\) on its Fock span, write
\[
  \kappa_t\eqdef\max_{\underline{n}\in C,\,1\leq e\leq M}
       \langle\underline{n}\rvert E_e^\dagger E_e\lvert\underline{n}\rangle
\]
for the error intensity of
\cite[Definition~13]{ElimelechAlbertBargAQEC2026}, restricted to that
family and support.

General random-subspace bounds relate approximate recovery to the output
purity of a maximally mixed input \cite[Eq.~(30)]{Klesse2007}, and
Haar-random codes approximately correct \(m\) Pauli errors once \(Km\) is
much smaller than the dimension \cite{MaTanWright2025}.  Applied to
the Fock span \(V\) of a separated support of size \(L\), where the outputs
\(E_e\lvert\underline n\rangle\) are orthogonal, Klesse's bound gives average
code entanglement fidelity at least \(1-\sqrt{KM\kappa_t/L}\) for
Haar-random \(K\)-dimensional subcodes of \(V\); the scale
\(KM\kappa_t\ll L\) is therefore known.  The point here is to attain it with
equal-weight partition codes, under the worst-case criterion, with high
probability over a uniform equipartition and with an explicit error
exponent.  For one normalized channel, the diagonal weights form a
probability vector at every support point.  Separation makes the
complementary channel classical on the logical basis, so by the criterion
of B\'eny and Oreshkov \cite{BenyOreshkov2010} the recovery error is
controlled by the largest \(\ell_1\) distance between a block average of
these vectors and their overall average; the next lemma bounds that
distance.

In the lemma, an effective loss Kraus family on the span of a support is a
trace-preserving Kraus family \(\{E_{\underline{e}}\}\) on that span, indexed
by loss patterns, in which each \(E_{\underline{e}}\) is a nonzero multiple of
\(A_{\underline{e}}^{(\gamma)}\) and no two operators are proportional.  It is
obtained by deleting the operators that vanish on the span and merging each
proportionality class into one operator labeled by one of its patterns.
Merging is legitimate: if operators in
one class satisfy \(E_a=c_aE\) on that span, replacing the class by
\(\widetilde E\eqdef(\sum_a|c_a|^2)^{1/2}E\) leaves the restricted channel
unchanged, and Kraus counts refer to the chosen representation.  Merging
matters, for example, for mixtures of pure-loss channels with different
\(\gamma\).  For the family \eqref{eq:normalized-truncated-kraus}, two
distinct non-vanishing patterns are never proportional on the span of a
support, because at a point where both act their outputs are orthogonal;
hence only vanishing patterns are deleted.  A global
representation may also contain the completion from
Lemma~\ref{lem:zero-kraus}; it annihilates the input span and is not
counted.

\begin{lemma}[Random equipartition of a deterministic support]
\label{lem:random-equipartition}
Consider a sequence of classical codes \(C_N\) in discrete simplices, of sizes
\(L_N=K_NT_N\to\infty\), where \(K_N\geq2\), and minimum
distances at least \(t_N+1\).  Suppose a completely positive
trace-preserving (CPTP) channel on
\(\operatorname{span}\{\lvert\underline{n}\rangle:\underline{n}\in C_N\}\)
admits on this input span an effective loss Kraus family
\[
  \{E_{\underline{e}}:\underline{e}\in\cE_N^{\mathrm{eff}}\},
  \qquad
  |\underline{e}|\leq t_N,
\]
and put \(M_N\eqdef|\cE_N^{\mathrm{eff}}|\) and
\(\kappa_N\eqdef\max_{\underline{n}\in C_N,\,\underline{e}\in\cE_N^{\mathrm{eff}}}
\langle\underline{n}\rvert E_{\underline{e}}^\dagger E_{\underline{e}}\lvert\underline{n}\rangle\leq1\),
the error intensity of this family on \(C_N\).  Partition \(C_N\) uniformly at random into
\(K_N\) labeled blocks \(I_0,\ldots,I_{K_N-1}\) of size \(T_N\), and put
\[
  \lvert c_j\rangle
  \eqdef
  \frac1{\sqrt {T_N}}
  \sum_{\underline{n}\in I_j}\lvert\underline{n}\rangle.
\]
If, for some fixed \(\eta>0\) and all sufficiently large \(N\),
\begin{equation}
\label{eq:equipartition-condition}
  K_NM_N\kappa_N\leq L_N^{1-\eta}
\end{equation}
(which holds when \(K_NM_N\leq L_N^{1-\eta}\), since \(\kappa_N\leq1\)),
then, for every fixed \(0<\theta<\eta/2\) and all sufficiently large
\(N\), with probability at least
\begin{equation}
\label{eq:equipartition-success-probability}
  1-
  K_N\exp\!\left(-\frac{L_N^{\eta-2\theta}}{64}\right),
\end{equation}
these states span a code that is \(\varepsilon_N\)-correcting for that
channel, with
\begin{equation}
\label{eq:equipartition-error}
  \varepsilon_N\leq L_N^{-\theta/2}.
\end{equation}
\end{lemma}

\begin{proof}
Suppress the sequence index and write
\[
  \begin{gathered}
  C\eqdef C_N,\qquad t\eqdef t_N,\qquad L\eqdef L_N,\qquad K\eqdef K_N,\\
  T\eqdef T_N,\qquad M\eqdef M_N,\qquad \kappa\eqdef\kappa_N,\qquad
  \cE^{\mathrm{eff}}\eqdef\cE_N^{\mathrm{eff}}.
  \end{gathered}
\]
Fix a bijection
\[
  \iota:\{1,\ldots,M\}\longrightarrow\cE^{\mathrm{eff}},
\]
and abbreviate \(E_e\eqdef E_{\iota(e)}\).  For every effective Kraus operator
and support point, define
\[
  Y_e(\underline{n})
  \eqdef
  \langle\underline{n}\rvert E_e^\dagger E_e\lvert\underline{n}\rangle
  \in[0,1],
  \qquad
  Y(\underline{n})\eqdef(Y_e(\underline{n}))_{e=1}^{M}.
\]
The effective channel is trace preserving on the Fock span of \(C\), so
each \(Y(\underline{n})\) belongs to the probability simplex
\[
  \Delta_M
  \eqdef
  \left\{p\in\mathbb R_{\geq0}^{M}:\sum_{e=1}^{M}p_e=1\right\}.
\]
Put
\[
  \lambda
  \eqdef
  \frac1L\sum_{\underline{n}\in C}Y(\underline{n})
  \in\Delta_M,
  \qquad
  p_j
  \eqdef
  \frac1T\sum_{\underline{n}\in I_j}Y(\underline{n})
  \in\Delta_M.
\]

First we identify the exact channel norm that must be controlled.
Classical separation makes all off-diagonal complementary-channel terms
zero.  Indeed, if
\(\langle\underline{n}\rvert E_e^\dagger E_f\lvert\underline{m}\rangle\neq0\),
then
\[
  \underline{n}-\iota(e)=\underline{m}-\iota(f).
\]
The two loss orders must agree, say
\(r\eqdef|\iota(e)|=|\iota(f)|\leq t\), and therefore
\[
  d(\underline{n},\underline{m})
  =
  \tfrac12\norm{\iota(e)-\iota(f)}_1
  \leq r
  \leq t.
\]
Separation forces \(\underline{n}=\underline{m}\), hence
\(\iota(e)=\iota(f)\), and bijectivity gives \(e=f\).  Consequently, on a logical matrix
\(\rho=(\rho_{j\ell})\), the complementary channel is the
classical--classical map
\begin{equation}
\label{eq:diagonal-complementary-channel}
  \mathcal N^c(\rho)
  =
  \sum_{j=0}^{K-1}\rho_{jj}
  \sum_{e=1}^{M}p_j(e)\lvert e\rangle\langle e\rvert.
\end{equation}
Let
\[
  \mathcal R_\lambda(\rho)
  \eqdef
  \operatorname{Tr}(\rho)
  \sum_{e=1}^{M}\lambda_e\lvert e\rangle\langle e\rvert,
  \qquad
  \mathcal B\eqdef\mathcal N^c-\mathcal R_\lambda.
\]
Then
\begin{equation}
\label{eq:exact-diagonal-diamond}
  \norm{\mathcal B}_\diamond
  =
  \max_{0\leq j<K}\norm{p_j-\lambda}_1.
\end{equation}
To see the upper bound, let \(X\) be an arbitrary operator on a reference
system and the code, and write \(X_{jj}\) for its \(j\)-th logical
diagonal block.  The orthogonality of the environment basis and the
triangle inequality give
\begin{align*}
  \norm{(\operatorname{id}\otimes\mathcal B)(X)}_1
  &=
  \sum_e
  \norm{\sum_j(p_j(e)-\lambda_e)X_{jj}}_1\\
  &\leq
  \max_j\norm{p_j-\lambda}_1
  \sum_j\norm{X_{jj}}_1
  \leq
  \max_j\norm{p_j-\lambda}_1\norm{X}_1.
\end{align*}
The last step is trace-norm contractivity of logical-basis pinching,
which is an average of unitary conjugations.  Equality is attained on a
logical basis state maximizing \(\norm{p_j-\lambda}_1\), proving
\eqref{eq:exact-diagonal-diamond}.  The B\'eny--Oreshkov
complementary-channel criterion
\cite[Thm.~1 and Cor.~2]{BenyOreshkov2010}, followed by the
Fuchs--van de Graaf inequalities~\cite{FuchsVandeGraaf1999}, gives a
recovery error at most
\begin{equation}
\label{eq:recovery-from-diagonal-diamond}
  \sqrt{\frac12\norm{\mathcal B}_\diamond}
  \leq
  \sqrt{\norm{\mathcal B}_\diamond}.
\end{equation}

It remains to concentrate the \(\ell_1\) deviations.  A fixed block is
marginally a uniform size-\(T\) sample without replacement from the
labeled population \(C\).  The finite-population variance identity, the
decomposition
\(\frac1L\sum_{\underline{n}}\norm{Y(\underline{n})-\lambda}_2^2
=\frac1L\sum_{\underline{n}}\norm{Y(\underline{n})}_2^2-\norm{\lambda}_2^2\),
and \(\norm{Y(\underline{n})}_2^2\leq\max_eY_e(\underline{n})\sum_eY_e(\underline{n})\leq\kappa\)
give
\begin{align}
  \E\norm{p_j-\lambda}_2^2
  &=
  \frac{L-T}{T(L-1)}
  \frac1L\sum_{\underline{n}\in C}
  \norm{Y(\underline{n})-\lambda}_2^2
  \leq
  \frac{\kappa}T,
  \label{eq:probability-vector-variance}\\
  \E\norm{p_j-\lambda}_1
  &\leq
  \sqrt{\frac{M\kappa}T}
  \leq
  L^{-\eta/2}.
  \label{eq:probability-vector-mean}
\end{align}
The final inequality uses \(T=L/K\) and
\eqref{eq:equipartition-condition}.

The symmetric function
\(f(I_j)\eqdef\norm{p_j-\lambda}_1\) changes by at most \(2/T\) when one
sampled population label is replaced.  The bounded-differences inequality
for sampling without replacement
\cite[Prop.~2]{SambaleSinulis2022} therefore yields
\begin{equation}
\label{eq:probability-vector-concentration}
  \Pr\{f(I_j)-\E f(I_j)\geq u\}
  \leq
  \exp\!\left(
    -\frac{Tu^2}{16(1-T/L)}
  \right)
  \leq
  \exp\!\left(-\frac{Tu^2}{16}\right).
\end{equation}
Set \(a_L\eqdef L^{-\theta}\).  Since \(\theta<\eta/2\),
\eqref{eq:probability-vector-mean} is at most \(a_L/2\) for all
sufficiently large \(L\).  Moreover, \(M\kappa\geq1\) because each
\(Y(\underline{n})\) lies in \(\Delta_M\), so
\eqref{eq:equipartition-condition} gives \(T=L/K\geq L^\eta\).  Substitution of
\(u\eqdef a_L/2\) into
\eqref{eq:probability-vector-concentration}, followed by a union bound
over the \(K\) blocks, gives
  \eqref{eq:equipartition-success-probability}.  On the resulting good
  event, \eqref{eq:exact-diagonal-diamond} is at most \(a_L\), and
\eqref{eq:recovery-from-diagonal-diamond} gives
\eqref{eq:equipartition-error}.  A global completion from
Lemma~\ref{lem:zero-kraus} merely appends a zero coordinate to
\(Y,p_j,\lambda\) and changes none of these norms.
\end{proof}

To apply Lemma~\ref{lem:random-equipartition} to the normalized channel
\eqref{eq:normalized-truncated-channel}, we count its effective Kraus
operators.  For fixed \(B,\alpha\), put
\begin{equation}
\label{eq:bar-Lambda-B}
  \overline{\Lambda}_B(\alpha,\delta)
  \eqdef
  \max_{0\leq s\leq\delta}\Lambda_B(\alpha,s).
\end{equation}
On a cap-\(B\) support, a loss operator
\(A_{\underline{e}}^{(\gamma)}\) vanishes
whenever some \(e_i>B\); the same deletion is used for qudit amplitude
damping in \cite[Sec.~5.2.2]{ElimelechAlbertBargAQEC2026}.  For
\(q_N=\lfloor\alpha N\rfloor\) and
\(t_N=\lfloor\delta N\rfloor\), the number of loss patterns that can act
nontrivially is therefore at most
\begin{equation}
\label{eq:effective-kraus-count}
  \widetilde M_N^{(B)}
  \eqdef
  \sum_{r=0}^{t_N}[z^r]F_B(z)^{q_N}.
\end{equation}
Since \(\Lambda_B(\alpha,s)=\alpha H_B^{\mathrm{coeff}}(s/\alpha)\), the
monotonicity recorded below \eqref{eq:coefficient-exponent} shows that
\(\Lambda_B(\alpha,s)\) increases in \(s\) up to \(s=\alpha B/2\) and
decreases afterwards.  Therefore
\[
  \overline{\Lambda}_B(\alpha,\delta)
  =
  \Lambda_B\left(
    \alpha,\min\left\{\delta,\frac{\alpha B}{2}\right\}
  \right).
\]
The exponent below is the capped-alphabet specialization of
\cite[Lem.~30, Eq.~(101)]{ElimelechAlbertBargAQEC2026}: take their alphabet
size to be \(B+1\), their block length to be \(q_N\sim\alpha N\), and
convert their per-coordinate exponent to base two per excitation.  Their
saturation point becomes \(s=\alpha B/2\).  Thus
\begin{equation}
\label{eq:effective-kraus-exponent}
  \widetilde M_N^{(B)}
  =
  2^{N(\overline{\Lambda}_B(\alpha,\delta)+o(1))}.
\end{equation}

\begin{corollary}[Fixed-cap approximate-QEC exponent]
\label{cor:fixed-cap-aqec}
Fix \(B\geq1\), \(\alpha B>1\), a loss probability
\(\gamma\in(0,1)\), and
\(0<\delta<1\).  Put
\(q_N\eqdef\lfloor\alpha N\rfloor\), write
\(D(\delta\|\gamma)\eqdef\delta\log_2\frac\delta\gamma+(1-\delta)\log_2\frac{1-\delta}{1-\gamma}\),
and put
\[
  c_\kappa(\gamma,\delta)
  \eqdef
  \begin{cases}
    -\log_2(1-\gamma)-D(\delta\|\gamma)\,\mathbf 1_{\{\gamma>\delta\}},
      &0<\gamma\leq1/(B+1),\\
    0,&\gamma>1/(B+1).
  \end{cases}
\]
For every exponent
\begin{equation}
\label{eq:aqec-conditions}
  0<k<
  R_B^{\mathrm{cc}}(\alpha,\delta)
  -\overline{\Lambda}_B(\alpha,\delta)
  +c_\kappa(\gamma,\delta),
\end{equation}
and every fixed
\[
  0<c<
  \frac14
  \left(
    R_B^{\mathrm{cc}}(\alpha,\delta)
    -k-\overline{\Lambda}_B(\alpha,\delta)+c_\kappa(\gamma,\delta)
  \right),
\]
there are cap-\(B\) Fock codes that are \(\varepsilon_N\)-correcting for the
normalized channel
\(\mathcal{N}_{\gamma,\leq\lfloor\delta N\rfloor}\) in
\eqref{eq:normalized-truncated-channel}, with a uniformly random
equipartition succeeding with probability \(1-o(1)\), such that
\[
  \liminf_{N\to\infty}\frac1N\log_2K_N\geq k,
  \qquad
  \varepsilon_N\leq2^{-cN}.
\]
In particular, deterministic good partitions exist for all sufficiently
large \(N\).

For \(B=2,\alpha=1\) and \(0<\delta<1/2\), condition \eqref{eq:aqec-conditions} becomes
\[
  k<
  R_{\mathrm{cc}}^{(2)}(\delta)-\Lambda_2(\delta)+c_\kappa(\gamma,\delta).
\]
For \(\gamma>1/3\), where only \(\kappa_N\leq1\) is used, its positive-rate endpoint is \(\delta_2^\ast\approx0.1355\);
for \(\gamma\leq1/3\) it is larger, for example about \(0.146\), \(0.158\),
\(0.185\) and \(0.227\) at \(\gamma=0.05\), \(0.1\), \(0.2\) and \(1/3\); the
last two lie below \(\gamma\), so they concern the truncated channel only
(Remark~\ref{rem:random-scope}).
For \(B=1,\alpha=2\) and \(0<\delta<1/2\), it becomes
\[
  k<
  2-2h_2(\delta)-2h_2(\delta/2)+c_\kappa(\gamma,\delta),
\]
which is positive for \(\delta<\delta_1^\ast(2)\approx0.151\) and, for
\(\gamma\leq1/2\), beyond it (about \(0.170\) at \(\gamma=0.1\)).
\end{corollary}

\begin{proof}
Choose a deterministic classical code from
Proposition~\ref{prop:general-cc} with
\[
  L_N
  =
  2^{N(R_B^{\mathrm{cc}}(\alpha,\delta)-o(1))}.
\]
Discarding fewer than \(K_N\) points makes its size divisible by
\(K_N\eqdef\lfloor2^{kN}\rfloor\) without changing the exponent.  By
\eqref{eq:effective-kraus-exponent}, restriction to the span of the
support gives a channel with at most \(\widetilde M_N^{(B)}\) effective
loss Kraus operators, and hence Kraus exponent at most
\(\overline{\Lambda}_B(\alpha,\delta)\).  An optional global completion
annihilates that span and is irrelevant by
Lemma~\ref{lem:zero-kraus}.

On a cap-\(B\) support the normalized
operators \eqref{eq:normalized-truncated-kraus} have diagonal weights
\[
  \langle\underline{n}\rvert E_{\underline{e}}^\dagger E_{\underline{e}}
  \lvert\underline{n}\rangle
  =
  \frac{(1-\gamma)^N}{p_{N,t_N,\gamma}}
  \prod_i\binom{n_i}{e_i}\Bigl(\frac{\gamma}{1-\gamma}\Bigr)^{e_i}.
\]
If \(\gamma\leq1/(B+1)\), each factor is at most one, because
\(\binom{n_i}{e_i}\leq B^{e_i}\) and \(\gamma/(1-\gamma)\leq1/B\), and the
no-loss operator attains equality; thus
\(\kappa_N=(1-\gamma)^N/p_{N,t_N,\gamma}\).  The binomial tail gives
\(N^{-1}\log_2p_{N,\lfloor\delta N\rfloor,\gamma}\to
-D(\delta\|\gamma)\mathbf 1_{\{\gamma>\delta\}}\), so
\(\kappa_N=2^{-N(c_\kappa(\gamma,\delta)+o(1))}\).  For
\(\gamma>1/(B+1)\) we use only \(\kappa_N\leq1\).  At any support point
these weights sum to one, and at most \(\widetilde M_N^{(B)}\) of them are
nonzero, so \(\widetilde M_N^{(B)}\kappa_N\geq1\); with
\eqref{eq:effective-kraus-exponent} this gives
\(c_\kappa(\gamma,\delta)\leq\overline{\Lambda}_B(\alpha,\delta)\), trivially
if \(c_\kappa=0\).  Hence \eqref{eq:aqec-conditions} forces
\(k<R_B^{\mathrm{cc}}(\alpha,\delta)\), as the discarding step above requires.

The strict inequality in
\eqref{eq:aqec-conditions} yields some fixed \(\eta>0\) for which
\[
  K_N\widetilde M_N^{(B)}\kappa_N
  \leq
  L_N^{1-\eta}
\]
after absorbing the \(o(N)\) terms.
Lemma~\ref{lem:random-equipartition}
gives a successful partition with high probability.  Taking \(\eta\)
arbitrarily close to
\[
  \frac{
    R_B^{\mathrm{cc}}(\alpha,\delta)
    -k-\overline{\Lambda}_B(\alpha,\delta)+c_\kappa(\gamma,\delta)
  }{
    R_B^{\mathrm{cc}}(\alpha,\delta)
  }
\]
and then taking \(\theta\) arbitrarily close to \(\eta/2\) in
Lemma~\ref{lem:random-equipartition} gives every error exponent
strictly below the one stated.

For \(B=2,\alpha=1\) and \(0<\delta<1/2\),
\(\overline{\Lambda}_2(1,\delta)=\Lambda_2(\delta)\).  For \(\gamma>1/3\)
the positive-rate condition is therefore the condition
\(R_{\mathrm{cc}}^{(2)}(\delta)-\Lambda_2(\delta)>0\) analyzed in
Theorem~\ref{thm:exact-cap-two}; for \(\gamma\leq1/3\) the positive term
\(c_\kappa\) moves the endpoint beyond \(\delta_2^\ast\).  The quoted
endpoints are floating-point evaluations.

For \(B=1,\alpha=2\) and \(0<\delta<1/2\), the closed formulas are
\[
  R_1^{\mathrm{cc}}(2,\delta)=2-2h_2(\delta),
  \qquad
  \overline{\Lambda}_1(2,\delta)=2h_2(\delta/2)
\]
on this interval, which give the stated condition.
\end{proof}

For \(\gamma>1/(B+1)\) the corollary uses only the bound \(\kappa_N\leq1\),
so \(c_\kappa=0\) and its region lies inside the exact-code existence
region of Theorem~\ref{thm:general-fixed-cap}, since
\(\overline{\Lambda}_B(\alpha,\delta)\geq\Lambda_B(\alpha,\delta)\).  The
threshold \(1/(B+1)\) is a feature of this bound: below it the loss-free
pattern is the most probable one at every support point, and above it the
error intensity is not evaluated here.  For
\(\gamma\leq1/(B+1)\) the exponentially small error intensity of the
normalized channel enlarges it strictly beyond that region.  Indeed, every
\(\pi\in\Pi_B(\alpha)\) has \(\sum_{a,b}|a-b|\pi_a\pi_b\leq B/2\), so the
product coupling is admissible in \eqref{eq:Theta-B-constraints} and
\(R_B^{\mathrm{cc}}(\alpha,\delta)=0\) once \(\delta\geq\alpha B/4\).  Hence
\(\overline{\Lambda}_B=\Lambda_B\) wherever \(k_B(\alpha,\delta)>0\), and there
the condition reads \(k<k_B(\alpha,\delta)+c_\kappa(\gamma,\delta)\) with
\(c_\kappa(\gamma,\delta)\geq-\log_2(1-\min\{\gamma,\delta\})>0\).  The cap-two
endpoints above show the size of the gain.  For \(\gamma\geq\delta\) this
concerns the truncated channel only (Remark~\ref{rem:random-scope}).  In
every case the corollary
also gives the high-probability random-equipartition statement and the
explicit error exponent for one fixed normalized channel.

\begin{remark}[Coordinatewise benchmark]
\label{rem:coordinatewise-benchmark}
The coordinatewise constructions of
\cite[Prop.~2, Prop.~9, and Eq.~(18)]{ElimelechAydinBarg2026}, which control
each diagonal Knill--Laflamme coordinate separately, were extended
to controlled channel families in
\cite[Thm.~21, Eqs.~(73)--(74), and Remark~6]{ElimelechAlbertBargAQEC2026}.
The latter bound retains the error intensity \(\kappa_t\) as well as a
packing probability for independently sampled support points.  To compare
it with Corollary~\ref{cor:fixed-cap-aqec}, use the objects of
Lemma~\ref{lem:random-equipartition} for one support: \(C\) has minimum
distance at least \(t+1\) and size \(L=KT\), with integers \(K,T\geq1\);
\(\{E_e\}_{e=1}^M\) is an effective loss Kraus family of
orders at most \(t\) on its Fock span, to which the count \(M\) and the
intensity \(\kappa_t\) refer; \(Y_e\) are the diagonal weights of the
proof of that lemma; and \(Q\) is spanned by the block states
\(\lvert c_j\rangle\) of a labeled equipartition
\(C=I_0\sqcup\cdots\sqcup I_{K-1}\) into blocks of size \(T\).  The
coordinatewise argument can also be
applied to a uniform equipartition of \(C\): each block is a uniform
size-\(T\) sample without replacement, and a coordinate in \([0,\kappa_t]\)
has sample-mean deviation probability at most
\(2\exp(-2Ta^2/\kappa_t^2)\) at tolerance \(a\).
Indeed, the moment generating function of a sum sampled without replacement
is bounded by that for independent sampling \cite[Thm.~4]{Hoeffding1963}: for
real \(s\), the normalized elementary symmetric mean of the positive numbers
\(\exp(sY_e(\underline{n}))\) is at most the corresponding power of their
arithmetic mean.  Hoeffding's lemma then gives the
stated tail bound.  Using the sufficient coordinate tolerance
\(a\eqdef\varepsilon^2/(2KM)\) from the cited argument and a union bound over the
\(KM\) coordinates gives the fixed-support benchmark
\[
  \Pr\{Q\text{ is }\varepsilon\text{-correcting}\}
  \geq
  1-2KM\exp\!\left(-\frac{L\varepsilon^4}{2K^3M^2\kappa_t^2}\right).
\]
Applying \cite[Thm.~21]{ElimelechAlbertBargAQEC2026} directly would instead
require its packing-probability term.

Replacing \(\kappa_t\) by the coarse bound \(1\), the cap-two benchmark
gives the sufficient positive-rate condition
\[
  R_{\mathrm{cc}}^{(2)}(\delta)-2\Lambda_2(\delta)>0,
\]
whose endpoint is approximately \(0.084\), whereas
Corollary~\ref{cor:fixed-cap-aqec} with \(c_\kappa=0\) needs only
\(R_{\mathrm{cc}}^{(2)}(\delta)-\Lambda_2(\delta)>0\).  For the
family \eqref{eq:normalized-truncated-kraus} and
\(0<\gamma\leq1/(B+1)\), the intensity
\(\kappa_t=(1-\gamma)^N/p_{N,t,\gamma}\) found in the proof of
Corollary~\ref{cor:fixed-cap-aqec} enters both bounds.  The benchmark
then requires
\(k<(R_B^{\mathrm{cc}}-2\overline{\Lambda}_B+2c_\kappa)/3\), whereas
Corollary~\ref{cor:fixed-cap-aqec} requires
\(k<R_B^{\mathrm{cc}}-\overline{\Lambda}_B+c_\kappa\).  Whenever the former
is positive the latter is larger, since their difference
\((2R_B^{\mathrm{cc}}-\overline{\Lambda}_B+c_\kappa)/3\) is then positive.
At \((B,\alpha)=(2,1)\) and \(\gamma=1/3\), for example, the benchmark stays
positive up to \(\delta\approx0.178\) and Corollary~\ref{cor:fixed-cap-aqec}
up to \(\delta\approx0.227\).  There \(\gamma>\delta\), so the truncated
channel conditions on an atypically small number of losses and the
transfer of Remark~\ref{rem:random-scope} does not apply.
\end{remark}

\begin{remark}[Scope of the random-equipartition result]
\label{rem:random-scope}
Corollary~\ref{cor:fixed-cap-aqec} concerns the normalized truncated loss
channels \eqref{eq:normalized-truncated-channel}.  Its \(KM\kappa_t\) condition uses the normalization
of one fixed CPTP channel and is not uniform over arbitrary linear
combinations of an error set.  For that channel-family requirement,
the general theorem of \cite{ElimelechAlbertBargAQEC2026} supplies a
guarantee with its prescribed error family, coefficient constraint,
and intensity parameter.

The prior truncated-to-full-channel transfer
\cite[Thm.~5 and App.~A.2]{ElimelechAydinBarg2026} applies without a new
proof.  In particular, whenever \(\gamma<\delta\), one may take
\[
  \varepsilon_N'
  \eqdef
  \sqrt{
    1-
    \bigl(1-\varepsilon_N^2\bigr)
    p_{N,\lfloor\delta N\rfloor,\gamma}
  },
\]
and the binomial tail implies \(\varepsilon_N'\to0\).
\end{remark}

\begin{remark}[Loss overlaps and the same-Kraus covariance]
\label{rem:loss-overlap-aqec-obstruction}
If a support has minimum distance at least \(t+1\), then
\eqref{eq:loss-overlap-support} gives
\(\Omega_r(\underline{n},\underline{m})=0\) for distinct support points and every
\(r\leq t\).  Thus the same-output off-diagonal overlaps are already
eliminated before the concentration argument begins.

With the indexing and \(M\) from Lemma~\ref{lem:random-equipartition},
suppress the sequence index and define
\[
  \Gamma(\underline{n},\underline{m})
  \eqdef
  \sum_{e=1}^{M}Y_e(\underline{n})Y_e(\underline{m}).
\]
For a uniformly random size-\(T\) sample \(I\) from a support \(C\) of
size \(L\), the finite-population variance identity is
\begin{align}
\label{eq:Gamma-finite-population-variance}
  \E\norm{
    \frac1T\sum_{\underline{n}\in I}Y(\underline{n})-\lambda
  }_2^2
  ={}&
  \frac{L-T}{T(L-1)}
  \left(
    \frac1L\sum_{\underline{n}\in C}\Gamma(\underline{n},\underline{n})
    -
    \frac1{L^2}
    \sum_{\underline{n},\underline{m}\in C}\Gamma(\underline{n},\underline{m})
  \right).
\end{align}
The kernel \(\Omega_r\) matches loss patterns having the same output
descendant, whereas \(\Gamma\) matches the same Kraus label at two support
points.  Consequently, the bracket in
\eqref{eq:Gamma-finite-population-variance} controls the
normalized-channel approximate-QEC region.  It is at most
\(\frac1L\sum_{\underline{n}\in C}\Gamma(\underline{n},\underline{n})\leq\kappa_t\),
which is the intensity refinement used in
Lemma~\ref{lem:random-equipartition}; a covariance-sensitive bound on the
bracket (compare \cite[Remark~6]{ElimelechAlbertBargAQEC2026}) could improve
the region further, and no such improvement is claimed here.
\end{remark}

\section{Growing caps without exponent loss}
\label{sec:log-caps}

The logarithmically capped codes of \cite{ElimelechAydinBarg2026} use a
maximum-ball classical bound; as noted in
\cite[Sec.~6.1]{ElimelechAydinBarg2026}, the average-ball bound of
\cite{GoyalDaoKovacevicKiah2025} does not by itself give supports with a
bounded \(\ell_\infty\) norm.
Corollary~\ref{cor:subexponential-restriction} shows that the stronger
average-ball and Marcus--Roth exponents of \cite{GoyalDaoKovacevicKiah2025}
survive every occupation cap \(B_N\) that tends to infinity with \(N\); for
the explicit caps of order \(\log N\) used below, the typical-set estimates
of \cite{ElimelechAydinBarg2026} and a tail bound show moreover that the
capped part keeps a fraction \(1-o(1)\) of the simplex or of the slice.
This regime serves as a comparison for the fixed-cap results.  We first
record the density consequence of the typical-set estimates.

Fix \(\alpha>0\) and put \(q_N\eqdef\lfloor\alpha N\rfloor\).
For \(\varepsilon>0\), define
\begin{align}
  B_N^{\mathrm{typ}}
  &\eqdef
  \left\lceil
    (1+\varepsilon)\log_{1+\alpha}N
  \right\rceil,
  \label{eq:typical-cap}\\
  \cS_{q_N,N}^{\mathrm{typ}}
  &\eqdef
  \left\{
    \underline{n}\in\cS_{q_N,N}:
    \norm{\underline{n}}_\infty\leq B_N^{\mathrm{typ}},
    \left|
      |\supp\underline{n}|
      -
      \frac{\alpha}{1+\alpha}N
    \right|
    \leq N^{2/3}
  \right\}.
  \label{eq:typical-set}
\end{align}

\begin{proposition}[Prior typical-density estimate]
\label{prop:typical-density}
By \cite[Lems.~23--24 and part~1 of the proof of
Prop.~15]{ElimelechAydinBarg2026}, for every fixed
\(\alpha,\varepsilon>0\),
\[
  \frac{|\cS_{q_N,N}^{\mathrm{typ}}|}
       {|\cS_{q_N,N}|}
  \longrightarrow1.
\]
\end{proposition}

To state the consequence without reproducing the multivariate formulas of
\cite{GoyalDaoKovacevicKiah2025}, let
\(R_{\mathrm{G}}(\alpha,\sigma)\) denote their classical
average-ball exponent per excitation in the simplex
\(\cS_{\lfloor\alpha N\rfloor,N}\), when the required Manhattan (ordinary
\(\ell_1\)) distance is \(\sigma N\).  Let
\(R_{\mathrm{MR}}(\alpha,\sigma)\) denote their stronger
Marcus--Roth (MR)-type exponent \cite{MarcusRoth1992}: the maximum, over the
density \(\tau\) of zero coordinates, of twice the space exponent of the slice
with \(\lfloor\tau N\rfloor\) zero coordinates minus its pair exponent at
Manhattan distance \(\sigma N\).
More precisely, in the notation of the arXiv version of
\cite{GoyalDaoKovacevicKiah2025}, \(R_{\mathrm{G}}\) is the right-hand
side of their standard-simplex Theorem~10, and
\(R_{\mathrm{MR}}\) is the right-hand side of their Theorem~16, with
\(\rho=\alpha\) and their distance parameter set to \(\sigma\); the two slice
exponents are their Propositions~11 and~15, and their Theorem~16 gives the
maximizing \(\tau\) parametrically.
Equivalently, \(R_{\mathrm G}(\alpha,2\delta)=2C_\infty(\alpha)-\Psi_\infty(\alpha,\delta)\),
where \(\Psi_\infty\) is given by \eqref{eq:two-saddle-equations} and \eqref{eq:Psi-two-saddle} with
\(Z_\infty(r,u)\eqdef(1+ru)/((1-r^2)(1-ru))\) in place of \(Z_B\); this is the
saddle-point form of \cite[arXiv version, Eqs.~(22)--(24)]{GoyalDaoKovacevicKiah2025},
since the denominator \((1-x^2)(1-xz)-y(1+xz)\) of their pair generating
function vanishes exactly when \(y=1/Z_\infty(x,z)\).
This also fixes the normalization used in Table~\ref{tab:log-cap-thresholds}.
Our deletion distance is half the ordinary \(\ell_1\) distance, so
\(\sigma=2\delta\).

In this regime the cap-aware denominator has the same exponential order as
the published AAB denominator.
Indeed, let \(B_N\to\infty\), and put
\(q_N\eqdef\lfloor\alpha N\rfloor\) and
\(t_N\eqdef\lfloor\delta N\rfloor\).  The case \(\delta=0\) is immediate.
For \(\delta>0\), fix any integer
\(b>\max\{\delta/\alpha,1/\alpha\}\).  Eventually,
\[
  m_b(q_N,t_N)
  \leq
  m_{B_N}(q_N,t_N)
  \leq
  |\cS_{q_N,t_N}|.
\]
The fixed-cap coefficient asymptotic in
\eqref{eq:general-moment-asymptotic} gives the lower exponential bound,
whereas the stars-and-bars count gives the upper bound
\(\Lambda_\infty(\alpha,\delta)\) of \eqref{eq:Lambda-infty}.  Moreover,
\(\Lambda_b(\alpha,\delta)\uparrow
\Lambda_\infty(\alpha,\delta)\) as \(b\to\infty\): this follows directly
from the saddle formula \eqref{eq:H-cap-coefficient} because \(F_b\) and
\(F_b'\) converge uniformly on
compact subsets of \((0,1)\) to \((1-z)^{-1}\) and \((1-z)^{-2}\), and the
saddle converges to \(\delta/(\alpha+\delta)\).  Hence
\begin{equation}
\label{eq:growing-cap-moment-exponent}
  \frac1N\log_2m_{B_N}(q_N,t_N)
  \longrightarrow
  \Lambda_\infty(\alpha,\delta).
\end{equation}
Thus the cap-aware and published AAB denominators have the same exponent
whenever the cap tends to infinity.  Since
\(|\cS_{q,N}^{(B)}|=m_B(q,N)\), the same argument with \(t_N=N\), using
\eqref{eq:general-space-asymptotic} for the lower bound and
\(\Lambda_\infty(\alpha,1)=C_\infty(\alpha)\) from \eqref{eq:full-simplex-exponent},
gives \(N^{-1}\log_2|\cS_{q_N,N}^{(B_N)}|\to C_\infty(\alpha)\): for every
\(B_N\to\infty\) the capped simplex has relative density \(2^{-o(N)}\) in
\(\cS_{q_N,N}\).  These limits hold along each fixed sequence \((B_N)\), not
uniformly over all sequences.

\begin{corollary}[Growing-cap restriction of the prior full-simplex bound]
\label{cor:goyal-transfer}
Fix \(\alpha>0\), a relative deletion distance \(0\leq\delta<1\), and an
integer sequence \(B_N\to\infty\).  There are classical codes contained in
\(\cS_{q_N,N}^{(B_N)}\), of minimum deletion distance at least
\(\lfloor\delta N\rfloor+1\), with exponent at least
\[
  R_{\mathrm{G}}(\alpha,2\delta).
\]
After the cap-aware AAB lifting, there are exact Fock codes of distance
\(\lfloor\delta N\rfloor+1\) with local occupation at most
\(B_N\) and dimension exponent at least
\begin{equation}
\label{eq:goyal-quantum-transfer}
  R_{\mathrm{G}}(\alpha,2\delta)
  -
  \Lambda_\infty(\alpha,\delta),
\end{equation}
whenever this quantity is positive.  If \(B_N\geq B_N^{\mathrm{typ}}\) for a
fixed \(\varepsilon>0\), the capped simplex keeps a fraction \(1-o(1)\) of
\(\cS_{q_N,N}\).
\end{corollary}

\begin{proof}
By the density statement after \eqref{eq:growing-cap-moment-exponent},
\(\log_2(|\cS_{q_N,N}|/|\cS_{q_N,N}^{(B_N)}|)=o(N)\), so
Corollary~\ref{cor:subexponential-restriction} transfers the ambient
average-ball exponent to \(\cS_{q_N,N}^{(B_N)}\).  The factor of two
comes from Definition~\ref{def:simplex}.  Finally,
Theorem~\ref{thm:cap-aware-aab} and
\eqref{eq:growing-cap-moment-exponent} subtract the exponent
\(\Lambda_\infty(\alpha,\delta)\) in \eqref{eq:Lambda-infty}.  If
\(B_N\geq B_N^{\mathrm{typ}}\), then
\(\cS_{q_N,N}^{\mathrm{typ}}\subseteq\cS_{q_N,N}^{(B_N)}\), and
Proposition~\ref{prop:typical-density} gives the fraction \(1-o(1)\).
\end{proof}

The use of the full-simplex average-ball rate of \cite{GoyalDaoKovacevicKiah2025} for quantum
partition codes is
already present in \cite[Sec.~5.5.2, Eqs.~(84)--(85), and
App.~A.3.1]{ElimelechAlbertBargAQEC2026}; there an occupation bound
\(O(\ln N)\) is shown only for the random partition codes, which are
approximate (AQEC) codes, and only in the sense that, with high probability,
every codeword has all but a vanishing part of its weight on number states of
occupation \(O(\ln N)\)
\cite[Def.~15, Thm.~22, and Sec.~5.5]{ElimelechAlbertBargAQEC2026}; the exact cap
\((1+\varepsilon)\log_{1+\alpha}N\) with the weaker maximum-ball exponent is due
to \cite[Prop.~15 and Thm.~12]{ElimelechAydinBarg2026}.  Here the rate is
evaluated in the binary-logarithm normalization of
\cite{GoyalDaoKovacevicKiah2025}.  Corollary~\ref{cor:goyal-transfer} adds that every
support point may be restricted to any cap \(B_N\to\infty\), for example the explicit cap \(B_N^{\mathrm{typ}}=O(\log N)\), without
loss of the average-ball exponent.

The same argument can be applied before optimizing over the zero-count
slices.  Suppose a slice has \(s_N=\beta N+o(N)\) nonzero coordinates,
where \(0<\beta<1\).
Conditional on the set of nonzero coordinates, the positive occupations are uniform positive
compositions of \(N\) into \(s_N\) parts.  Those with all parts at most
\(B_N\) number \(m_{B_N-1}(s_N,N-s_N)\), against
\(|\cS_{s_N,N-s_N}|=\binom{N-1}{s_N-1}\) in all, so the argument for
\eqref{eq:growing-cap-moment-exponent}, with \((q_N,t_N,B_N)\) replaced by
\((s_N,N-s_N,B_N-1)\), shows that for every \(B_N\to\infty\) they form a
fraction \(2^{-o(N)}\) of the slice.  For explicit caps the fraction is
\(1-o(1)\): for one positive part \(Y_1\),
\begin{equation}
\label{eq:positive-composition-tail}
  \Pr(Y_1\geq b)
  =
  \frac{\binom{N-b}{s_N-1}}
       {\binom{N-1}{s_N-1}}.
\end{equation}
Consequently,
\begin{equation}
\label{eq:slice-max-tail}
  \Pr\left(\max_i Y_i\geq b\right)
  \leq
  s_N
  \left(\frac{N-s_N}{N-1}\right)^{b-1}.
\end{equation}
A cap
\[
  B_N^{\mathrm{MR}}
  =
  (1+\varepsilon)
  \log_{1/(1-\beta)}N+O(1)
\]
therefore removes only an \(o(1)\) fraction of that slice.

\begin{corollary}[Zero-count-slice transfer]
\label{cor:mr-transfer}
Fix \(\alpha>0\), \(0\leq\delta<1\), and an integer sequence
\(B_N\to\infty\).  Suppose an optimizing zero-count
slice in the Marcus--Roth refinement has \(s_N=\beta N+o(N)\) nonzero
coordinates for a fixed \(0<\beta<1\).  That slice may be intersected with
the occupation cap \(B_N\) without changing its exponent.  The resulting
exact quantum codes of distance \(\lfloor\delta N\rfloor+1\), with local
occupation at most \(B_N\), have dimension
exponent at least
\begin{equation}
\label{eq:mr-quantum-transfer}
  R_{\mathrm{MR}}(\alpha,2\delta)
  -
  \Lambda_\infty(\alpha,\delta),
\end{equation}
whenever the right-hand side is positive.  If \(B_N\geq B_N^{\mathrm{MR}}\)
for a fixed \(\varepsilon>0\), the intersection keeps a fraction \(1-o(1)\)
of the slice.
\end{corollary}

\begin{proof}
For every zero-count slice of fixed density used in the
optimization of~\cite{GoyalDaoKovacevicKiah2025}, the capped portion has
relative density \(2^{-o(N)}\) by the slice count above, so
Corollary~\ref{cor:subexponential-restriction} preserves the classical
slice exponent.  Optimize over the slice parameter and apply
Theorem~\ref{thm:cap-aware-aab} together with
\eqref{eq:growing-cap-moment-exponent}.  If \(B_N\geq B_N^{\mathrm{MR}}\),
\eqref{eq:slice-max-tail} gives the fraction \(1-o(1)\).
\end{proof}

The interior-slice hypothesis holds for every \(0\leq\delta<1\).  In the parameterization of
\cite[arXiv version, Thm.~16]{GoyalDaoKovacevicKiah2025}, with \(\rho=\alpha\),
\[
  \tau_{\mathrm{opt}}
  =
  \frac{\alpha^2}{1+\alpha-\lambda},
  \qquad
  2\delta
  =
  \frac{2\lambda\alpha}{1-\lambda^2+\lambda\alpha}.
\]
For \(0<2\delta<2\), the corresponding parameter satisfies
\(0<\lambda<1\), and the nonzero-coordinate density is
\[
  \beta
  =
  \alpha-\tau_{\mathrm{opt}}
  =
  \frac{\alpha(1-\lambda)}{1+\alpha-\lambda}
  \in(0,1).
\]
Thus every zero-count-slice endpoint in Table~\ref{tab:log-cap-thresholds}
satisfies the hypothesis of Corollary~\ref{cor:mr-transfer}.

Table~\ref{tab:log-cap-thresholds} compares the positive-rate
endpoints.  The second column gives the endpoints of the maximum-ball bound of
\cite[Prop.~15 and Thm.~12]{ElimelechAydinBarg2026}, whose dimension exponent is
\(C_\infty(\alpha)-\bigl(\tfrac{\alpha}{1+\alpha}+\delta\bigr)
h_2\bigl(\tfrac{\alpha}{\alpha+(1+\alpha)\delta}\bigr)-2\Lambda_\infty(\alpha,\delta)\);
the next two give the endpoints of
the full-simplex and zero-count-slice transfers above.  All values use the half-\(\ell_1\)
normalization of this paper.  The prior construction uses the cap
\(B_N^{\mathrm{typ}}\) of \eqref{eq:typical-cap}, and the two transfers hold for
every cap \(B_N\to\infty\), for instance \(B_N^{\mathrm{typ}}\) and the
slice-dependent cap \(B_N^{\mathrm{MR}}\) above.  If a common cap is desired, take
\(B_N^{\mathrm{common}}\eqdef
\max\{B_N^{\mathrm{typ}},B_N^{\mathrm{MR}}\}=O(\log N)\); all three
exponents are then
unchanged.  The table compares endpoints, not the smallest cap constants.

\begin{table}[ht]
\centering
\caption{Positive-rate endpoints in relative deletion distance for exact Fock
codes with growing occupation caps: the prior bound uses a cap of order
\(\log N\), and the two transfers hold for every cap \(B_N\to\infty\); all
values use the per-excitation normalization and are rounded to three decimal
places.}
\label{tab:log-cap-thresholds}
\begin{tabular}{@{}rrrr@{}}
\toprule
\(\alpha\)&
prior maximum-ball bound&
full-simplex transfer&
zero-count-slice transfer\\
\midrule
0.25&0.104&0.116&0.119\\
0.50&0.139&0.150&0.154\\
1.00&0.177&0.188&0.192\\
2.00&0.216&0.227&0.230\\
4.00&0.253&0.265&0.267\\
\bottomrule
\end{tabular}
\end{table}

At \((\alpha,\delta)=(1,0.15)\), the corresponding quantum dimension
exponents per excitation are approximately \(0.209\), \(0.259\),
and \(0.275\), respectively.

\section{Conclusion and open problems}
\label{sec:conclusion}

For every fixed occupation cap \(B\) and \(\alpha B>1\),
Theorem~\ref{thm:general-fixed-cap} gives explicit exponents for exact
Fock-state codes that correct a linear number of losses.  At
\((B,\alpha)=(2,1)\), Tolhuizen's classical input gives a positive
dimension exponent for every \(\delta<\delta_2^\ast\), with
\(\delta_2^\ast\approx0.1355\), and about \(0.281\) at \(\delta=0.1\)
(Theorem~\ref{thm:exact-cap-two}).  For every \(\delta\leq0.130\) it exceeds the
exponents of the padded dual-rail and qutrit-concatenated codes of
Remark~\ref{rem:dual-rail} with outer codes adapted to their located errors,
about \(0.098\) and \(0.171\) at \(\delta=0.1\), but these codes stay positive slightly
beyond \(\delta_2^\ast\), up to \(\delta\approx0.139\) and \(0.140\); at cap
one and \(q=2N\) the dual-rail codes are better for
\(\delta\geq0.1212\).  Against the qutrit rate of
\cite[Eq.~(82)]{ElimelechAlbertBargAQEC2026}, about \(0.275\) at \(\delta=0.1\), the gain comes from the
optimized histogram layer; on that layer, the cap-aware pattern count raises
the exponent at \(\delta=0.1\) only by about \(0.001\).
Tables~\ref{tab:cap-two-improvement}, \ref{tab:fixed-cap-rates},
and~\ref{tab:log-cap-thresholds} record the fixed- and growing-cap
tradeoffs; the growing-cap codes of Section~\ref{sec:log-caps} reach
larger distances under their weaker occupation constraint.  The numerical
endpoints and the finite tables were obtained by computer
(Appendix~\ref{app:verification}).  Exact rational interval arithmetic
certifies the enclosures of \(\delta_2^\ast\) and of the cap-one endpoint \(\delta_1^\ast(2)\approx0.151\), of the
endpoint \(0.1352\) obtained with the count \(\binom{q+t-1}{q-1}\) of AAB
(Remark~\ref{rem:cap-two-numerics}) and of the endpoint \(0.084\) of the coordinatewise benchmark
with error intensity bounded by one (Remark~\ref{rem:coordinatewise-benchmark}), and the
values of \(k_2(0.1)\) and of the optimal-layer fraction \(\rho_{0.1}\) (Remark~\ref{rem:cap-two-numerics}).  The counts in
Example~\ref{ex:finite-moments} and Table~\ref{tab:finite-comparison} are
exact computations; Tables~\ref{tab:fixed-cap-rates}
and~\ref{tab:log-cap-thresholds} and the other quoted non-integer constants
are numerical evaluations.  The computational certificates are archived at
\cite{VerificationSoftware2026}.

Theorem~\ref{thm:minimum-capped-ball}, which identifies the saturated and balanced points as the centers of the smallest and largest capped deletion balls, extends the uncapped results of
\cite{KreindelEssayagZabokritskiy2026b} and gives sphere-packing bounds for separated supports; at
\((B,\alpha)=(2,1)\) they limit what any separated support gives through
the cap-aware count \(m_2\) to exponents below about \(0.620\) at \(\delta=0.1\),
positive only for \(\delta<0.2081\).  The loss-overlap enumerator
refines the pair enumerator by squared Kraus overlaps and admits a
polynomial-size recurrence for fixed cap; it is not used in the existence
bounds.  For arbitrary capped quantum subspaces, degenerate codes included,
Theorem~\ref{thm:quantum-cap-converse} combines extremal-monomial
derivative ranks with quantum packing to give an explicit converse, and
Proposition~\ref{prop:quantum-converse-comparison} gives
coherent-information bounds; neither dominates the other.  Both remain far
above the constructions, with exponents about \(1.188\) and \(1.413\)
against \(0.281\) at \((B,\alpha,\delta)=(2,1,0.1)\); this converse stays
positive for every \(\delta<1\), and only the cap-free coherent-information
bound forces zero rate, for \(\delta\geq1/2\), so at \((B,\alpha)=(2,1)\) none of these bounds excludes
positive rate for \(\delta<1/2\).

For the normalized pure-loss channel truncated to at most
\(\lfloor\delta N\rfloor\) losses, uniformly random equipartitions of a
deterministic separated support give, with high probability, codes with
exponentially vanishing worst-case error
(Corollary~\ref{cor:fixed-cap-aqec}).  For every loss
probability \(\gamma\leq1/(B+1)\), the loss-free pattern is the most probable
one at every support point, so each loss pattern has an explicitly bounded,
exponentially small probability, and the guaranteed rate region strictly
contains the exact-code region of Theorem~\ref{thm:general-fixed-cap}; at
\((B,\alpha)=(2,1)\) with \(\gamma=0.1\) it extends to \(\delta\approx0.158\), and for \(\delta>\gamma\) these codes also have vanishing worst-case error on the untruncated channel (Remark~\ref{rem:random-scope}).
For larger \(\gamma\) we bound the error intensity only by one, so the
threshold \(1/(B+1)\) is a feature of the bound.  The bound of
Corollary~\ref{cor:fixed-cap-aqec} dominates the coordinatewise random-partition benchmark adapted from
\cite[Prop.~9]{ElimelechAydinBarg2026} and
\cite[Thm.~21]{ElimelechAlbertBargAQEC2026} for every \(\gamma\) when both use
the same error-intensity bound.
At \((B,\alpha)=(2,1)\), Theorem~\ref{thm:random-fixed-cap} gives, for every \(0<\delta<1/2\), a
distribution of classical codewords of the kind requested in \cite[Sec.~6.2]{ElimelechAydinBarg2026}.
Efficient constructions of the optimized classical supports, Tverberg
partitions, exact encoders, and recoveries remain open.

Further questions concern the structure of the optimizing types in the
general \((B,\alpha,\delta)\) problem, which we evaluate only
at cap one and at \((B,\alpha)=(2,1)\), sharp converses for arbitrary capped
quantum codes,
and the joint optimization, over the choice of type and support, of the size
of a separated classical support and its affine-hull denominator
\(r_t(C)+1\leq m_B(q,t)\), which may replace \(m_B(q,t)\) in the lifting threshold.  For approximate correction,
the same-output overlaps of Remark~\ref{rem:loss-overlap-aqec-obstruction} already vanish on separated supports, and the
error-intensity bound on the same-Kraus covariance \(\Gamma\) defined there gives the
region above; variance-sensitive concentration
\cite[Remark~6]{ElimelechAlbertBargAQEC2026}, guarantees uniform over
channel families, an evaluation of the error intensity for
\(\gamma>1/(B+1)\), and sharper transfer to the full loss channel could
strengthen it further.

\begin{appendices}
\renewcommand{\theHequation}{appendix.\Alph{section}.\arabic{equation}}

\section{Computer-assisted finite verification}
\label{app:verification}

The finite enumerations, numerical tables and certified endpoint
enclosures described in this appendix are archived at
\cite{VerificationSoftware2026}.

Proposition~\ref{prop:finite-enumerator} was evaluated with exact integer
arithmetic as follows.  Let
\[
  D_j^{\mathrm{pair}}(s_x,s_y,d_{\mathrm{del}})
\]
be the number of ordered assignments to the first \(j\) coordinates with
first sum \(s_x\), second sum \(s_y\), and accumulated positive difference
\(d_{\mathrm{del}}\eqdef\sum_{i=0}^{j-1}(x_i-y_i)_+\).  This statistic equals deletion
distance at the equal-sum terminal states \(s_x=s_y=N\).
Initialize \(D_0^{\mathrm{pair}}(0,0,0)=1\), with all other states zero,
and use the
recurrence
\begin{equation}
\label{eq:dp-recurrence}
  D_{j+1}^{\mathrm{pair}}(s_x,s_y,d_{\mathrm{del}})
  =
  \sum_{a,b=0}^{B}
  D_j^{\mathrm{pair}}\bigl(s_x-a,s_y-b,d_{\mathrm{del}}-(a-b)_+\bigr),
\end{equation}
where states with a negative argument are zero.  Then
\[
  A_{N,q,B}(m)=D_q^{\mathrm{pair}}(N,N,m).
\]
States exceeding either target sum or the required distance cutoff may be
discarded.

The coefficient recurrence was checked against independent labeled
brute-force enumeration for
\[
  2\leq q\leq5,\qquad
  0\leq N\leq8,\qquad
  0\leq B\leq N,\qquad
  0\leq t\leq\min\{4,N\}.
\]
This comprises 180 parameter instances and 820 shell comparisons.  When
\(B\geq N\), every distance shell also agrees with the known full-simplex
enumerator.  Maximum induced balls are optimized over coordinate histograms
of the center and then checked by exhaustive enumeration over all labeled
centers, giving 37685 labeled-center checks.  The calculation also reproduces
all six rows of Table~\ref{tab:finite-comparison} and the dimension change
from 2321 to 2336 in the inactive-moment example.

The four-index loss-overlap recurrence
\eqref{eq:loss-overlap-dp} was checked against independent brute-force
enumeration of every common descendant for
\[
  B=1,\ 1\leq q\leq6;
  \qquad
  B=2,\ 1\leq q\leq4;
  \qquad
  B=3,\ 1\leq q\leq3,
\]
and every \(0\leq N\leq Bq\).  This comprises 72 parameter
instances, 3172 ordered pairs, 14391 pair/loss-order comparisons, and
1384 comparisons of values of \(W_{N,q,B}(m,r)\).  The positivity criterion
\eqref{eq:loss-overlap-support}, the pair-joint-type formula
\eqref{eq:loss-overlap-joint-type}, the cap-one specialization
\eqref{eq:loss-overlap-cap-one}, and
Example~\ref{ex:unequal-loss-overlap} all agreed exactly.

Classical ratios were rounded upward exactly, and quantum dimensions were
computed as in \eqref{eq:finite-K}.  The closed-form cap-two and cap-one functions and
their four endpoints were computed from the entropy and saddle equations
with 100-decimal-digit arithmetic.  Independent exact rational interval
arithmetic certifies the endpoint signs and value enclosures below: it
rounds outwards on a \(10^{-70}\) grid, encloses square roots by integer
square roots, and explicitly bounds the positive tail of the range-reduced
logarithm series.  The resulting certified enclosures have width
\(10^{-21}\); we display outward-rounded four-decimal bounds below, and the
archived certificates retain the full-precision endpoints.  The four
equations are
\begin{align*}
  R_{\mathrm{cc}}^{(2)}(\delta)-\Lambda_2(\delta)&=0,
  &R_{\mathrm{cc}}^{(2)}(\delta)
    -(1+\delta)h_2\!\left(\frac1{1+\delta}\right)&=0,\\
  R_{\mathrm{cc}}^{(2)}(\delta)-2\Lambda_2(\delta)&=0,
  &2-2h_2(\delta)-2h_2(\delta/2)&=0.
\end{align*}
They define, respectively,
\(\delta_2^\ast\), \(\delta_{\mathrm{AAB}}^\ast\),
\(\delta_{\mathrm{AQ,coord}}^\ast\), and
\(\delta_1^\ast(2)\).  In particular,
\(\delta_{\mathrm{AQ,coord}}^\ast\) refers only to the coarse
unit-intensity coordinatewise benchmark of Section~\ref{sec:random}.
Write \(u\eqdef u_\delta^{\mathrm{cc}}\) and \(v\eqdef v_\delta\).
On \(0<\delta<1/2\) these four left-hand
sides are strictly decreasing: the first three derivatives reduce to
\(2\log_2u+\log_2v\),
\(2\log_2u-\log_2((1+\delta)/\delta)\), and
\(2\log_2u+2\log_2v\), respectively, which are negative because \(0<u<1\)
and \(0<v<1\) for \(0<\delta<1/2\) by \eqref{eq:u-cc-implicit} and
\eqref{eq:v-delta}, while the cap-one derivative is the
strictly negative expression in the proof of
Corollary~\ref{cor:cap-one}.  Together with this analytic uniqueness, the
outward-rounded certified brackets below locate the corresponding roots:
\begin{align*}
0.1355
  &<\delta_2^\ast<0.1356,\\
0.1351
  &<\delta_{\mathrm{AAB}}^\ast<0.1352,\\
0.0839
  &<\delta_{\mathrm{AQ,coord}}^\ast<0.0840,\\
0.1513
  &<\delta_1^\ast(2)<0.1514.
\end{align*}
In floating point the four roots are \(0.135544\), \(0.135189\),
\(0.083944\) and \(0.151359\); the values quoted in the text
are their nearest roundings at the displayed precision.
With \(c_\kappa=0\) the channel-normalized random-equipartition condition
gives the same endpoint as Theorem~\ref{thm:exact-cap-two}; the
\(\gamma\)-dependent endpoints quoted in Corollary~\ref{cor:fixed-cap-aqec}
are floating-point evaluations.  At
\(\delta=0.1\), direct interval evaluation gives
\begin{align*}
0.3527
  &<\rho_\delta<0.3528,\\
0.2807
  &<k_2(\delta)<0.2808.
\end{align*}
Table~\ref{tab:cap-two-improvement} combines the case \(B=2\) of
Table~\ref{tab:fixed-cap-rates}, with its endpoint evaluated to four
decimals, with these values of \(\delta_2^\ast\) and \(k_2(0.1)\); the remaining values in Remark~\ref{rem:cap-two-numerics} are
evaluations of \eqref{eq:u-cc-delta}, \eqref{eq:R-cc-two} and
\eqref{eq:cap-two-published-aab} at \(\delta=0.1\).
The 15 endpoints in Table~\ref{tab:log-cap-thresholds} and its three values
at \((\alpha,\delta)=(1,0.15)\) were computed from
\cite[arXiv version, Prop.~9 and Thms.~10, 16]{GoyalDaoKovacevicKiah2025} with ordinary
distance \(2\delta\), the maximum-ball formula of
\cite{ElimelechAydinBarg2026}, and subtraction of the exponent
\(\Lambda_\infty(\alpha,\delta)\) of the published AAB denominator.  All 18 values round to the displayed
three decimal places; they are high-precision (100-digit) numerical
evaluations, not interval certificates.
Table~\ref{tab:fixed-cap-rates} was evaluated with ordinary
floating-point arithmetic and checked to round to the displayed three
decimal places; it is a numerical benchmark, separate from the rational
interval certificates above.
The coefficient identities and recurrences, and the existence statements of
the asymptotic theorems, are proved analytically and do not depend on these
computations.  The numerical values of \(\delta_2^\ast\),
\(\delta_{\mathrm{AAB}}^\ast\), \(\delta_{\mathrm{AQ,coord}}^\ast\),
\(\delta_1^\ast(2)\), \(\rho_{0.1}\) and \(k_2(0.1)\) quoted in
Theorem~\ref{thm:exact-cap-two}, Remark~\ref{rem:cap-two-numerics},
Table~\ref{tab:cap-two-improvement}, Corollaries~\ref{cor:cap-one}
and~\ref{cor:fixed-cap-aqec}, and Section~\ref{sec:random} rest on the
certified enclosures above; Example~\ref{ex:finite-moments} and
Table~\ref{tab:finite-comparison} rest on the exact computations described
above.  The remaining quoted values, namely those in
Remark~\ref{rem:dual-rail}, the values \(0.178\) and \(0.227\) in
Remark~\ref{rem:coordinatewise-benchmark}, the asymptotic exponents and the
endpoint \(0.2081\) in Section~\ref{sec:converses}, and the evaluations in
Remark~\ref{rem:cap-two-numerics}, are floating-point evaluations of
closed-form expressions, not interval certificates; the integer bounds in
Section~\ref{sec:quantum-converse} are exact.

\section{Majorization and extremal balls\texorpdfstring{\\}{} in a capped simplex}
\label{app:minimum-ball}

We prove Theorem~\ref{thm:minimum-capped-ball}.  By the reduction from
majorization to unit smoothings below, it is enough to compare two centers that differ in two
coordinates.  Suppose the old center has
entries \((a,b)\) with \(b-a\geq2\), and the new center has
\((a+1,b-1)\).  Fix the sum \(s\) of the corresponding two coordinates
of a target point.  Write that target pair as \((j,s-j)\) with
\(L_s\leq j\leq U_s\), where
\[
  L_s\eqdef\max\{0,s-B\},
  \qquad
  U_s\eqdef\min\{B,s\}.
\]
On this fiber, the set of target points with the same other coordinates and
the same pair sum \(s\), apply the cyclic permutation
\[
  \phi_s(j)
  \eqdef
  \begin{cases}
    j+1,&j<U_s,\\
    L_s,&j=U_s.
  \end{cases}
\]
Together over all fibers, this defines a bijection of
\(\cS_{q,N}^{(B)}\) to itself.

Call the \(\ell_1\) contribution of these two coordinates the local
distance.  When \(j<U_s\), the local distance from the new center to the
image is exactly the local distance from the old center to the original
target.
At the wrap point \(j=U_s\), define
\[
  f(j)\eqdef|j-a|+|s-j-b|.
\]
The old distance is \(f(U_s)\), while the new one is \(f(L_s-1)\).
The convex function \(f\) is symmetric about
\[
  \mu\eqdef\frac{a+s-b}{2}.
\]
Since \(L_s+U_s=s\) and \(b-a\geq2\),
\[
  \frac{L_s-1+U_s}{2}
  =
  \frac{s-1}{2}
  \geq
  \mu+\frac12.
\]
Thus \(L_s-1\) is no farther from the symmetry center than \(U_s\), and
\(f(L_s-1)\leq f(U_s)\).  The bijection therefore maps every ball about
the old center injectively into the ball of the same radius about the
new center.

Suppose that \(\underline{n}\) majorizes \(\underline{m}\), that both are
sorted in decreasing order, and that they differ.  Let \(i\) be the first
index with \(n_i\neq m_i\), so that \(n_i>m_i\), and let \(j>i\) be the first
index with \(n_j<m_j\).  Then \(n_i\geq m_i+1\geq m_j+1\geq n_j+2\), and
moving one unit from coordinate \(i\) to coordinate \(j\) keeps the cap,
gives a vector that still majorizes \(\underline{m}\) (its sorted partial
sums dominate its unsorted ones), and lowers the
\(\ell_1\) distance to \(\underline{m}\) by two.  Hence a sequence of these
unit Robin--Hood smoothings, interspersed with coordinate permutations,
transforms \(\underline{n}\) into \(\underline{m}\); this is the integer form
of the Hardy--Littlewood--P\'olya criterion \cite{MarshallOlkinArnold2011}.  Iterating the injection
proves \eqref{eq:ball-majorization}.  The saturated point
\(\underline{n}^{\mathrm{sat}}\) in \eqref{eq:extreme-center} majorizes every
capped vector of sum \(N\), while the balanced point
\(\underline{n}^{\mathrm{bal}}\) is
majorized by every such vector.  They therefore minimize and maximize,
respectively, the ball volume at every radius.

It remains to count the balls.  At a center coordinate of value \(j\), a
target coordinate may lose \(0,\ldots,j\) units or gain
\(0,\ldots,B-j\) units, contributing \(b_j(U,V)\).  Because the total
excitation is fixed, total loss and total gain are equal, say to \(m\);
by \eqref{eq:positive-part-distance}, \(m\) is the deletion distance.
Multiplying over coordinates, extracting the diagonal coefficient, and
summing \(m\leq r\) proves \eqref{eq:center-ball-volume}.  Substitution
of \(\underline{n}^{\mathrm{sat}}\) and \(\underline{n}^{\mathrm{bal}}\) gives
\eqref{eq:minimum-ball-volume} and \eqref{eq:maximum-ball-volume}.
Disjoint balls of radius
\(\lfloor(D-1)/2\rfloor\) around codewords then give
\eqref{eq:capped-sphere-packing}.

\end{appendices}

\phantomsection

\end{document}